\documentclass{article}

\usepackage{epsfig,xspace}
\usepackage{latexsym}
\usepackage[matrix,arrow,curve]{xy}\xyoption{all}

\usepackage[usenames,dvipsnames,svgnames,table]{xcolor}
\usepackage{amssymb,amsmath,latexsym,amsfonts,amsthm,alltt}
\usepackage[usenames,dvipsnames,svgnames,table]{xcolor}
\usepackage[dvips,usenames]{stmaryrd}
\usepackage{enumitem}
\usepackage{url}
\usepackage{graphicx}
\usepackage{float}
\usepackage{amstext}
\usepackage{cite}
\usepackage{hyperref}
\usepackage{MnSymbol}
\usepackage{accents}
\usepackage{multirow}%
\usepackage{mathrsfs}%
\usepackage[title]{appendix}%
\usepackage{xcolor}%
\usepackage{textcomp}%
\usepackage{listings}%

\usepackage{epsfig,xspace}
\usepackage[dvips,usenames]{stmaryrd}
\usepackage{amstext}
\usepackage{rotating} 
\usepackage{cleveref}
\usepackage{pdflscape}
\usepackage[a4paper, total={7in, 8.03in}]{geometry}

\numberwithin{equation}{section}

\numberwithin{equation}{section}
\theoremstyle{plain}
\newtheorem{Theorem}{Theorem}[section]
\newtheorem{Lemma}{Lemma}
\newtheorem{Proposition}{Proposition}
\newtheorem{Corollary}{Corollary}

\theoremstyle{definition}
\newtheorem{Definition}{Definition}
\newtheorem{Remark}{Remark}

\newcommand\crule[3][black]{\textcolor{#1}{\rule{#2}{#3}}}

\newcommand{\ra}{\rightarrow}
\newcommand{\lra}{\longrightarrow}

\newcommand{\la}{\leftarrow}

\newcommand{\Ra}{\Rightarrow}
\newcommand{\La}{\Leftarrow}

\newcommand{\Lra}{\Longrightarrow}

\newcommand{\midsp}{\;|\;}
\newcommand{\sub}[2]{#1_{{}_{#2}}}
\newcommand{\telos}{\hfill$\Box$}
\newcommand{\type}[1]{{\tt #1}}

\newcommand{\iso}{\backsimeq}

\newcommand{\val}[1]{\mbox{$[\![#1]\!]$}}
\newcommand{\forces}{\Vdash}
\newcommand{\dforces}{\forces^{\!\!\partial}}
\newcommand{\yvval}[1]{\mbox{$(\!|#1|\!) $}}

\newcommand{\infrule}[2]{\frac{\mbox{\rm $#1$}}{\mbox{\rm $#2$}}}
\newcommand{\proves}{\vdash}
\newcommand{\vproves}{\mbox{$\medvert\!\!\!\!\sim\;$}}
\newcommand{\vmodels}{\mbox{$\medvert\!\!\!\!\approx\;$}}
\newcommand{\zmodels}{\mbox{$\medvert\!\!\!\!\eqsim\;$}}

\newcommand{\upv}{\upVdash}
\newcommand{\rperp}{\mbox{${}^{\upv}$}}
\newcommand{\gphi}{{\mathcal  G}(W_\partial)}
\newcommand{\gpsi}{{\mathcal  G}(W_1)}

\newcommand{\lperp}{{}\rperp}

\newcommand{\lbdiamond}{\raisebox{-2pt}{\mbox{\Huge {$\filleddiamond$}}}}

\newcommand{\lbbox}{\raisebox{-1.2pt}[0pt][0pt]{\crule[black]{0.27cm}{0.27cm}}\hspace*{1pt}}

\newcommand{\stx}[2]{\mbox{ST$_{#1}(#2)$}}
\newcommand{\sty}[2]{\mbox{ST$_{#1}(#2)$}}
\newcommand{\filt}{\mbox{\rm Filt}}

\newcommand{\idl}{\mbox{\rm Idl}}

\newcommand{\lfspoon}{\leftfilledspoon}
\newcommand{\rfspoon}{\rightfilledspoon}

\newcommand{\tright}{\triangleright}
\newcommand{\tleft}{\triangleleft}

\newcommand{\Mtright}{\mathrel{\mbox{$|\!\!\largetriangleright$}}}
\newcommand{\Mtleft}{\mathrel{\mbox{$\largetriangleleft\!\!|$}}}

\newcommand{\mcite}[2]{\cite{#2}(#1)}

\newcommand{\ltright}{\largetriangleright\ }
\newcommand{\ltleft}{\largetriangleleft\ }
\newcommand{\lspoon}{\leftspoon}
\newcommand{\rspoon}{\rightspoon}

\title{{\bf A General (Uniform) Relational Semantics for Sentential Logics}}
\author{Chrysafis  Hartonas\\
Department of Digital Systems, University of Thessaly, Larissa, Greece\\
hartonas@uth.gr}

\begin{document}
\maketitle

\begin{abstract}
We present a  general relational semantics framework which, by varying the axiomatization and components of the relational structures, provides a uniform semantics for sentential logics, classical and non-classical alike.  The approach we take rests on a generalization of the J\'{o}nsson-Tarski representation (and duality) for Boolean algebras with operators to the cases of posets, semilattices, or bounded lattices (with, or without distribution) with quasi-operators. Completeness proofs rely on a choice-free construction of canonical extensions for the algebras in the quasivarieties of the equivalent algebraic semantics of the logics. Correspondence results for axiomatic extensions of the logics of implication that we study rely on a fully abstract translation into their modal companions and they are calculated using a generalized Sahlqvist - van Benthem algorithm.
\end{abstract}

\section{Introduction}
\label{intro}
Comparing to algebraic semantics \mcite{1996}{font-jansana-book}, relational semantics appears to be fragmented and ad hoc, even at times where the aim is to devise an all-encompassing semantic framework, such as in Dunn's project of generalized Galois logics  \mcite{1990,1993,2008}{dunn-ggl,dunn-partial,gglbook}. Indeed, in a joint article of Dunn, Gehrke and Palmigiano \mcite{2005}{dunn-gehrke} the point is made that in his article \cite{dunn-partial}   Dunn
``{\em attempted to provide a uniform approach to relational semantics for various implicative substructural logics}'', but  ``{\em while he achieved a uniform approach in the sense that the relational semantics obtained all arose through concrete representation of the algebras, he had to change his method of representation in ad hoc ways to fit the various logics.}''

As an alternative, an approach to relational semantics via canonical extensions is advocated in \cite{dunn-gehrke}. It should be mentioned that at the time that Dunn was developing his project, announced in \mcite{1990}{dunn-ggl}, the theory of canonical extensions was not yet in place, except for the case of Boolean algebras with operators, where they were called ``perfect extensions'' by J\'{o}nsson and Tarski \mcite{1951-52}{jt1,jt2}. It was only about a decade later that the theory was generalized to the case of general lattices by Gehrke and Harding \mcite{2001}{mai-harding} and to the case of distributive lattices by Gehrke and J\'{o}nsson \mcite{2004}{mai-jons}. The joint article \mcite{2005}{dunn-gehrke} generalized further the theory of canonical extensions to the case of partial orders with maps of given monotonicity properties. The missing case of canonical extensions for semilattices was addressed by Gouveia and Priestley in \mcite{2014}{hilary-sem}.

Gehrke's approach to relational semantics \mcite{2005 and 2006}{dunn-gehrke,mai-gen} is to uniformly use canonical extensions and the associated discrete duality between perfect lattices (with additional operations) and separated and reduced (RS) frames (with relations) to obtain complete semantics. 

A perfect lattice is a complete lattice $\mathbf{L}$ that is join-generated by its set $J^\infty(\mathbf{L})$ of completely join-irreducible elements and meet-generated by its set $M^\infty(\mathbf{L})$ of completely meet-irreducible elements. Given a perfect lattice, its dual frame $(J^\infty(\mathbf{L}),\leq~, M^\infty(\mathbf{L}))$ is an RS-frame, where $\leq$ is the lattice order. Conversely, given an RS-frame $(X,Y,R)$ its dual perfect lattice is the lattice $\mathcal{G}(X)$  of Galois stable subsets of $X$. Stability refers to the Galois connection generated by the binary relation $R\subseteq X\times Y$ and the closure operator formed by composition of the Galois maps. Galois stable sets are precisely the fixpoints of this closure operator. The discrete duality of perfect lattices and RS-frames extends to the case of perfect lattice-ordered algebras and RS-frames with relations. For example, given a lattice with an implication operation $(\mathbf{L},\ra)$, a canonical extension $((\alpha,\mathbf{A}),\ra^\pi)$ is first constructed, where $\alpha:\mathbf{L}\hookrightarrow\mathbf{A}$ is the canonical embedding and $\ra^\pi$ is defined as the $\pi$-extension of $\ra$, as specified in \cite{mai-harding}. The dual RS-frame $\mathfrak{G}=(J^\infty(\mathbf{A} ),\leq,M^\infty(\mathbf{A} ),R_\ra)$ is subsequently defined, where $\leq$ is the order in $\mathbf{A} $ and $R_\ra\subseteq J^\infty(\mathbf{A} )\times J^\infty(\mathbf{A} )\times M^\infty(\mathbf{A} )$ is defined from $\ra^\pi $ by setting $R_\ra(x,y,m)$ iff $y\leq x\ra^\pi  m$. The logic is complete if its associated (quasi)variety of equivalent algebraic semantics is closed under canonical extensions.

To show that the lattice $\mathcal{G}(X)$ is perfect in the above discrete duality argument, an appeal to the axiom of choice becomes necessary (see \mcite{Lemma~3.4}{mai-harding}, or \mcite{Remark~2.7 and Theorem~2.8}{dunn-gehrke}) and it is then worth devising an approach that avoids this. We do that in this article. The issue of constructive semantics even for classical logic has become one of current interest, triggered by a recent article of Holliday and Bezhanishvili \mcite{2020}{choice-free-BA} on choice-free duality for Boolean algebras, followed by several publications \cite{choice-free-dmitrieva-bezanishvili,choice-free-Ortho,choice-free-deVries,choice-free-inquisitive-DBLP:conf/wollic/BezhanishviliGH19, choiceFreeHA,choiceFreeStLog,choice-free-Tarski-Celani2025} in this research area.

Canonical extensions are constructed by a representation argument, in a tradition of topological representations that reaches back to the Stone representation of Boolean algebras \mcite{1938}{stone1}, the Stone \mcite{1937}{stone2} and Priestley \mcite{1970}{hilary} representations of distributive lattices, Esakia's representation and duality for Heyting algebras \mcite{1974}{esakia2}, the lattice representations of Urquhart \mcite{1978}{urq}, Hartung \mcite{1992}{hartung},  Allwein and Hartonas \mcite{1993}{iulg-bounded}, Hartonas and Dunn \mcite{1993}{iulg}, Plo\v{s}\v{c}ica \mcite{1995}{plo}, Hartonas \mcite{1997}{dloa}, Hartonas and Dunn \mcite{1997}{sdl},  Moshier and Jipsen \mcite{2014}{Moshier2014a,Moshier2014b}, Gehrke and Van Gool \mcite{2014}{lattice-duality-Gehrke2014}, Celani and Gonzales \mcite{2020}{lattice-duality-Celani2020} and Bezhanishvili et al \mcite{2024}{choice-free-dmitrieva-bezanishvili}.

Representation arguments mediate completeness proofs in logic, but invariably logics involve additional operators. 

For extensions of the classical propositional calculus the seminal work of J\'{o}nsson and Tarski \mcite{1951-52}{jt1,jt2} for Boolean algebras with operators provides the technical background of relational semantics for expansions of classical propositional logic, for example with modal operators (though historically the introduction of relational semantics originates with Kripke's work \mcite{1959 and 1963}{kripke1,kripke2} and for intuitionistic logic \mcite{1965}{Kripke-int}, followed by the Routley and Meyer semantics framework \mcite{1973}{routley1973semantics} for Relevance logic and by Goldblatt's orthoframe semantics \mcite{1974}{goldb} for Orthologic). 

The relational representation of distributive lattices with additional operations of various types, extending the J\'{o}nsson-Tarski framework, was studied by Urquhart \mcite{1979 and 1996}{urq1979,Urquhart1996}, by Sofronie-Stokkermans \mcite{2000}{Sofronie-Stokkermans00,Sofronie-Stokkermans00a}, and it has been the express objective of Dunn's gaggle theory project \mcite{1990}{dunn-ggl}. 

In \mcite{2018,2022,2023}{sdl-exp,kata2z,duality2} we extended the J\'{o}nsson-Tarski representation (and duality) framework to the case of bounded lattices with arbitrary quasi-operators (normal lattice operators, in our preferred terminology). Choice-free versions of the duality \mcite{2023}{duality2} were presented for some cases of interest \mcite{2023,2025}{choiceFreeHA,choiceFreeStLog} and applications to the semantics of some  calculi of interest were given in \mcite{2025}{choiceFreeStLog,dfnmlA}.

In this article, we extend the generalized J\'{o}nsson-Tarski framework of relational representation to the cases of logics whose induced partial order in their Lindenbaum-Tarski algebra is weaker than a lattice order, perhaps a semilattice, or even just a mere partial order. We also briefly review the case of lattice-based logics and discuss the particular cases of distributive logics, or of logics on a classical underlying propositional calculus. 

A uniform relational semantics is used as a convenient tool for addressing the correspondence problem for various logics in a uniform way. 
In all cases, the representation argument  ``calculates'' a modal companion for the logic of interest and a  GMT-like \cite{goedel} translation and dual translation (co-translation) can be given, proven to be fully abstract in \mcite{2025}{vb}. This identifies the logic of interest as a fragment of its sorted  companion modal logic. The translation reflects the projection of sorted image operators in the dual sorted powerset algebra of a frame, to normal operators in the complete lattice of Galois stable sets (the frame's full complex algebra). This allows for transferring techniques and results from the classical though sorted companion modal logic to the calculus of interest. 

We exploit in this article the generalized Sahlqvist - van Benthem correspondence algorithm presented in \mcite{2025}{dfmlC} for distribution-free modal logics, as it is also uniformly applicable to the cases of logics on a semilattice, or on a mere partial-order basis. 

Section~\ref{prelims section} presents a brief review of canonical extensions and of the background developed by the author in previous publications on the duality of normal lattice expansions and sorted residuated frames with relations and it can be safely skipped by the reader familiar with canonical extensions and with this author's recent previous work.

The remaining sections detail the framework by studying the logics of implicative posets (\cref{posets section}), the Lambek calculus (\cref{lambek section}) with, or without association, the logic of implicative semilattices (\cref{semilattices section}) and the logic of implicative lattices (\cref{lattices section}). \cref{substructural section} discusses the question of relational semantics for the full Lambek calculus (associative, or not) and of substructural logics. We conclude with \cref{classical section} where the application of the framework to distributive, intuitionistic and classical logic is addressed.

\section{Implicative Algebras, Frames, Logics, and Canonical Extensions}
\label{prelims section}
Abstract definitions of implicative algebras have been given by Rasiowa \mcite{1977}{rasiowa1977} and  by Cintula \mcite{2006}{cintula06}. Rasiowa's {\em implicative} (or {\em implication}) {\em algebras} are structures $\mathbf{A}=(A,\ra, \bigvee)$ satisfying four axioms from which it follows that the relation $\leq$ defined by $a\leq b$ iff $a\ra b=\bigvee$ turns $\mathbf{A}$ to a bounded above preorder with $\bigvee$ as its greatest element. Cintula, working more in the context of modern abstract algebraic logic, generalizes the concept to that of an {\em implicative matrix} $\mathbf{A}=(A,\ra,D)$, with $D\subseteq A$ a {\em subset of designated elements}, inducing an order on $\mathbf{A}$ by setting $a\leq b$ iff $a\ra b\in D$. In constructing the Lindenbaum-Tarski matrix of what is called in \cite{cintula06} a weakly implicative logic, the set of designated elements is simply the singleton $\{[\top]\}$ containing the top element (the equivalence class of the constant $\top$ (true)). The Lindenbaum-Tarski implicative matrix of a substructural logic, including a truth constant $\type{t}$, will then accordingly be a matrix $\mathbf{A}=(A,\ra,D)$ with $D$ being the principal filter $D=\{[\varphi]\midsp [\type{t}]\leq [\varphi]\}$.

In either case, the base definition does not imply any of the familiar properties of implication, such as being antitone in the first and monotone in the second argument place. It is only with a strengthened notion of {\em positive implication algebra} in Rasiowa's  \mcite{Chapter~II, Section~2}{rasiowa1977}, satisfying at least two additional axioms (p1),(p2), that the monotonicity properties can be deduced from the definition, see \mcite{2.3, page~22}{rasiowa1977}. But (p1) already introduces the weakening property $a\ra(b\ra a)=\bigvee$, both (p1) and (p2) are needed to conclude that the induced preorder is a partial order and, moreover, exchange $a\ra(b\ra c)=b\ra(a\ra c)$ can be proven in a positive implication algebra, as demonstrated in \mcite{2.3, page~22}{rasiowa1977}.

The typical definition of an algebra (a $\Sigma$-algebra, more accurately) in the context of universal algebra is as a pair $\mathbf{A}=(A,\Sigma)$ where $\Sigma$, the signature (sometimes referred to as the language) of the algebra, is a pair $(\{f_j\}_{j\in J},\alpha)$ and where the $f_j$ are operator names (symbols) corresponding to operators $f^\mathbf{A}_j$ on the algebra and $\alpha:J\lra\mathbb{N}$ is the arity map, fixing the number of argument places. If $\mathbf{A}$ is an ordered algebra, whose order is perhaps a (semi)lattice order, then monotonicity and (co)distribution properties are relegated to the axiomatization of the algebra. It facilitates specifying the correspondence of algebraic and relational semantics for a logic to mildly generalize and define a suitable notion of a qualified ordered algebra as a tuple $\mathbf{A}=(A,\leq,(\{f_j\}_{j\in J},\delta))$ where $\delta:J\lra\{1,\partial\}^{n(j)+1}$ assigns to $f_j$ an arity but also a monotonicity, or distribution type. We explain in the sequel.

\subsection{Logics of Ordered Algebras}
Let $\mathbf{L}$ be a (bounded) lattice, $\{1,\partial\}$ be a 2-element set, $\mathbf{L}^1=\mathbf{L}$ and $\mathbf{L}^\partial=\mathbf{L}^{\mathrm{op}}$ (the opposite lattice). Extending established terminology \cite{jt1}, a function $f:\mathbf{L}_1\times\cdots\times\mathbf{L}_n\lra\mathbf{L}_{n+1}$ will be called {\em additive} and {\em normal}, or a {\em normal operator}, if it distributes over finite joins of each lattice $\mathbf{L}_i$, for $i=1,\ldots n$, delivering a join in the lattice $\mathbf{L}_{n+1}$.

An $n$-ary operation $f$ on a bounded lattice $\mathbf{L}$ is {\em a normal lattice operator of distribution type  $\delta(f)=(i_1,\ldots,i_n;i_{n+1})\in\{1,\partial\}^{n+1}$}  if it is a normal additive function  $f:\mathbf{L}^{i_1}\times\cdots\times\mathbf{L}^{i_n}\lra\mathbf{L}^{i_{n+1}}$ (distributing over finite joins in each argument place), where  each $i_j$, for  $j=1,\ldots,n+1$,   is in the set $\{1,\partial\}$, hence $\mathbf{L}^{i_j}$ is either $\mathbf{L}$, or $\mathbf{L}^\partial$. For example, each of $\ra,\Box,\Diamond$ is a normal lattice operator of respective distribution type $\delta(\ra)=(1,\partial;\partial)$, $\delta(\Diamond)=(1;1)$, $\delta(\Box)=(\partial;\partial)$.

If $\tau$ is a tuple  of distribution types, a {\em normal lattice expansion of (similarity) type $\tau$} is a lattice with a normal lattice operator of distribution type $\delta$ for each $\delta$ in $\tau$.

A lattice-ordered algebra is a structure $\mathbf{A}=(A,\leq,\wedge,\vee,0,1,(\{f_j\}_{j\in J},\delta))$ where $\delta(j)\in\{1,\partial\}^{n(j)+1}$ is the distribution type of the normal lattice operator $f_j$. If the ordering is a mere partial-order, then we refer to $\delta(j)$ as the monotonicity type (order type) of $f_j$. For the semi-lattice case a mixed notion of monotonicity and distribution properties is captured by $\delta(j)$, but we shall not insist on that, restricting to the case of interest (implication) in the present article. For implication, $\delta(\ra)=(1,\partial;\partial)$. If the underlying poset $\mathbf{A}$ is a lattice, we understand $\delta(\ra)$ as the distribution type of $\ra$ (specifying that $\ra$ takes joins in the first argument place and meets in the second argument place to meets, in both cases). If the algebra is merely a partially ordered algebra, then we understand $\delta(\ra)$ as specifying that $\ra$ is antitone in the first and monotone in the second position. In the case of a  $\wedge$-semilattice-ordered algebra, $\delta(\ra)$ may be understood as specifying that $\ra$ is antitone in the first argument place and that it distributes over meets in the second place. In other words, an elementary part of the axiomatization is already dictated by the order/distribution type. 

A class of ordered algebras of type $\tau$ will be set in correspondence with a class of frames of type $\tau$, to be described in \cref{prelim frames}.

The propositional language $\mathcal{L}=\mathcal{L}(\tau)$ of normal lattice expansions of similarity type $\tau=(\delta_j)_{j\in J}$, for some countable index set $J$, is defined by the grammar
\begin{eqnarray}
\mathcal{L}\ni\varphi &=& p_i(i\in \mathbb{N})\midsp\top\midsp\bot\midsp \varphi\wedge\varphi\midsp\varphi\vee\varphi\midsp (f_j(\varphi_1,\ldots,\varphi_{n(j)}))_{j\in J}.\label{language}
\end{eqnarray}
 Where $\delta_j=(j_1,\ldots,j_{n(j)};j_{n(j)+1})$ we let $\delta_j(k)=j_k$ for $1\leq k\leq j_{n(j)+1}$. Of particular interest is the case where an implication connective $\ra$ and a truth constant $\type{t}$ (a 0-ary operator) are included in the logic, of distribution type $\delta(\ra)=(1,\partial;\partial)$ and $\delta(\type{t})=(\ ;1)$.
 
The proof system for the minimal logic $\mathbf{\Lambda}_\tau$ (for any similarity type $\tau$), assuming also constants (0-ary operators) $\top$ and $\bot$,  is presented in Table~\ref{minimal proof system}. To simplify, we use vectorial notation $\vec{\varphi}$ for a tuple of sentences, we let $\vec{\varphi}[\;]_k$ be the vector with an empty spot at the $k$-th argument place and $\vec{\varphi}[\psi]_k$ either to display the sentence at the $k$-position, or to designate the result of filling the empty place in $\vec{\varphi}[\;]_k$, or the result of replacing the entry $\varphi_k$ with $\psi$ at the $k$-th argument place.

\begin{table}[!htbp]
\caption{Proof System for the Minimal Logic $\mathbf{\Lambda}_\tau$ of Normal Lattice Expansions of type $\tau$}
\label{minimal proof system}
\hrule
\mbox{}\\[1mm]
\begin{tabular}{lllllllll}
\multicolumn{2}{l}{Ground axioms/rules} &
 $\varphi\proves\varphi$ \hskip8mm $\bot\proves\varphi$  &  \hskip6mm$\varphi\proves\top$  &
     $\infrule{\varphi\proves\psi\hskip5mm\psi\proves\vartheta}{\varphi\proves\vartheta}$(Cut)  
   \\[3mm]
 \multicolumn{3}{l}{Substitution rule} && $\infrule{\varphi(p)\proves\psi(p)}{\varphi[\vartheta/p]\proves\psi[\vartheta/p]}$
 \\[3mm]
 \multicolumn{2}{l}{Meet semilattice rules} &
 $\infrule{\varphi\proves\vartheta}{\varphi\wedge\psi\proves\vartheta}$  & $\infrule{\psi\proves\vartheta}{\varphi\wedge\psi\proves\vartheta}$ &
 $\infrule{\varphi\proves\psi\hskip5mm\varphi\proves\vartheta}{\varphi\proves\psi\wedge\vartheta}$
 \\[3mm]
 \multicolumn{2}{l}{Join semilattice rules} &
  $\infrule{\vartheta\proves\varphi}{\vartheta\proves\varphi\vee\psi}$ & $\infrule{\vartheta\proves\psi}{\vartheta\proves\varphi\vee\psi}$ & $\infrule{\varphi\proves\vartheta\hskip5mm\psi\proves\vartheta}{\varphi\vee\psi\proves\vartheta}$
  \\[5mm]
  \multicolumn{5}{l}{
 For each $j\in J$ and where $\delta_j=(j_1,\ldots,j_k,\ldots, j_{n(j)};j_{n(j)+1})$ is the distribution type of $f_j$}\\[2mm]
 \multicolumn{4}{l}{$\bullet$ Monotonicity/Antitonicity rules (for poset-based logics)}\\[2mm]
   \multicolumn{3}{l}{
  $\infrule{\psi\proves\vartheta}{f_j(\vec{\varphi}[\psi]_k)\proves f_j(\vec{\varphi}[\vartheta]_k)}$ if $j_k=j_{n(j)+1}$
  }
  &
  \multicolumn{3}{l}{
  $\infrule{\psi\proves\vartheta}{f_j(\vec{\varphi}[\vartheta]_k)\proves f_j(\vec{\varphi}[\psi]_k)}$ if $j_k\neq j_{n(j)+1}$
  }\\[4mm]
\multicolumn{4}{l}{$\bullet$  Distribution axioms (for semilattice-based logics)}\\[2mm]
 \multicolumn{4}{l}{
 $f_j(\vec{\varphi}[\psi\vee\chi]_k)\proves f_j(\vec{\varphi}[\psi]_k)\vee f_j(\vec{\varphi}[\chi]_k)$\hskip1cm if $j_k=1=j_{n(j)+1}$
 }\\[3mm]
  \multicolumn{4}{l}{
 $f_j(\vec{\varphi}[\psi]_k)\wedge f_j(\vec{\varphi}[\chi]_k)\proves f_j(\vec{\varphi}[\psi\wedge\chi]_k)$ \hskip1cm if $j_k=\partial=j_{n(j)+1}$
 }\\[4mm]
 \multicolumn{4}{l}{$\bullet$  Co-Distribution axioms (for lattice-based logics)}\\[2mm]
 \multicolumn{4}{l}{
$f_j(\vec{\varphi}[\psi\vee\chi]_k)\proves f_j(\vec{\varphi}[\psi]_k)\wedge f_j(\vec{\varphi}[\chi]_k)$\hskip1cm   if $j_k=1\neq\partial=j_{n(j)+1}$}\\[3mm]
 \multicolumn{4}{l}{
 $f_j(\vec{\varphi}[\psi]_k)\vee f_j(\vec{\varphi}[\chi]_k)\proves  f_j(\vec{\varphi}[\psi\wedge\chi]_k)$
  \hskip1cm if $j_k=\partial\neq 1=j_{n(j)+1}$}
  \\[5mm]
\multicolumn{3}{l}{$\bullet$  Normality axioms}\\[2mm] 
 \multicolumn{4}{l}{
$f_j(\vec{\varphi}[\bot]_k)\proves\bot$ \hskip1cm if $j_k=1=j_{n(j)+1}$}\\[3mm]
 \multicolumn{4}{l}{
$\top\proves f_j(\vec{\varphi}[\top]_k)$ \hskip1cm if $j_k=\partial=j_{n(j)+1}$}\\[3mm]
 \multicolumn{4}{l}{
 $\top\proves f_j(\vec{\varphi}[\bot]_k)$ \hskip1cm if  $j_k=1\neq\partial=j_{n(j)+1}$}\\[3mm]
 \multicolumn{4}{l}{
 $f_j(\vec{\varphi}[\top]_k)\proves\bot$      \hskip1cm   if $j_k=\partial\neq 1=j_{n(j)+1}$}
 \\[5mm]
 \multicolumn{5}{l}{$\bullet$ Implication and truth constant (when included in the language of the logic)}\\[2mm]
\hskip6mm $\proves\type{t}$ && $\infrule{\varphi\proves\psi}{\overline{\type{t}\proves\varphi\rfspoon\psi}}$
\end{tabular}
\mbox{}\\[1mm]
\hrule
\end{table}

Writing $\varphi\equiv\psi$ to mean that both $\varphi\proves\psi$ and $\psi\proves\varphi$ are provable sequents and $[\varphi]$ for the equivalence class of $\varphi$, the axioms and rules of the minimal logic $\mathbf{\Lambda}_\tau$ ensure that $\equiv$ is a congruence and that the Lindenbaum-Tarski algebra of the logic is a normal lattice expansion of type $\tau=(\delta_j)_{j\in J}$.

\subsection{Sorted Residuated Frames with Relations}
\label{prelim frames}
\begin{Definition}
\label{frame}
By a {\em sorted residuated frame} (simply a {\em frame}, in the sequel) we mean a structure $\mathfrak{F}=(s,W,I,(R_j)_{j\in J},\sigma)$, where
\begin{itemize}
\item $s$ is a finite list of sorts,
\item $W=(W_t)_{t\in s}$ is a nonempty sorted set (none of the sorts $W_t$ is allowed to be empty),
\item $I\subseteq\prod_{t\in s}W_t$ is a distinguished sorted relation,
\item $J$ is a countable index set,
\item $\sigma$ is a sorting map on $J$ with  $\sigma(j)\in s^{n(j)+1}$ and
\item $(R_j)_{j\in J}$ is a family of sorted relations such that if $\sigma(j)=(i_{n(j)+1}:i_1\cdots i_{n(j)})\in s^{n(j)+1}$, then $R_j\subseteq W_{i_{n(j)+1}}\times\prod_{j=1}^{n(j)}W_{i_j}$. 
\end{itemize}
\end{Definition} 
We make no assumption of disjointness of sorts. For the purposes of this article $s$ will be fixed to be a list of two sorts, $s=\{1,\partial\}$.

For mnemonic reasons, we often display the sort of a relation as a superscript, as in $\mathfrak{F}=(s,W,I,(R^{\sigma(j)}_j)_{j\in J})$. For example, $R^{11}, T^{\partial 1\partial}$ designate sorted relations $R\subseteq W_1\times W_1$ and $T\subseteq W_\partial\times (W_1\times W_\partial)$. 

\begin{Definition}
\label{corresponding frames}
  To every similarity type $\tau=(\delta_j)_{j\in J}$ we associate a class $\mathbb{F}_\tau$ of frames (the $\tau$-frames) $\mathfrak{F}=(s,W,I,(R_j)_{j\in J},\sigma)$, where if $\delta(j)=(i_1,\ldots,i_{n(j)};i_{n(j)+1})$, then $\sigma(j)=(i_{n(j)+1}:i_1\cdots i_{n(j)})$.
\end{Definition}
Thus to the language of Positive Lattice Logic (the logic of bounded lattices) we associate frames $\mathfrak{F}=(s,W,I)$. To the language of the logic of implicative posets $\mathbf{P}=(P,\leq,1,\ra)$, or implicative lattices $\mathbf{L}=(L,\leq,\wedge,\vee,0,1,\ra)$ and given that $\delta(\ra)=(1,\partial;\partial)$ we associate frames $\mathfrak{F}=(s,W,I,T,\sigma)$, where $\sigma(T)=(\partial:1\partial)$, i.e. $T\subseteq W_\partial\times(W_1\times W_\partial)$.

\begin{Remark}[Kripke, Routley-Meyer and Goldblatt Frames]
\label{special case frames}
Structures $\mathfrak{F}=(s,W,I,(R_j)_{j\in J},\sigma)$, as described above, generalize Kripke frames for classical, intuitionistic, or relevance logic. Indeed, classical Kripke frames constitute the subclass of frames arising by letting  $W_1=W_\partial$ and where $I\subseteq W_1\times W_\partial$ is the identity relation. Mentioning the sort list $s$ and the relation $I$ becomes superfluous in this case. 

For an intuitionistic Kripke frame, take again $s=\{1,\partial\}$, but with $W_1=W_\partial$ and $I$ the identity relation and equip the frame with a single relation $R$, axiomatized as a preorder relation. A Routley-Meyer frame for relevance logic is typically described as a quadruple $(W,R,*,0)$, with $0\in W$, $R$ a ternary relation on $W$ and $*$ a point star operator. The semantics can be equivalently given in frames with a perp relation $(W,R,\bot,0)$, see Dunn \cite{star-perp}. For Dunn's frames for relevance logic take again $s,I$ as in the previous cases and equip the frame with a ternary relation $R$, a binary relation $\bot$ and a unary relation (a set) $O=\{0\}$ on $W$. Similarly, for Goldblatt's orthoframes $(W,\bot)$ we let $s,I$ be as above and $\bot$ a binary relation.
\end{Remark}

\subsubsection{The Underlying Polarity of a Frame -- The Lattices of Stable and Co-stable Sets}
The relation $I$ generates a residuated pair $\sub{\largediamond}{I}:\powerset(W_1)\leftrightarrows\powerset(W_\partial):\sub{\lbbox}{I}$, defined  by
\[
\sub{\largediamond}{I}(U)=\{y\in W_\partial\midsp\exists x\in W_1(xIy\wedge x\in U)\}\hskip1cm
\sub{\lbbox}{I}(V)=\{x\in W_1\midsp\forall y\in W_\partial(xIy\lra y\in V)\}.
\]

The complement of $I$ will be designated by $\upv$ and we refer to it as the {\em Galois relation of the frame}. It generates a Galois connection $(\;)\rperp:\powerset(W_1)\leftrightarrows\powerset(W_\partial)^{\rm op}:\rperp(\;)$ defined by
\[
U\rperp=\{y\in W_\partial\midsp\forall u\in W_1 (u\in U\lra u\upv y)\}\hskip1cm
\rperp V=\{x\in W_1\midsp\forall y\in W_\partial(y\in V\lra x\upv y)\}.
\]

The closure operators generated by the residuated pair and the Galois connection are identical, i.e. $\sub{\lbbox}{I}\sub{\largediamond}{I} U=\rperp(U\rperp)$ and $\sub{\largesquare}{I}\sub{\lbdiamond}{I} V=(\rperp V)\rperp$,  by the fact that $U\rperp=\sub{\largesquare}{I}(-U)$ and ${}\rperp V=\sub{\lbbox}{I}(-V)$ (where $\sub{\largesquare}{I}=-\sub{\largediamond}{I}-$ and $\sub{\lbdiamond}{I}=-\sub{\lbbox}{I}-$). 

To simplify, we often use a priming notation for both Galois maps $(\;)\rperp$ and $\rperp(\;)$, i.e. we let $U'=U\rperp$, for $U\subseteq W_1$, and $V'=\rperp V$, for $V\subseteq W_\partial$. Hence $U''=\rperp(U\rperp)=\sub{\lbbox}{I}\sub{\largediamond}{I} U$ and $V''=(\rperp V)\rperp=\sub{\largesquare}{I}\sub{\lbdiamond}{I} V$. 

\begin{Remark}
\label{Galois map is complement}
The reader can easily verify that in the case of a classical Kripke frame (\cref{special case frames}) the Galois connection is simply the set-complement operation, $U'=-U$.
\end{Remark}

\begin{Definition}\label{Galois set lattice}
The complete lattice of all {\em Galois stable} sets $A''=A\subseteq W_1$ will be designated by $\mathcal{G}(W_1)$ and the complete lattice of all {\em Galois co-stable} sets $B''=B\subseteq W_\partial$ will be similarly denoted by $\mathcal{G}(W_\partial)$. We refer to Galois stable and co-stable sets as {\em Galois sets}.
\end{Definition} 

\begin{Remark}
\label{all Galois}
If the frame is a classical Kripke frame, then $A''=--A=A$, for any subset $A$, hence all subsets are Galois stable and thereby $\gpsi=\{A\subseteq W_1\midsp A=A''\}=\powerset(W_1)$.
\end{Remark}

Note that each of $W_1, W_\partial$ is a Galois set, but the empty set need not be Galois. The quasi-seriality condition~\eqref{quasi-seriality}
\begin{equation}
\label{quasi-seriality}
(\forall x\in W_1\exists y\in W_\partial\; xIy)\wedge(\forall y\in W_\partial\exists x\in W_1\; xIy)
\end{equation}
enforces that $\emptyset''=\emptyset$ (in both sorts).

A preorder relation is defined on each of $W_1,W_\partial$ by $u\leq w$ iff $\{u\}'\subseteq\{w\}'$. We call a frame {\em separated} if $\leq$ is in fact a partial order. For an element $u$ (of either $W_1$ or $W_\partial$) we write $\Gamma u$ for the set of elements $\leq$-above it. 

Sets $\Gamma w$ and $\{w\}'$ will be referred to as {\em principal elements}.  $\Gamma w$ will be referred to as a {\em closed element} and $\{w\}'$ as an {\em open element}. 

\begin{Remark}
\label{discrete order}
If the frame $\mathfrak{F}=(s,W,I,(R_j)_{j\in J},\sigma)$ is classical, then for $x,z\in W_1$ (= $W_\partial$), $x\leq z$ iff $\{x\}'\subseteq\{z\}'$ iff $-\{x\}\subseteq-\{z\}$ iff $x=z$. In other words the order is discrete, hence every classical frame is separated. By discreteness of the order, $\Gamma x=\{x\}$ and $\{x\}'=-\{x\}$, i.e. the closed and open elements are the atoms and co-atoms of the powerset algebra.
\end{Remark}

The following basic facts have a straightforward proof and they will be often used. For $G\subseteq W_1, v\in W_\partial$, we write $G\upv v$,  as an abbreviation for $\forall x(x\in G\lra x\upv v)$ and similarly for $u\upv G$, for $u\in W_1$ and $G\subseteq W_\partial$.
Let $\mathfrak{F}=(s,W,I,(R_j)_{j\in J},\sigma)$ be a  frame, $u\in W_1\cup W_\partial$ and $\upv$ the Galois relation of the frame. Let $v|G$ refer to either $G\upv v$, if $G\in\gpsi, v\in W_\partial$, or $v\upv G$, if $v\in W_1$ and $G\in\gphi$.
\begin{itemize}
\item $\upv$ is increasing in each argument place (and thereby its complement $I$ is decreasing in each argument place).
\item $(\Gamma u)'=\{u\}'$ and $\Gamma u=\{u\}^{\prime\prime}$ is a Galois set.
\item Galois sets are increasing, i.e. $u\in G$ implies $\Gamma u\subseteq G$.
\item For a Galois set $G$, $G=\bigcup_{u\in G}\Gamma u$.
\item For a Galois set $G$, $G=\bigvee_{u\in G}\Gamma u=\bigcap_{v|G}\{v\}'$.
\item For a Galois set $G$ and any set $F$, $F^{\prime\prime}\subseteq G$ iff $F\subseteq G$.
\end{itemize}

\begin{Definition}
\label{Galois dual relation}
For a sorted $(n+1)$-ary frame relation $R_j$, its {\em Galois dual relation} $R^\prime_j$ is defined by $R_j^\prime u_1\cdots u_n=(R_ju_1\cdots u_n)'$, where $R_ju_1\cdots u_n=\{u\midsp uR_ju_1\cdots u_n\}$.
\end{Definition}

\begin{Remark}
Evidently, in a classical Kripke frame $R'_ju_1\cdots u_n=-R_ju_1\cdots u_n$.
\end{Remark}

Notation is simplified by using vectors $\vec{u}=u_1\cdots u_n$, so that the definition is $R_j^\prime\vec{u}=(R_j\vec{u})'$. We let $\vec{u}[\;]_k$ be the vector with a hole (or just a place-holder) at the $k$-th position and write $u[w]_k$ either to display the element at the $k$-th place, or to designate the result of filling the $k$-th place of $u[\;]_k$, or to denote the result of replacing the element $u_k$ in $\vec{u}$ by the element $w$.
\begin{Definition}
For $1\leq k\leq n$, the {\em $k$-th section of an $(n+1)$-ary relation $S$} is the set $wS\vec{u}[\;]_k=\{x\in W_k\midsp wS\vec{u}[x]_k\}$. For $k=n+1$ the section is simply the set $S\vec{u}$.
\end{Definition}

\subsubsection{The Dual Sorted Powerset and Full Complex Algebra of a Frame}
Given a frame $\mathfrak{F}=(s,W,I,(R_j)_{j\in J},\sigma)$, each relation $R_j\subseteq W_{j_{n(j)+1}}\times\prod_{k=1}^{n(j)}W_{j_k}$ generates a sorted image operator, defined as in the Boolean case, except for the sorting
\begin{align}\label{sorted image ops}
F_j(\vec{U})&=\;\{w\in U_{i_{n(j)+1}}\midsp \exists \vec{u}\;(wR_j\vec{u}\wedge\bigwedge_{s=1}^{n(j)}(u_s\in U_s))\} &=\; \bigcup_{\vec{u}\in\vec{U}}R_j\vec{u}.
\end{align}
Equation~\eqref{sorted image ops} generalizes the J\'{o}nsson-Tarski definition of image operators in the representation of Boolean Algebras with Operators (BAOs) \cite{jt1}.

\begin{Definition}\label{sorted powerset algebra defn}
The {\em dual sorted powerset algebra} of a frame $\mathfrak{F}=(s,W,I,(R_j)_{j\in J},\sigma)$ is the algebra $\mathbf{P}=((\;)':\powerset(W_1)\leftrightarrows\powerset(W_\partial):(\;)',(F_j)_{j\in J} )$, where for each $j\in J$, $F_j$ is the sorted image operator  generated by the frame relation $R_j$ by~\eqref{sorted image ops}. 
\end{Definition}

For a  subset $U$ of $W_1$ or $W_\partial$ let $\overline{U}=U''$ be its closure and
if $F_j$ is the (sorted) image operator generated by the frame relation $R_j$, let $\overline{F}_j$ be the closure of the restriction of $F_j$ to Galois sets (stable, or co-stable, according to sort).
\[
\xymatrix{
\prod_{k=1}^{n(j)}\powerset(W_{j_k})\ar^{F_j}[rr]\ar^{(\;)''}@<0.5ex>@{->>}[d] && \powerset(W_{j_{n(j)+1}})\ar^{(\;)''}@<0.5ex>@{->>}[d]\\
\prod_{k=1}^{n(j)}\mathcal{G}(W_{j_k})\ar^{\overline{F}_j}[rr]\ar@<0.5ex>@{^{(}->}[u] && \mathcal{G}(W_{j_{n(j)+1}})\ar@<0.5ex>@{^{(}->}[u]
}
\]

The Galois connection is a dual isomorphism of the complete lattices of stable and co-stable sets, $(\;)':\mathcal{G}(W_1)\iso\mathcal{G}(W_\partial)^\mathrm{op}:(\;)'$. This allows for extracting single-sorted operators $\overline{F}_j^1:\prod_{k=1}^{n(j)}\mathcal{G}(W_1)\lra\mathcal{G}(W_1)$ and $\overline{F}_j^\partial:\prod_{k=1}^{n(j)}\mathcal{G}(W_\partial)\lra\mathcal{G}(W_\partial)$, by composition with the Galois connection maps
\begin{equation}\label{2single-sorted}
  \overline{F}^1_j(A_1,\ldots,A_{n(j)})=
  \left\{
  \begin{array}{cl}
  \overline{F}_j(\ldots,\underbrace{A_k}_{j_k=1},\ldots,\underbrace{A'_r}_{j_r=\partial},\ldots) & \mbox{if } j_{n(j)+1}=1
  \\
  (\overline{F}_j(\ldots,\underbrace{A_k}_{j_k=1},\ldots,\underbrace{A'_r}_{j_r=\partial},\ldots))' & \mbox{if } j_{n(j)+1}=\partial.
  \end{array}
  \right.
\end{equation}
The dual operators $\overline{F}_j^\partial:\prod_{k=1}^{n(j)}\mathcal{G}(W_\partial)\lra\mathcal{G}(W_\partial)$ are defined similarly.

\begin{Definition}\label{full complex algebra defn}
For a frame $\mathfrak{F}=(s,W,I,(R_j)_{j\in J},\sigma)$, its full complex algebra $\mathfrak{Cm}(\mathfrak{F})$ is defined as the normal lattice expansion $\mathfrak{Cm}(\mathfrak{F})=(\gpsi,\subseteq,\bigcap,\bigvee,\emptyset'', W_1,(\overline{F}^1_j)_{j\in J})$. 

A frame $\mathfrak{F}$ is {\em distributive} if its full complex algebra $\mathfrak{Cm}(\mathfrak{F})$ is a (completely) distributive lattice. Analogously for a Heyting and a Boolean frame.
\end{Definition}

\cref{classical section} discusses the cases of distributive, intuitionistic and Boolean frames.

\subsection{From Frames to Models}
\label{models section}
A relational model $\mathfrak{M}=(\mathfrak{F},V)$ consists of a frame $\mathfrak{F}$ and a sorted valuation $V=(V^1,V^\partial)$ of propositional variables, interpreting a variable $p$ as a Galois stable set $V^1(p)\in\gpsi$ and co-interpreting it as a Galois co-stable set $V^\partial(p)=V^1(p)\rperp\in\gphi$.  Interpretations and co-interpretations determine each other in the sense that for any sentence $\varphi\in\mathcal{L}_\tau$, if $\val{\varphi}\in\gpsi$ is an interpretation extending a valuation $V^1$ of propositional variables as stable sets, then $\val{\varphi}\rperp=\yvval{\varphi}\in\gphi$  is the co-interpretation extending the valuation $V^\partial$.

Satisfaction ${\forces}\subseteq W_1\times\mathcal{L}_\tau$ and co-satisfaction (refutation) ${\dforces}\subseteq W_\partial\times\mathcal{L}_\tau$ relations are then defined as expected, by $W_1\ni x\forces\varphi$ iff $x\in\val{\varphi}$ and $W_\partial\ni y\dforces\varphi$ iff $y\in\yvval{\varphi}$. Since satisfaction and co-satisfaction determine each other, for each operator it suffices to provide either its satisfaction, or its co-satisfaction (refutation) clause.

A generic, uniform relational interpretation for the language \eqref{language} of normal lattice expansions is presented in Table~\ref{sat}, following the principle of order-dual relational semantics introduced in \cite{odigpl}. The same interpretation applies to the case of the language of semilattice, or poset expansions (omitting the clauses for conjunction and/or disjunction when absent from the language). 

The interpretation pattern for the logical connectives $f_j$ depends only on their distribution (or monotonicity) types, in particular on their output types 1 or $\partial$, hence we present just two cases $f_1$ and $f_\partial$. In Table~\ref{sat} we assume that $x,u_j,v_j\in W_1$, $y,u_r,v_r\in W_\partial$ and that $R'_1,S'_\partial$ are the Galois duals of the relations $R_1,S_\partial$ corresponding to the logical operators $f_1,f_\partial$. 

\begin{table}[!htbp]
\caption{(Co)Satisfaction relations}
\label{sat}
\hrule
\begin{tabbing}
$x\forces p_i$\hskip8mm\=iff\hskip3mm\= $x\in V(p_i)$\\
$x\forces\top$ \>iff\> $x=x$\\
$y\dforces\bot$\>iff\> $y=y$\\
$x\forces\varphi\wedge\psi$\>iff\> $x\forces\varphi$ and $x\forces\psi$\\
$y\dforces\varphi\vee\psi$\>iff\> $y\dforces\varphi$ and $y\dforces\psi$\\
$y\dforces f_1(\varphi_1,\ldots,\varphi_n)$
\hskip.5cm\= iff\hskip3mm\= $\forall u_1\cdots u_n\;[\bigwedge_j^{i_j=1}(u_j\forces\varphi_j) \wedge\bigwedge_r^{i_r=\partial}(u_r\dforces\varphi_r)\lra yR'_1 u_1\cdots u_n]$\\
$x\forces f_\partial(\varphi_1,\ldots,\varphi_n)$
\>iff\> $\forall v_1\cdots v_n\;[\bigwedge_j^{i_j=1}(v_j\forces\varphi_j)\wedge\bigwedge_r^{i_r=\partial}(v_r\dforces\varphi_r)\lra xS'_\partial v_1\cdots v_n]$
\end{tabbing}
\hrule
\end{table}

Alternative, but equivalent interpretations can be given in concrete cases that may be more suitable for some purposes. For related examples, see \cite{choiceFreeStLog,dfnmlA,dfmlC} and \cref{imp frames section}.

\subsection{Canonical Extensions}
Define a {\em  filter} of a poset  $\mathbf{P}$ as a downwards directed upset $x\subseteq P$, i.e. $x$ is upwards closed in the $\leq$-order ($a\in x$ and $a\leq b$ imply $b\in x$) and it contains a lower bound $c_{a,b}\in x$ for each pair of elements $a,b\in x$. Dually, an {\em  ideal} is an upwards directed downset $y\subseteq P$, i.e. $y$ is downwards closed in the $\leq$-order and it contains an upper bound $d_{a,b}\in y$ for each pair of elements $a,b\in y$. Let $\filt(\mathbf{P}),\idl(\mathbf{P})$ be the sets of filters and ideals, respectively, of $\mathbf{P}$. 

If the poset is a meet-semilattice, then a filter $x$ as defined above is a meet-semilattice filter (an upset with $a,b\in x$ iff $a\wedge b\in x$) and similarly for ideals and join-semilattices. If the poset is a lattice, then filters and ideals as defined are the usual lattice filters and ideals. 

Throughout this article, we use $x_a=a{\uparrow}$ and $y_a=a{\downarrow}$ for the principal filter and ideal, respectively, generated by the element $a$.

To every poset $\mathbf{P}$ we associate the frame $\mathfrak{F}=\mathbf{P}_+=(s,W,I)$ where $s=\{1,\partial\}$,  $W=(\filt(\mathbf{P}),\idl(\mathbf{P}))$ and $I\subseteq \filt(\mathbf{P})\times\idl(\mathbf{P})$ is the non-intersection relation $xIy$ iff $x\cap y=\emptyset$ (so $x\upv y$ iff $x\cap y\neq\emptyset$). 

A completion of a poset (a mere partial order, or semilattice, or lattice) $\mathbf{P}$ is a complete lattice $\mathbf{C}$ with an order embedding $\alpha:\mathbf{P}\hookrightarrow\mathbf{C}$ which preserves existing finite joins and meets in $\mathbf{P}$. 
For an example, if $\mathbf{P}$ is a join semilattice, or a lattice, then $\alpha:\mathbf{P}\hookrightarrow\idl(\mathbf{P})$ is a completion, the ideal completion of $\mathbf{P}$, which is a join-dense (a $\Sigma_0$-) completion.

For a subset $A\subseteq P$ the Dedekind-MacNeille maps, defined by letting $A^\ell$ be the set of lower bounds of $A$, $A^u$ be its set of upper bounds, form a Galois connection: $B\subseteq A^\ell$ iff $A\subseteq B^u$. The collection $\mathrm{DM}(\mathbf{P})$ of all Galois closed subsets $A=A^{u\ell}$ is a complete lattice and the map $\alpha:\mathbf{P}\hookrightarrow \mathrm{DM}(\mathbf{P})$ defined by $\alpha(p)=\{p\}^{\ell}$ is an order embedding which preserves all existing joins and meets in $\mathbf{P}$, hence it is a completion of $\mathbf{P}$.  $\mathrm{DM}(\mathbf{P})$ is characterized \mcite{Theorem~7.41}{hilary_davey_2002} as the unique, up to an order-isomorphism, completion $\mathbf{C}$ of $\mathbf{P}$ in which the image of $\mathbf{P}$ under the embedding $\alpha:\mathbf{P}\hookrightarrow\mathbf{C}$  is both join-dense and meet-dense in $\mathbf{C}$. We say that $\mathbf{C}$ is a $\Delta_0$-completion of $\mathbf{P}$ (both a $\Sigma_0$- and a $\Pi_0$-completion).

If $\alpha:\mathbf{P}\hookrightarrow\mathbf{C}$ is a completion then a {\em filter-element} $F$ (also referred to as a closed element) of $\mathbf{C}$ is a meet in $\mathbf{C}$ of the image $\alpha[x]\subseteq C$ of a filter $x$ of $\mathbf{P}$, i.e. $F=F_x=\bigwedge_\mathbf{C}\alpha[x]$, for $x\in\filt(\mathbf{P})$. Let $K(\mathbf{C})$ be the set of filter (closed)-elements.
Dually, an {\em ideal-element} (also referred to as an open element) of $\mathbf{C}$ is a join in $\mathbf{C}$ of the image $\alpha[y]\subseteq C$ of an ideal $y$ of $\mathbf{P}$, i.e. $G=G_y=\bigvee_\mathbf{C}\alpha[y]$, for $y\in\idl(\mathbf{P})$. Let $O(\mathbf{C})$ be the set of ideal (open)-elements.

If the set $\mathcal{F}=K(\mathbf{C})$ of filter elements is join-dense in $\mathbf{C}$, then $\mathbf{C}$ is a $\Sigma_1$-completion of $\mathbf{P}$ (with respect to $\mathcal{F}$). Every element of $\mathbf{C}$ is a join of meets of elements from the image of $\mathbf{P}$. It is also common to say that $\mathbf{C}$ is a $\bigvee\bigwedge$-dense completion. Dually, if the set $\mathcal{I}=O(\mathbf{C})$ of ideal elements is meet dense in $\mathbf{C}$, then $\mathbf{C}$ is a $\Pi_1$-completion of $\mathbf{P}$ (with respect to $\mathcal{I}$). Every element of $\mathbf{C}$ is a meet of joins of elements from the image of $\mathbf{P}$. We also say that $\mathbf{C}$ is a $\bigwedge\bigvee$-dense completion. If $\mathbf{C}$ is both a $\bigvee\bigwedge$-dense (so $\Sigma_1$) completion and a $\bigwedge\bigvee$-dense (so $\Pi_1$) completion, then we refer to $\mathbf{C}$ as a $\Delta_1$-completion, or a doubly-dense, or just dense completion (with respect to $(\mathcal{F,I})$). $\Delta_1$-completions are studied in \cite{delta1}.

Note that if $\mathbf{C}$ is a $\Delta_1$-completion of $\mathbf{P}$, then it is the MacNeille completion of the poset $(\mathcal{F}\cup \mathcal{I},\leq_\mathbf{C})$, which may be referred to as the {\em intermediate structure}.  

The choice of $(\mathcal{F,I})$ in a proof of $\Delta_1$-completion is not unique \cite{delta1,Morton2014}. In the proof of existence of canonical extensions for bounded lattices \mcite{Proposition~2.6}{mai-harding} the sets of filter and ideal elements are used. In \cite{delta1} the choice is made to take $\mathcal{F}$ to consist of order-filters (upsets) and $\mathcal{I}$ of order-ideals (downsets) of a partial-order $\mathbf{P}$. For partial orders and lattices, the RS-frames project \cite{dunn-gehrke,mai-gen} takes $\mathcal{F}=J^\infty(\mathbf{C}), \mathcal{I}=M^\infty(\mathbf{C})$, the sets of completely join and completely meet irreducible elements, respectively, of the complete lattice $\mathbf{C}$ of Galois stable sets. In the latter case, choice is needed in the proof of $\Delta_1$-completion. Appeal to choice is unnecessary when taking $\mathcal{F}=K(\mathbf{C})$, the set of filter elements, and  $\mathcal{I}=O(\mathbf{C})$, the set of ideal elements, as seen in the proof of \cref{delta1-lemma} below.

\begin{Proposition}
\label{delta1-lemma}
  Let $\mathbf{P}$ be a poset, a semilattice or a lattice and let $\mathfrak{F}=\mathbf{P}_+$ be its associated frame $(s,W,I)$ where $W_1=\filt(\mathbf{P}), W_\partial=\idl(\mathbf{P})$ and $xIy$ iff $x\cap y=\emptyset$. Let $\gpsi$ be the complete lattice of Galois stable sets of filters and $\alpha:\mathbf{P}\hookrightarrow\gpsi$ be the map $\alpha(a)=\{x\in W_1=\filt(\mathbf{P})\midsp a\in x\}$. Then $\gpsi$ is a $\Delta_1$-completion of $\mathbf{P}$.
\end{Proposition}
\begin{proof}
If $x$ is a filter of $\mathbf{P}$, then $\alpha[x]=\{\alpha(a)\midsp a\in x\}=\{\Gamma x_a\midsp a\in x\}$. The filter-elements are then of the form $K=\bigcap_{a\in x}\Gamma x_a$. So $u\in K$ iff $\forall a(a\in x\lra x_a\subseteq u)$ iff $x\subseteq u$. Hence $K=\Gamma x$ which we had called a closed element in any frame. If $A$ is any stable set, then it is straightforward to verify that $A=\bigvee_{x\in A}\Gamma x=\bigcup_{x\in A}\Gamma x$. 

An ideal of $\mathbf{P}$ is a filter of $\mathbf{P}^\partial$ (order reversed). A filter element of $\gphi$ is then a principal upper set $\Gamma y$, i.e. a closed element of $\gphi$. Then for its Galois image $(\Gamma y)'=\{y\}'$ we have that $\{y\}'=(\Gamma y)'=(\bigcap_{a\in y}\Gamma y_a)'=\bigvee_{a\in y}\{y_a\}'=\bigvee_{a\in y}\Gamma x_a=\bigvee\{\alpha(a)\midsp a\in y\}$, using the fact that $\{y_a\}'=\{x\in W_1\midsp x\upv y_a\}=\{x\in W_1\midsp a\in x\}=\Gamma x_a$. Therefore, an ideal-element $I\in O(\gpsi)$ is of the form $I=\{y\}'=\lperp\{y\}$, which we called an open element in any frame. For any stable set $A$ it is immediate that $A=\bigcap_{A\upv y}\lperp\{y\}=\bigcap_{A\upv y}\{y\}'$, where $A\upv y$ iff $A\subseteq\lperp\{y\}$.

Hence $\gpsi$ is both a $\Sigma_1$- and a $\Pi_1$-completion, so a $\Delta_1$-completion of $\mathbf{P}$.
\end{proof}

For a bounded lattice $\mathbf{L}$, a {\em canonical extension} of $\mathbf{L}$ was defined in \cite{mai-harding} as a {\em compact $\Delta_1$-completion}, where compactness can be defined in any of the equivalent ways below, by \mcite{Definition~2.3, Lemma~2.4}{mai-harding}:
\begin{itemize}
\item[(1)] For any set $A$ of filter (closed) elements and any set $B$ of ideal (open) elements $\bigwedge A\leq\bigvee B$ iff there are finite subsets $A_f\subseteq A$, $B_f\subseteq B$ such that $\bigwedge A_f\leq\bigvee B_f$.
\item[(2)] For any $S,T\subseteq L$,   $\bigwedge\alpha[S]\leq\bigvee\alpha[T]$  iff there exist finite subsets $S_f\subseteq S,  T_f\subseteq T$ such that $\bigwedge\alpha[S_f]\leq\bigvee\alpha[T_f]$.
\item[(3)] For any lattice filter $x$ and ideal $y$, and any $S\subseteq x, T\subseteq y$, $\bigwedge\alpha[S]\leq\bigvee\alpha[T]$ iff $x\cap y\neq\emptyset$.
\end{itemize}
The following is readily verified.
\begin{Proposition}
\label{compactness lemma}
The completion of a lattice proven to be a $\Delta_1$-completion in~\cref{delta1-lemma} is compact. 
\end{Proposition}
\begin{proof}
Given the dual isomorphism $\gpsi\iso\gphi^\partial$, let $\eta(a)=(\alpha(a))'\in\gphi$, so that $\eta(a)=\{y\in\idl(\mathbf{L})\midsp a\in y\}=\Gamma y_a$. Note that for any $S,T\subseteq L$,  $\bigwedge\alpha[S]=\bigwedge_{a\in S}\Gamma x_a=\Gamma(\bigvee_{a\in S}x_a)$ and letting $x^s$ be the filter generated by $S$, $x^s=\bigvee_{a\in S}x_a$, we obtain $\bigwedge\alpha[S]=\Gamma x^s$. Similarly, $\bigwedge\eta[T]=\Gamma y^t$, where $y^t$ is the ideal generated by $T$. Then $\bigvee\alpha[T]=(\bigwedge\eta[T])'=(\Gamma y^t)'=\lperp\{y^t\}$. 

The assumption that $\bigwedge\alpha[S]\leq\bigvee\alpha[T]$ means that $\Gamma x^s\subseteq\lperp\{y^t\}$, which is equivalent to $x^s\upv y^t$. But this means, by definition, that $\exists a\in L$ such that $a\in x^s\cap y^t\neq\emptyset$. 

Hence there are finite subsets $\{e_1,\ldots,e_n\}=S_f\subseteq S$ and $\{b_1,\ldots,b_m\}=T_f\subseteq T$ such that $e=e_1\wedge\cdots\wedge e_n\leq a\leq b_1\vee\cdots\vee b_m=b$. It follows that $\bigwedge\alpha[S_f]=\Gamma x_e\subseteq\Gamma x_a\subseteq \Gamma x_b=\lperp\{y_b\}=\lperp(\Gamma y_b)=\lperp(\bigwedge\eta[T_f])=\bigvee\alpha[T_f]$. 
\end{proof}

An inspection of the proof shows that compactness (for a lattice canonical extension) can be equivalently defined by a simplified version of condition (3) in \mcite{Lemma~2.4}{mai-harding}:
\begin{itemize}
\item[] ({\bf Compactness})\hskip2cm For any filter $x$ and any ideal $y$, $\bigwedge\alpha[x]\leq\bigvee\alpha[y]$ iff $x\cap y\neq\emptyset$.
\end{itemize}
As stated above the compactness condition applies uniformly to a poset $\mathbf{P}$, a semilattice $\mathbf{S}$, or a lattice $\mathbf{L}$ and it is indeed used as a definition of compactness by Gouveia and Priestley~\cite{hilary-sem} in their study of canonical extensions of semilattices (and their underlying posets). The uniformity of the definition of $\Delta_1$-completion and of compactness for all three cases aligns perfectly with the fact that the canonical extension of a lattice is the same as the canonical extensions of each of  its semilattice and its poset reducts. 

\begin{Proposition}
  \label{compactness of delta1}
Let $\mathbf{P}$ be a poset, a meet-semilattice, or a lattice. Then its $\Delta_1$-completion defined as in \cref{delta1-lemma} is compact, hence a canonical extension of $\mathbf{P}$.  
\end{Proposition}
\begin{proof}
Let $x,y$ be a filter and an ideal of $\mathbf{P}$. Then $\bigwedge\alpha[x]=\bigcap_{a\in x}\Gamma x_a=\Gamma x$ and $\bigvee\alpha[y]=(\bigwedge\eta[y])'=(\Gamma y)'=\rperp\{y\}$, where $\eta:\mathbf{P}\hookrightarrow\gphi\iso\gpsi^\partial$ is the co-embedding of $\mathbf{P}$. Thus $\bigwedge\alpha[x]\leq\bigvee\alpha[y]$ iff $\Gamma x\subseteq\lperp\{y\}$ iff $x\upv y$ iff $x\cap y\neq\emptyset$.
\end{proof}

\begin{Definition}
  \label{canext defn}
  Let $\mathbf{P}=(P,(f_j)_{j\in J})$ be a poset, or semilattice, or a lattice with operations $f_j$, for $j\in J$. The canonical extension $\mathbf{P}^\delta$ of $\mathbf{P}$ consists of a compact $\Delta_1$-completion of $\mathbf{P}$ together with an extension $F_j$ of each map $f_j$ such that if the type of $f_j$ is $\delta(j)=(\ldots;1)$, then $F_j=f^\sigma_j$ is its $\sigma$-extension and if its type is $\delta(j)=(\ldots;\partial)$, then $F_j=f^\pi_j$ is its $\pi$-extension, as these are defined in \cite{mai-harding}.
\end{Definition}

\section{Pure Implicative Logics}
\label{posets section}
\subsection{Algebras, Languages and Logics}
The focus in this Section is on the sentential language of the quasi-variety of implicative bounded above posets $\mathbf{P}=(P,\leq,1,\type{e},\ra)$, i.e. bounded above posets with a binary (implication) operation and a distinguished element $\type{e}$ subject to the axioms 
\begin{description}
\item $a\leq b\leftrightarrow \type{e}\le a\ra b$,\hskip1cm 
\item $(a\leq b)\wedge(c\leq d)\lra (b\ra c\leq a\ra d)$.
\end{description}
The poset is an {\em integral} implicative poset if $\type{e}=1$.

The sentential language of implicative bounded above posets is the language  $\mathcal{L}\ni\varphi:= p_i\;(i\in\mathbb{N})\midsp\top\midsp\type{t}\midsp\varphi~\rfspoon~\varphi$, where the monotonicity type of implication is $\delta(\rfspoon)=(1,\partial;\partial)$. Viewed as a 0-ary operation, we may assign to the truth constant $\type{t}$ the monotonicity type $\delta(\type{t})=(\;;1)$.  The relevant fragment of the full logic of normal lattice expansions of Table~\ref{minimal proof system} is displayed in~\eqref{implicative logic}. 

\begin{align}
\label{implicative logic}
\begin{tabular}{cccccc}
$\varphi\proves\varphi$ & $\varphi\proves\top$ & $\infrule{\varphi\proves\psi\hskip4mm\psi\proves\vartheta}{\varphi\proves\vartheta}$ & $\infrule{\varphi(p)\proves\psi(p)}{\varphi(\vartheta_{/p})\proves\psi(\vartheta_{/p})}$ 
\\[2mm]
$\proves\type{t}$ & $\infrule{\varphi\proves\psi}{\overline{\type{t}\proves\varphi\rfspoon\psi}}$ & $\infrule{\psi\proves\vartheta}{\varphi\rfspoon\psi\proves\varphi\rfspoon\vartheta}$ & $\infrule{\psi\proves\vartheta}{\vartheta\rfspoon\varphi\proves\psi\rfspoon\varphi}$ 
\end{tabular}
\end{align}

The integral system is the  special case, arising by adding the axiom $\top\proves\type{t}$. We let $\mathbf{\Lambda}_t$ be the logic generated by the proof system (equivalently viewed either as a set of provable sequents $\varphi\proves\psi$ or theorems $\varphi$, provided $\type{t}\proves\varphi$ is provable). 

The monotonicity and antitonicity rules ensure that provable equivalence is a congruence. It is straightforward to verify that the Lindenbaum-Tarski algebra of the logic is an implicative bounded above poset.

\subsection{Implicative Frames and Relational Semantics}
\label{imp frames section}
By \cref{corresponding frames}, frames $\mathfrak{F}=(s,W,I,U,T,\sigma)$ for the logic of implicative posets are  structures with $\sigma(U)=(1:\;)$ and $\sigma(T)=(\partial:1\partial)$, i.e. $U\in\gpsi$ is a distinguished Galois stable set (a unary relation). The frame is integral if $U=W_1$.

Given a model $\mathfrak{M}=(\mathfrak{F},V)$ on the frame $\mathfrak{F}$, the satisfaction clause for $\rfspoon$ is obtained by instantiating the relevant clause in Table~\ref{sat} for a connective $f_\partial$, hence we obtain the clause
\begin{equation}\label{sat implication}
  x\forces\varphi\rfspoon\psi\mbox{ iff }\forall z\in W_1\forall y\in W_\partial(z\forces\varphi\wedge y\dforces\psi\lra xT'zy)
\end{equation}
and for the constant $\type{t}$ we add the clause $u\forces\type{t}$ iff $u\in U$. If a residual product operator is included in the language (as in the Lambek calculus, Relevance logic, or Linear logic), then the interpretation can be equivalently given by the Routley-Meyer clause
\begin{equation}
x\forces \varphi\rfspoon\psi\mbox{ iff } \forall u,z\in W_1(u\forces\varphi\wedge  zR^{111}ux\lra z\forces\psi)
\end{equation}
where $R$ is a frame relation interpreting the product operator (see \cref{lambek section}). In the case of Intuitionistic logic, the interpretation can be equivalently given by the Kripke clause
\begin{equation}
x\forces \varphi\rfspoon\psi\mbox{ iff }\forall z\in W_1(z\forces\varphi\wedge x\leq z\lra z\forces\psi),
\end{equation}
see \mcite{Section~4.3}{choiceFreeHA} and \cref{classical section}, this article.

In a classical frame, where the order is discrete, the satisfaction clause reduces to the usual clause for Boolean implication
\begin{equation}\label{classical imp}
x\forces \varphi\rfspoon\psi\mbox{ iff }x\forces\varphi\mbox{ implies } x\forces\psi,
\end{equation}
see \cref{classical section}, this article.

Validity of a sequent $\varphi\proves\psi$ means that $\val{\varphi}\subseteq\val{\psi}$ and we say that a sentence $\varphi$ is valid if $\val{\type{t}}=U\subseteq\val{\varphi}$.

Let $\mathbb{PU}$ be the frame class axiomatized in \cref{substructural implicative poset frame axioms}.

\begin{table}[!htbp]
\caption{$\mathbb{PU}$ Frame Class axioms for the Logic of Implicative Posets with a Truth Constant}
\label{substructural implicative poset frame axioms}
\hrule
\begin{enumerate}
\item[(F1)] For all $x\in W_1$ and all $v\in W_\partial$, the section $Txv\subseteq W_\partial$ of the relation $T^{\partial 1\partial}$ is a Galois set.
\\[1mm] \underline{Unit Axiom}
\item[(U)] $U$ is a Galois stable set such that for all $x\in W_1$ and all $v\in W_\partial$, $x\upv v$ holds iff for all $u$ in $U$, $uT'xv$ holds.
\end{enumerate}
\hrule
\end{table}

\begin{Theorem}
\label{lambda t soundness}
The logic $\mathbf{\Lambda}_t$ is sound  in the class $\mathbb{PU}$ of frames axiomatized in Table~\ref{substructural implicative poset frame axioms}.
\end{Theorem}
\begin{proof}
Given a frame  $\mathfrak{F}=(s,W,I,U,T^{\partial 1\partial})$, where $T\subseteq W_\partial\times(W_1\times W_\partial)$,
a sorted image operator $F_T=\largetriangleright:\powerset(W_1)\times\powerset(W_\partial)\lra\powerset(W_\partial)$ is defined on $X\subseteq W_1, Y\subseteq W_\partial$ by instantiating the definition in \eqref{sorted image ops} as shown below
\[
X{\largetriangleright} Y=F_T(X,Y)=\{y\in W_\partial\midsp\exists x,v(x\in X\;\wedge\;v\in Y\;\wedge\;yTxv)\}=\bigcup_{x\in X}^{v\in Y}Txv.
\]
Let ${\Mtright}:\gpsi\times\gphi\lra\gphi$ designate the closure $\overline{F}_T$ of the restriction of $F_T=\largetriangleright$ to Galois sets, defined on $A\in\gpsi, B\in\gphi$ by
\[
A\Mtright B=(A\ltright B)''=\left(\{y\in W_\partial\midsp\exists x,v(x\in A\;\wedge\;v\in B\;\wedge\;yTxv)\}\right)''=\bigvee_{x\in A, v\in B}Txv.
\]
Instantiating the definition in \eqref{2single-sorted}, we define a single-sorted binary operation $\Ra\;=\;\overline{F}^1_T$ . 
For $A,C\in\gpsi$, $A\Ra C=(A\Mtright C')'={}\rperp(A\Mtright C\rperp)$ and
observe that $A\Ra C=(A\Mtright C')'=(A\largetriangleright C')'''=(A\largetriangleright C')'$.
For membership in $A\Ra C$, we calculate that
\begin{tabbing}
$u\in(A\Ra C)$\hskip3mm\= iff\hskip2mm\= $u\upv (A\ltright C')$\\
\>iff\> $\forall y\in W_\partial(y\in A\ltright C'\lra u\upv y)$\\
\>iff\> $\forall y\in W_\partial(\exists x\in W_1\exists v\in W_\partial(x\in A\wedge C\upv v\wedge yTxv)\lra u\upv y)$\\
\>iff\> $\forall v\in W_\partial\forall x\in W_1(x\in A\wedge C\upv v\lra(Txv\subseteq\{u\}'))$\\
\>\> (using the stability axiom (F1))\\
\>iff\> $\forall v\in W_\partial\forall x\in W_1(x\in A\wedge C\upv v\lra(\Gamma u\subseteq T'xv))$\\
\>iff\>  $\forall v\in W_\partial\forall x\in W_1(x\in A\wedge C\upv v\lra uT'xv)$.
\end{tabbing}
Given the satisfaction clause for implication in~\eqref{sat implication}, for any model $\mathfrak{M}=(\mathfrak{F},V)$, $\val{\varphi\rfspoon\psi}_\mathfrak{M}=\val{\varphi}_\mathfrak{M}{\Ra}\val{\psi}_\mathfrak{M}= (\val{\varphi}_\mathfrak{M}\ltright\yvval{\psi}_\mathfrak{M})'$, recalling that $\val{\psi}'_\mathfrak{M}=\yvval{\psi}_\mathfrak{M}$.
By the definition of $\Ra$, the monotonicity rules are valid. 

For validity of the double-line rule {\small $\infrule{\varphi\proves\psi}{\overline{\type{t}\proves\varphi\rfspoon\psi}}$}, assume first that $A\subseteq C$, for $A,C\in\gpsi$, and let $u\in U$. To show $u\in A\Ra C$ and given the membership condition computed above, let $z\in A$ and $C\upv v$. The assumption implies $z\upv v$ and then by the frame axiom (U) it follows that $uT'zv$. Hence $u\in A\Ra C$ and so $U\subseteq A\Ra C$. For the converse, the inclusion $U\subseteq A\Ra C$ is equivalent to the condition $\forall u\in U\forall z,v(z\in A\wedge C\upv v\lra uT'zv)$, which is equivalent to $\forall z,v(z\in A\wedge C\upv v\lra\forall u\in U uT'zv)$. By the frame axiom (U), this is equivalent to $\forall z,v(z\in A\wedge C\upv v\lra z\upv v)$, which means that $C'\subseteq A'$ and so $A\subseteq C$.
\end{proof}

\begin{Remark}
\label{Ra as restriction rem}
A relation $R^{111}$ can be defined from $T^{\partial 1\partial}$, as in \mcite{Definition~6}{dfnmlA}. It generates an image operator $\bigodot$, residuated with implication operations ${}_T\hskip-1mm\La$ and $\Ra_T$. It was shown in \mcite{Proposition~5}{dfnmlA}, as an application of \mcite{Theorem~3.14}{duality2}, that $\Ra$ is the restriction of $\Ra_T$ to stable sets. This result is not available in the present context of frames for partial orders as it requires either the additional frame axioms corresponding to the (co)distribution properties of implication in a lattice, or the inclusion in the logic of residuals for implication, as in the Lambek calculus, see \cref{lambek section}.
\end{Remark}

\subsection{Representation, Canonical Extension and Completeness}
\label{rep section}
Let $\mathbf{P}$ be an implicative bounded above poset. 

\begin{Definition}
\label{canonical frame defn}
Define the  structure $\mathbf{P}_+=\mathfrak{F}=(s,W,I,U,T,\sigma)$, the dual frame of the implicative poset $\mathbf{P}$, by setting $s=\{1,\partial\}$, $\sigma(T)=(\partial: 1\partial)$ and
\begin{itemize}
\item $W_1=\filt(\mathbf{P})$, the set of filters of $\mathbf{P}$
\item $W_\partial=\idl(\mathbf{P})$, the set of ideals of $\mathbf{P}$
\item $I\subseteq W_1\times W_\partial$, defined by $xIy$ iff $x\cap y=\emptyset$
\item $U=\Gamma x_\type{e}$,
\item $T\subseteq W_\partial\times(W_1\times W_\partial)$, defined by 
\[
\mbox{$yTxv$ iff $\forall a,b\in P(a\in x\wedge b\in v\lra (a\ra b)\in y)$.}
\]
\end{itemize}
Let $\gpsi$, $\gphi$ be the respective complete lattices of Galois stable and costable sets. The representation and co-representation maps $\alpha,\eta$ are given by $\alpha(a)=\{x\in\filt(\mathbf{P})\midsp a\in x\}=\Gamma x_a$ and $\eta(a)=\{y\in\idl(\mathbf{P})\midsp a\in y\}=\Gamma y_a$.
\end{Definition}

We establish that both frame axioms (F1) and (U) hold for $\mathfrak{F}$, in Proposition~\ref{F1 is canonical} and Proposition~\ref{U is canonical}, respectively.

Recall first that in a frame $\mathfrak{F}=(s,W,I,(R_j)_{j\in J},\sigma)$ and $x\in W_1, y\in W_\partial$ we call the principal upper sets $\Gamma x, \Gamma y$ {\em closed elements} of $\gpsi$ and of $\gphi$, respectively. And that we also call $\rperp\{y\},\{x\}\rperp$ {\em open elements} of $\gpsi$ and of $\gphi$, respectively.  In the canonical frame we have
\begin{description}
  \item  $(\Gamma x_a)\rperp=\{x_a\}\rperp=\Gamma y_a$ and $\lperp\{y_a\}=\lperp(\Gamma y_a)=\Gamma x_a$.
\end{description}
Indeed, $x\upv y_a$ iff $a\in x$, hence $\rperp\{y_a\}=\eta(a)'=\Gamma x_a=\alpha(a)$ and similarly $\{x_a\}\rperp=\alpha(a)'=\eta(a)=\Gamma y_a$. 

\begin{Lemma}
\label{tright lemma}
For a filter $x$ and an ideal $v$, the set of elements $x\tright v$ defined in~\eqref{x-tright-v} is an ideal
\begin{equation}\label{x-tright-v}
x\tright v=\{e\midsp\exists a\in x\exists b\in v\;e\leq a\ra b\}
\end{equation}
and for poset elements $a,b$, the identity $x_a{\tright}y_b=\Gamma y_{a\ra b}=\eta(a\ra b)$ holds.
\end{Lemma}
\begin{proof}
Clearly $x{\tright}v$ is a downset. Let now $e_1,e_2\in (x\tright v)$ so that, by definition, there exist $a_1,a_2\in x$ and $b_1,b_2\in v$ such that $e_1\leq a_1\ra b_1$ and $e_2\leq a_2\ra b_2$.
  
Since $x$ is a filter, let $c_{12}\leq a_1, a_2$ be a lower bound in $x$. Since $v$ is an ideal, let $d_{12}\geq b_1,b_2$ be an upper bound in $v$. From $c_{12}\leq a_1$ and $b_1\leq d_{12}$ we obtain that $a_1\ra b_1\leq c_{12}\ra d_{12}$ and, similarly also that $a_2\ra b_2\leq c_{12}\ra d_{12}$. Then $c_{12}\ra d_{12}$ is in $x\tright v$, by definition of the latter, and it is also an upper bound of $e_1,e_2$. Thus $x\tright v$, which is obviously a downset, is also upwards directed, hence an ideal.

The point operator $\tright:W_1\times W_\partial\lra W_\partial$ defined in~\eqref{x-tright-v} is monotone in both argument places. Furthermore, $e\in x_a{\tright}y_b$ iff for some $a_1,b_1$ with $a\leq a_1,b_1\leq b$ we have $e\leq a_1\ra b_1\leq a\ra b$. Conclude that $x_a{\tright}y_b=\Gamma y_{a\ra b}$.
\end{proof}

\begin{Proposition}
\label{F1 is canonical}
  $Txv$ is a Galois set, for any filter $x$ and ideal $v$. More precisely, it is a closed element $Txv=\Gamma(x{\tright}v)$ of $\gphi$. Hence the canonical frame satisfies the frame axiom (F1).
\end{Proposition}
\begin{proof}
We establish that $yTxv$ iff $(x\tright v)\subseteq y$.

Assume $yTxv$ and let $e\in(x{\tright}v)$. Let $a\in x, b\in v$ so that $e\leq a\ra b$. By definition of $T$, $a\in x, b\in v$ implies that $a\ra b\in y$. Now $y$ is an ideal and $e\leq a\ra b$, so $e\in y$.

For the converse, assuming $(x\tright v)\subseteq y$ and $a\in x, b\in v$ we need to show that $a\ra b\in y$. But $(a\ra b)\in (x{\tright}v)\subseteq y$. 

The argument showed that $Txv$, for any filter $x$ and ideal $v$, is a Galois set, more precisely that it is a closed element $\Gamma(x{\tright}v)=\{y\midsp (x{\tright}v)\subseteq y\}$ of $\gphi$.
\end{proof}

\begin{Proposition}
\label{U is canonical}
For any filter $x\in W_1$ and ideal $v\in W_\partial$, $x\upv v$ iff $\forall u\in U\; uT'xv$. Therefore the canonical frame satisfies the frame axiom (U), as well, and so it belongs to the frame class $\mathbb{PU}$.
\end{Proposition}
\begin{proof}
First, notice that $uT'xv$ iff $u\upv(x{\tright}v)$, i.e. $T'xv=\lperp\{x{\tright}v\}$. This follows from the definition of the Galois dual relation $T'$ of $T$, from the fact that $yTxv$ iff $x{\tright}v\subseteq y$, as established in the proof of Proposition~\ref{F1 is canonical} and from the fact that $\upv$ is increasing in both argument places, by the way the order was defined.

Now assume $x\upv v$ and let $a\in x\cap v$. Then $\type{e}\leq a\ra a\in (x{\tright}v)$, hence $\type{e}\in x{\tright}v$. Then for any filter $u\in U=\Gamma x_\type{e}$, $u\cap(x{\tright}v)\neq\emptyset$. In other words, $u\upv x{\tright}v$, i.e. $uT'xv$. For the converse, assuming that $x{\tright}v$ intersects every filter $u\in U$ and taking in particular $u=x_\type{e}$, the principal filter generated by $\type{e}$, we obtain that $\type{e}\in(x{\tright}v)$. By definition of the ideal $x{\tright}v$, let $a\in x, b\in v$ such that $\type{e}\leq a\ra b$. By the double-line rule of the proof system this is equivalent to $a\leq b$ and this implies that $b\in x$ as well, hence $x\upv v$.
\end{proof}

Let $\gpsi$, $\gphi$ be the full complex and the dual full complex algebra of the dual frame of $\mathbf{P}$. 

\begin{Proposition}
\label{rep embedding prop}
The representation map $\alpha$ is an embedding $\alpha:\mathbf{P}\hookrightarrow\gpsi$ of implicative posets.
\end{Proposition}
\begin{proof}
It suffices to show that for any elements $a,b$ in $\mathbf{P}$, $\alpha(a\ra b)=\alpha(a)\Ra\alpha(b)$.

By definition, $\alpha(a)\Ra\alpha(b)=(\alpha(a)\ltright(\alpha(b))')'=(\alpha(a)\ltright\eta(b))'=\Gamma x_a\Ra\lperp\{y_b\}$.
\begin{tabbing}
$\alpha(a)\Ra\alpha(b)$\hskip2mm\==\hskip2mm\=$\Gamma x_a\Ra\rperp\{y_b\}$\hskip2mm\==\hskip2mm\=$\{u\in W_1\midsp \forall x,v(x\in\Gamma x_a\wedge \rperp\{y_b\}\upv v\lra uT'xv)\}$\\
\>\>\>=\> $\{u\in W_1\midsp \forall x,v(a\in x\wedge b\in v\lra uT'xv)\}$\\
\>\>\>=\> $\{u\in W_1\midsp \forall x,v(a\in x\wedge b\in v\lra u\upv(x{\tright}v))\}$.
\end{tabbing}
If $u\in (\Gamma x_a\Ra\rperp\{y_b\})$, then in particular $u\upv(x_a{\tright}y_b)$ and since $(x_a{\tright}y_b)=\Gamma y_{a\ra b}$ it follows that $(a\ra b)\in u$, i.e. $u\in \Gamma x_{a\ra b}=\alpha(a\ra b)$, and so $(\alpha(a)\Ra\alpha(b))\subseteq \alpha(a\ra b)$. 

Conversely, if $u\in \alpha(a\ra b)=\Gamma x_{a\ra b}$, i.e. $(a\ra b)\in u$, then $u\upv\Gamma y_{a\ra b}$, equivalently $u\upv (x_a{\tright}y_b)$ and thus, by monotonicity of $\tright$ and of $\upv$, for any $x,v$ such that $x_a\subseteq x, y_b\subseteq v$ we shall have $u\upv (x{\tright}v)$. This means that $\alpha(a\ra b)\subseteq(\alpha(a)\Ra\alpha(b))$. 

Conclude that $\alpha(a\ra b)=(\alpha(a)\Ra\alpha(b))=(\alpha(a){\ltright}(\alpha(b))')'=(\alpha(a){\ltright}\eta(b))'$. 
\end{proof}

\begin{Theorem}[Completeness]
The logic $\mathbf{\Lambda}_t$ of implicative posets is sound and complete in the class $\mathbb{PU}$ of frames validating the axioms (F1), (U) stated in Table~\ref{substructural implicative poset frame axioms}.
\end{Theorem}
\begin{proof}
The claim is that
for any sentences $\varphi,\psi$ in the language of the logic and any model $\mathfrak{M}=(\mathfrak{F},V)$, whose base frame $\mathfrak{F}$ validates axioms (F1) and (U), $\varphi\proves\psi$ iff $\val{\varphi}_\mathfrak{M}\subseteq\val{\psi}_\mathfrak{M}$. The soundness direction was proven in \cref{lambda t soundness}. For the converse, if $\varphi\not\proves\psi$, then we verify that in the canonical model $\val{\varphi}_\mathfrak{M}\not\subseteq\val{\psi}_\mathfrak{M}$. This follows from the fact that $\varphi\not\proves\psi$ iff $[\varphi]\not\leq[\psi]$ in the Lindenbaum-Tarski algebra of the logic and then by the fact that the representation map $\alpha$ is an implicative bounded above poset embedding, proven in Proposition~\ref{rep embedding prop}.
\end{proof}

\begin{Proposition}
\label{pi extension lemma}
The complete lattice $\mathcal{G}(\filt(\mathbf{P}))$ of Galois stable sets of the frame $(\filt(\mathbf{P}),I,\idl(\mathbf{P}))$ is a canonical extension of the poset $\mathbf{P}$. In addition, the stable-sets operation $\Ra$, defined by $A\Ra C=(A\ltright C')'$ is the $\pi$-extension of the implication operation $\ra$ in $\mathbf{P}$. 
\end{Proposition}
\begin{proof}
The first claim was proven in \cref{compactness of delta1}, so we only argue that $\Ra$ is the $\pi$-extension of the implication operator $\ra$. 

Clarifying a point that might cause some confusion, we note that in the RS-frames approach maps are characterized by their {\em order types}, which is common in the literature on ordered algebras. To an $n$-ary map, a sequence in $\{1,\partial\}^n$ is associated, with the convention that a $1$ in the $k$-th place means that the map is monotone at the $k$-th argument place and a $\partial$ in that place means that the map is antitone in that place. Hence implication has the order-type $(\partial, 1)$, i.e. in this tradition it is regarded as a map $f:\mathbf{P}^\partial\times\mathbf{P}\lra\mathbf{P}$. This is significant because it affects what is to be considered as a closed or open element in the definition of $\sigma$ and $\pi$-extensions in \cite{mai-harding}.

Recalling definitions, the $\pi$-extension $f^\pi$ of a monotone unary map $f:\mathbf{P}\lra\mathbf{P}$ (and analogously for an $n$-ary map of any given order type) is defined on an open element $q$ of the canonical extension $\mathbf{P}^\sigma$ of $\mathbf{P}$ by $f^\pi(q)=\bigvee\{f(a)\midsp a\in P\mbox{ and }a\leq q\}$, where $P$ is viewed, for simplicity, as a subset of $P^\sigma$. It is then extended to any element $p$ of the canonical extension by setting $f^\pi(p)=\bigwedge\{f^\pi(q)\midsp p\leq q \mbox{ and } q \mbox{ an open element}\}$.

Implication has the order-type $(\partial,1)$ so that an open element in $(\mathbf{P}^\sigma)^\partial\times\mathbf{P}^\sigma$ is a pair $(\Gamma x, \lperp\{v\})$ of a closed element $\Gamma x$ of $\mathbf{P}^\sigma$ (which is an open element of its dual, $(\mathbf{P}^\sigma)^\partial$) and an open element $\lperp\{v\}$ of $\mathbf{P}^\sigma$. The definition of $\ra^\pi $, after working out details, is then given by
\begin{equation}
\label{pi extension of implication}
\begin{array}{ccl}
  \Gamma x\ra^\pi\lperp\{y\} &=& \bigvee\{\Gamma x_{a\ra b}\midsp a\in x\mbox{ and }b\in y\} \\
  A \ra^\pi C &=& \bigcap\{\Gamma x\ra^\pi \lperp\{y\}\midsp x\in A\mbox{ and }C\upv y\}.
  \end{array}
\end{equation}
It was shown in the course of the soundness proof, \cref{lambda t soundness}, that for any stable sets $A,C$, we have that
$A\Ra C=(A\ltright C')'=\{u\in W_1\midsp\forall x\in W_1\forall v\in W_\partial(x\in A\wedge C\upv v\lra uT'xv)\}$. Hence, $A\Ra C=\bigcap_{C\upv v}^{x\in A}T'xv$. Given the definition \eqref{x-tright-v} of the ideal $x{\tright}v$ in the canonical frame (see Lemma~\ref{tright lemma}), it was shown in Proposition~\ref{F1 is canonical} for the canonical relation $T$ that $yTxv$ iff $(x{\tright}v)\subseteq y$, so that $Txv=\Gamma(x{\tright}v)$ and then $T'xv=\lperp\{x{\tright}v\}$. We then obtain that $A\Ra C=\bigcap\{\lperp\{x{\tright}v\}\midsp{C\upv v}\mbox{ and }{x\in A}\}$. 

We claim that the identity $\Gamma x{\ra^\pi} \lperp\{v\}=\lperp\{x{\tright}v\}$ holds in the canonical frame. Given the definition of $\ra^\pi $ and taking duals, this is equivalent to showing that $\Gamma(x{\tright}v)=\bigcap^{a\in x}_{b\in y}\Gamma y_{a\ra b}$, recalling that $\{x_{a\ra b}\}\rperp=\Gamma y_{a\ra b}$. It was shown in \cref{F1 is canonical} that $\Gamma(x{\tright}v)=Txv$ where, by \cref{canonical frame defn}, $yTxv$ iff $\forall a,b(a\in x\mbox{ and }b\in v \lra (a\ra b)\in y)$. Since $(a\ra b)\in y$ iff $y_{a\ra b}\subseteq y$ iff $y\in \Gamma y_{a\ra b}$ it follows that $\Gamma(x{\tright}v)=\bigcap^{a\in x}_{b\in y}\Gamma y_{a\ra b}$, as needed.

Conclude that the operation $\Ra$, defined by $A\Ra C=(A\ltright C')'$, satisfies the identities~\eqref{pi extension of implication} defining the $\pi$-extension of implication in the canonical frame. This means that $(\gpsi,\subseteq,W_1,\Ra)$ is a canonical extension of the implicative poset $\mathbf{P}=(P,\leq,1,\ra)$.
\end{proof}

\subsection{The Unit as a Left Identity Element}
We only assumed so far that the distinguished element $\type{e}$ of the implicative poset satisfies the axiom $a\leq b$ iff $\type{e}\leq a\ra b$, characterizing the order in the poset by means of implication. In substructural logics $\type{t}$ is an identity element (a unit element) for a product operator $\circ$ residuated with implication. Being a left identity element for $\circ$ can be expressed entirely in terms of implication alone, due to residuation, as the law $a=\type{e}\ra a$, corresponding to the two logic axioms $\varphi\proves\type{t}\rfspoon\varphi$ and $\type{t}\rfspoon\varphi\proves\varphi$. 

We identify three classes $\mathbb{PU}\ell_*\subseteq\mathbb{PU}\ell^*\subseteq\mathbb{PU}\ell$ with respect to which soundness and completeness of the logic $\mathbf{\Lambda}_t+\{\varphi\proves\type{t}\rfspoon\varphi\}+\{\type{t}\rfspoon\varphi\proves\varphi\}$ can be proven. 

Table~\ref{substructural implicative poset wth unit frame axioms} presents the axiomatization of the largest frame class $\mathbb{PU}\ell$.

\begin{table}[!htbp]
\caption{$\mathbb{PU}\ell$ Frame Class axioms for the Logic of Implicative Posets with Left-Unit}
\label{substructural implicative poset wth unit frame axioms}
\hrule
\begin{enumerate}
\item[(F1)] For all $x\in W_1$ and all $v\in W_\partial$, the section $Txv\subseteq W_\partial$ of the relation $T^{\partial 1\partial}$ is a Galois set.
\\[1mm] \underline{Unit Axioms}
\item[(U)] $U$ is a Galois stable set such that for all $x\in W_1$ and all $v\in W_\partial$, $x\upv v$ holds iff for all $u$ in $U$, $uT'xv$ holds. 
\item[(U1)] For all $y,v\in W_\partial$ and $x\in W_1$, if $yT^{\partial 1\partial}xv$ and $x\in U$, then $v\leq y$.
\item[(U2)] $\forall y{\in} W_\partial\forall x{\in} W_1[xIy\lra\exists v{\in} W_\partial(xIv\wedge \exists x_1{\in} W_1\exists v_1{\in} W_\partial(vT^{\partial 1\partial}x_1v_1\wedge y\leq v_1\wedge x_1\in U))]$
\end{enumerate}
\hrule
\end{table}
Axiom (U1) validates the inclusion $U\ltright B\subseteq B$, for any co-stable set $B\in\gphi$ and thereby it validates the logic axiom $p\proves\type{t}\rfspoon p$. Axiom (U2) only validates the inclusion $B\subseteq (U\ltright B)''$, for $B\in\gphi$, which is sufficient in order to validate the axiom $\type{t}\rfspoon p\proves p$. This frame axiom results by calculating the first-order local correspondent of the axiom $\type{t}\rfspoon p\proves p$, using the generalized Sahlqvist -- van Benthem algorithm presented in \cite{dfmlC}, which uses a translation and a co-translation of the logic into the language of its sorted modal companion logic, originally presented in \cite{pll7,redm}. 

The question arises whether the logic can be proven sound and complete in the smaller frame class $\mathbb{PU}\ell^*$ which is axiomatized in \cref{substructural up star implicative poset wth unit frame axioms} so as to validate the inclusion $B\subseteq U\ltright B$, rather than only the weaker inclusion $B\subseteq(U\ltright B)''$ validated in $\mathbb{PU}\ell$.

\begin{table}[!htbp]
\caption{$\mathbb{PU}\ell^*$ Frame Class axioms for the Logic of Implicative Posets with Left-Unit}
\label{substructural up star implicative poset wth unit frame axioms}
\hrule
\begin{enumerate}
\item[(F1)] For all $x\in W_1$ and all $v\in W_\partial$, the section $Txv\subseteq W_\partial$ of the relation $T^{\partial 1\partial}$ is a Galois set.
\\[1mm] \underline{Unit Axioms}
\item[(U)] $U$ is a Galois stable set such that for all $x\in W_1$ and all $v\in W_\partial$, $x\upv v$ holds iff for all $u$ in $U$, $uT'xv$ holds. 
\item[(U1)] For all $y,v\in W_\partial$ and $x\in W_1$, if $yT^{\partial 1\partial}xv$ and $x\in U$, then $v\leq y$.
\item[(U$^*$2)]  $\forall y\in W_\partial\exists x\in W_1\exists v\in W_\partial(x\in U\wedge yT^{\partial 1\partial}xv\wedge y\leq v)$
\end{enumerate}
\hrule
\end{table}

In the duality argument of \cite{duality2} several additional frame axioms are assumed, including an axiom (M) on the monotonicity properties of frame relations which in the case at hand is the axiom that for any $y\in W_\partial$ the binary relation $yT^{\partial 1\partial}$ is decreasing in each argument place. In other words, if $yTxv$ holds and $z\leq x, w\leq v$, then also $yTzw$ obtains. 

If the frame axiom (M) is assumed, then the unit axiom (U$^*$2) simplifies further to
\begin{tabbing}
(U$_*$2)\hskip5mm\= $\forall y\in W_\partial\exists x\in W_1(x\in U\wedge yT^{\partial 1\partial}xy)$.
\end{tabbing} 
Let $\mathbb{PU}\ell_*$, with axiomatization displayed in \cref{substructural down star implicative poset wth unit frame axioms}, be the respective frame class.

\begin{table}[!htbp]
\caption{$\mathbb{PU}\ell_*$ Frame Class axioms for the Logic of Implicative Posets with Left-Unit}
\label{substructural down star implicative poset wth unit frame axioms}
\hrule
\begin{enumerate}
\item[(F1)] For all $x\in W_1$ and all $v\in W_\partial$, the section $Txv\subseteq W_\partial$ of the relation $T^{\partial 1\partial}$ is a Galois set.
\item[(M)] For any $y\in W_\partial$ the binary relation $yT^{\partial 1\partial}$ is decreasing in each argument place
\\[1mm] \underline{Unit Axioms}
\item[(U)] $U$ is a Galois stable set such that for all $x\in W_1$ and all $v\in W_\partial$, $x\upv v$ holds iff for all $u$ in $U$, $uT'xv$ holds. 
\item[(U1)] For all $y,v\in W_\partial$ and $x\in W_1$, if $yT^{\partial 1\partial}xv$ and $x\in U$, then $v\leq y$.
\item[(U$_*$2)]  $\forall y\in W_\partial\exists x\in W_1(x\in U\wedge yT^{\partial 1\partial}xy)$
\end{enumerate}
\hrule
\end{table}

For the soundness proof, we use the correspondence algorithm published in \cite{vb,dfmlC}. The algorithm relies on a translation of the language of the logic into the language of its modal companion. A brief review for the modal translation, adapted to the present case, is presented in the next section. \cref{sahlqvist section} gives a brief review of the correspondence algorithm.

\subsection{Modal Translation and Correspondence}
In \cite{dfmlC} a generalization of the classical Sahlqvist - Van Benthem correspondence result was presented, applied to distribution-free modal logics with implication and negation operations. The underlying lattice structure, when it exists, affects the ways in which sentences can be modally translated but it is not relevant to the correspondence argument per se and we adapt it here to the case of the logics of implicative posets. The argument is based on a fully abstract translation of the language of the logic of interest into the language of the sorted modal logic (its modal companion) which is the logic of the dual sorted powerset algebra of frames. In our case of interest, the sorted powerset algebra is the algebra $(\;)':\powerset(W_1)\leftrightarrows\powerset(W_\partial)^{\rm op}:(\;)'$ with the sorted image operator $\ltright:\powerset(W_1)\times\powerset(W_\partial)\lra \powerset(W_\partial)$ generated by the frame relation $T^{\partial 1\partial}$. 

The language $\mathcal{L}_s=(\mathcal{L}_1,\mathcal{L}_\partial)$ of the sorted companion modal logic of pure implicative logic is therefore defined by the syntax
\begin{eqnarray*}
\mathcal{L}_1\ni\alpha,\eta,\zeta &=& P_i\;(i\in\mathbb{N})\midsp \top\midsp\bot\midsp \alpha\cap\alpha\midsp\alpha\cup\alpha \midsp\type{u} \midsp\beta'
\\
\mathcal{L}_\partial\ni\beta,\delta,\xi &=& P^i(i\in\mathbb{N})\midsp\top\midsp\bot\midsp \beta\cap\beta\midsp\beta\cup\beta\midsp\alpha'\midsp \alpha\tright\beta. 
\end{eqnarray*}
Given a frame $\mathfrak{F}$, an interpretation of $\mathcal{L}_s$ is a sorted function $V=(V_1,V_\partial)$ assigning a subset $V_1(P_i)\subseteq Z_1$ and $V_\partial(P^i)\subseteq Z_\partial$ to the propositional variables of each sort.  A model $\mathfrak{M}$ is a pair $\mathfrak{M}=(\mathfrak{F},V)$ consisting of a frame and an interpretation of sorted propositional variables in the structure, as above.

Given a model $\mathfrak{M}$, its interpretation function generates a sorted satisfaction relation $\zmodels=(\models,\vmodels)$, where ${\models}\subseteq W_1\times\mathcal{L}_1$ and ${\vmodels}\subseteq W_\partial\times\mathcal{L}_\partial$, defined in Table~\ref{sorted sat table}. For $\alpha\in\mathcal{L}_1$ and $\beta\in\mathcal{L}_\partial$, we let $\val{\alpha}_\mathfrak{M}=\{x\in W_1\midsp x\models\alpha\}$ and $\yvval{\beta}_\mathfrak{M}=\{y\in W_\partial\midsp y\vmodels\beta\}$. If the model $\mathfrak{M}$ is understood from context, we omit the subscript.

\begin{table}[!htbp]
\caption{Sorted satisfaction relation, given a model $\mathfrak{M}=(\mathfrak{F},V)$}
\label{sorted sat table}
 ($u\in W_1, v\in W_\partial$)
 \hrule
\begin{tabbing}
$u\models P_i$\hskip7mm\=iff\hskip2mm\= $u\in V_1(P_i)$\hskip3.2cm\= $v\vmodels P^i$ \hskip6mm\=iff\hskip2mm\= $v\in V_\partial(P^i)$
\\[2mm]
$u\models\top$ \>iff\> $u=u$ \> $v\vmodels\top$\>iff\> $v=v$
\\[2mm]
$u\models\bot$ \>iff\> $u\neq u$ \> $v\vmodels\bot$ \>iff\> $v\neq v$
\\[2mm]
$u\models\beta'$ \>iff\> $\forall y(uIy\lra y\not\vmodels\beta)$ \>
$v\vmodels\alpha'$ \>iff\> $\forall x(xIv\lra x\not\models\alpha)$
\\[2mm]
$u\models\alpha\cap\eta$\>iff\> $u\models\alpha$  and $u\models\eta$   \>  $v\vmodels\beta\cap\delta$ \>iff\> $v\vmodels\beta$ and $v\vmodels\delta$
\\[2mm]
$u\models\alpha\cup\eta$ \>iff\> $u\models\alpha$ or $u\models\eta$
    \>
    $v\vmodels\beta\cup\delta$ \>iff\> $v\vmodels\beta$ or $v\vmodels\delta$
\\[2mm]
$u\models\type{u}$\>iff\> $u\in U$\hskip3cm $v\vmodels\alpha\tright\beta\;$ iff $\;\exists x\in Z_1\exists y\in Z_\partial(vT^{\partial 1\partial}xy\wedge x\models\alpha\wedge y\vmodels\beta)$
\end{tabbing}
\hrule
\end{table}

Proof-theoretic consequence is defined as a sorted relation ${\Ra} =(\proves,\vproves)$, with ${\proves}\subseteq\mathcal{L}_1\times\mathcal{L}_1$ and ${\vproves}\subseteq\mathcal{L}_\partial\times\mathcal{L}_\partial$.  We refer to sequents $\alpha\proves\eta$ as {\em 1-sequents} and to sequents $\beta\vproves\delta$ as {\em $\partial$-sequents}. Validity of a sequent in a model is defined as usual. A sorted proof system is defined in Table~\ref{sorted proof system table 1}, the soundness of which can be easily verified by the reader.

\begin{table}[t]
\caption{Sorted Proof System}
\hrule
\label{sorted proof system table 1}
A. Axioms and Rules for Both Sorts\\
(if $\sigma$ is $\alpha\in\mathcal{L}_1$, then $\Mapsto$ is $\proves$, and if $\sigma$ is $\beta\in\mathcal{L}_\partial$, then $\Mapsto$ is $\vproves$)
\begin{tabbing}
$\sigma\Mapsto\sigma$ \hskip2cm\= $\bot\Mapsto\sigma$\hskip3cm\= $\sigma\Mapsto\top$ \hskip2cm\= $\top\Mapsto\bot'$
\\
$\sigma_1\Mapsto\sigma_1\cup\sigma_2$  \> $\sigma_2\Mapsto\sigma_1\cup\sigma_2$  \>  $\sigma_1\cap\sigma_2\Mapsto\sigma_1$  \> $\sigma_1\cap\sigma_2\Mapsto\sigma_2$\\
$\sigma_1\cap(\sigma_2\cup\sigma_3)\Mapsto(\sigma_1\cap\sigma_2)\cup(\sigma_1\cap\sigma_3)$
\\[2mm]
$\infrule{\sigma_1\Mapsto\sigma_2}{\sigma_1[\sigma/P]\Mapsto\sigma_2[\sigma/P]}$ \>\> $\infrule{\sigma_1\Mapsto\sigma_2\hskip4mm\sigma_2\Mapsto\sigma_3}{\sigma_1\Mapsto\sigma_3}$
\\[2mm]
$\infrule{\sigma_1\Mapsto\sigma\hskip4mm \sigma_2\Mapsto\sigma}{\sigma_1\cup\sigma_2\Mapsto\sigma}$ \>\>
$\infrule{\sigma\Mapsto\sigma_1\hskip4mm\sigma\Mapsto\sigma_2}{\sigma\Mapsto\sigma_1\cap\sigma_2}$
\end{tabbing}
B. Axioms and Rules for the (sorted) Modal Operator $\tright$
\begin{tabbing}
$(\alpha\cup\eta)\tright\beta\vproves(\alpha\tright\beta)\cup(\eta\tright\beta)$ \hskip2cm\= $\infrule{\alpha\proves\eta\hskip4mm \beta\vproves\delta}{\alpha\tright\beta\vproves\eta\tright\delta}$
\\
$\bot\tright\beta\vproves\bot$ \hskip7mm $\alpha\tright\bot\vproves\bot$  \> $\alpha\tright(\beta\cup\delta)\vproves(\alpha\tright\beta)\cup(\alpha\tright\delta)$
\end{tabbing}
C. Axioms and Rules for Sorted Negation
\begin{tabbing}
$\alpha\proves\alpha''$\hskip2.5cm\= $\beta\vproves\beta''$\hskip2.3cm\=   $\infrule{\beta\vproves\delta}{\delta'\proves\beta'}$ \hskip2.3cm\= $\infrule{\alpha\proves\eta}{\eta'\vproves\alpha'} $                                         
\end{tabbing}
D.  Unit\hskip2.3cm
$\infrule{\alpha''\proves\alpha\hskip4mm\eta''\proves\eta\hskip4mm\alpha\proves\eta}{\type{u}\proves(\alpha\tright\eta')'}$
\\[2mm]
\begin{tabbing}
E. Left-Identity\hskip1cm\=
$\infrule{\beta''\vproves \beta}{\type{u}\tright \beta\vproves \beta}$\hskip1.8cm\= $\infrule{\beta''\vproves \beta}{\beta\vproves\type{u}\tright \beta}$
\end{tabbing}
\hrule
\end{table} 

Table~\ref{syntactic translation into sorted} defines by mutual recursion a syntactic translation $(\;)^\bullet$ and co-translation $(\;)^\circ$ of the language $\mathcal{L}$ of modal lattices into the language $\mathcal{L}_s=(\mathcal{L}_1,\mathcal{L}_\partial)$ of sorted modal logic.

\begin{table}[!htbp]
\caption{Translation and co-translation of the language of implicative posets}
\label{syntactic translation into sorted}
\hrule
\begin{tabbing}
\hspace*{1cm}\=$p_i^\bullet$ \hskip1.5cm\==\hskip4mm\= $P_i''$ \hskip3cm\= $p_i^\circ$\hskip1.5cm\==\hskip4mm\= $P_i'$\\
\> $\top^\bullet$\>=\>$\top$ \> $\top^\circ$\>=\> $\bot''$\\
\> $\type{t}^\bullet$ \>=\> $\type{u}$ \> $\type{t}^\circ$ \>=\> $\type{u}'$
\\
\> $\bot^\bullet$ \>=\> $\bot''$ \> $\bot^\circ$ \>=\> $\top$\\
\> $(\varphi\wedge\psi)^\bullet$ \>=\> $\varphi^\bullet\cap\psi^\bullet$ \> $(\varphi\wedge\psi)^\circ$ \>=\> $(\varphi^\circ\cup\psi^\circ)''$\\
\> $(\varphi\vee\psi)^\bullet$ \>=\> $(\varphi^\bullet\cup\psi^\bullet)''$ \> $(\varphi\vee\psi)^\circ$ \>=\> $\varphi^\circ\cap\psi^\circ$
\\
\> $(\varphi\rfspoon\psi)^\bullet$ \>=\> $(\varphi^\bullet\tright\psi^\circ)'$ \> $(\varphi\rfspoon\psi)^\circ$ \>=\> $(\varphi^\bullet\tright\psi^\circ)''$ 
\\[2mm]
\> Translation and Co-translation of Sequents
\\
\> $(\varphi\proves\psi)^\bullet$ \>=\> $\varphi^\bullet\proves\psi^\bullet$
\>
$(\varphi\proves\psi)^\circ$ \>=\> $\psi^\circ\vproves\varphi^\circ$
\end{tabbing}
\hrule
\end{table}

\begin{Remark}
In \cite{dfnmlA}, an alternative, but equivalent translation $(\varphi\rfspoon\psi)^\bullet=\varphi^\bullet\rightspoon\psi^\bullet$ was given, where $\rightspoon$ is an implication operation in the language of sorted modal logic. This move is not available in the context of frames for partial orders and Remark~\ref{Ra as restriction rem} explained why. It does become available however in out treatment of the Lambek calculus, see \cref{lambek section}.
\end{Remark}

Given a model $\mathfrak{M}=(\mathfrak{F},V)$ for $\mathcal{L}_s$, a model $\mathcal{N}=(\mathfrak{F},\bar{V})$ for the language $\mathcal{L}$ of implicative posets is obtained by setting $\bar{V}^1(p_i)=V_1(P_i)''$, generating an interpretation and a co-interpretation of $\mathcal{L}$-sentences.

A sentence $\alpha\in\mathcal{L}_1$ is a {\em classical modal correspondent} of a sentence $\varphi\in\mathcal{L}$ iff for any $\mathcal{L}_s$-model $\mathfrak{M}=(\mathfrak{F},V)$, $\val{\alpha}_\mathfrak{M}=\val{\varphi}_\mathfrak{N}$, where $\mathfrak{N}$ is defined as above. Modal correspondents of sequents $\varphi\proves\psi$ are typically rules whose premisses are stability assumptions. For example, the weakening sequent $\varphi\proves\psi\rfspoon\varphi$ (or, $\psi\circ\varphi\proves\varphi$, when the residual $\circ$ is included) corresponds to the rule $\infrule{\eta''\proves\eta}{\alpha\odot\eta\proves\eta}$ in the companion modal logic (see \cref{adding section} and \cref{corr table}).

\begin{Theorem}[Full Abstraction]
\label{full abstraction of trans in sorted modal}
Let $\mathfrak{M}=(\mathfrak{F},V)$ be a model of the sorted modal language $\mathcal{L}_s=(\mathcal{L}_1,\mathcal{L}_\partial)$. Then, for any sentence $\varphi\in\mathcal{L}$,
\begin{enumerate}
\item  its translation $\varphi^\bullet$ is a classical modal correspondent of $\varphi$. In other words, $\val{\varphi^\bullet}_\mathfrak{M}=\val{\varphi}_\mathfrak{N}= \val{(\varphi^\circ)'}_\mathfrak{M}=\val{(\varphi^\bullet)''}_\mathfrak{M}$
\item  $\yvval{\varphi^\circ}_\mathfrak{M}=\yvval{\varphi}_\mathfrak{N}= \yvval{(\varphi^\bullet)'}_\mathfrak{M}=\yvval{(\varphi^\circ)''}_\mathfrak{M}$
\item for any sequent $\varphi\proves\psi$ in the language $\mathcal{L}$ of the logic of implicative posets   $\mathfrak{M}\models\varphi^\bullet\proves\psi^\bullet$ iff $\mathfrak{N}\models\varphi\proves\psi$ iff $\mathfrak{M}\vmodels\psi^\circ\vproves\varphi^\circ$,
\end{enumerate}
where $\mathfrak{N}$ is defined as above, by setting $\bar{V}^1(p_i)=V_1(P_i)''$.
\end{Theorem}
\begin{proof}
The translation and co-translation are special instances of the case of the languages and logics of arbitrary normal lattices expansions. A proof of all three claims for the general case was given in \cite[Theorem~3.2]{vb}.
\end{proof}

The standard translation of sorted modal logic into sorted  FOL  is exactly as in the single-sorted case, except for the relativization to two sorts, displayed in Table \ref{std-trans}, where $\stx{u}{}, \sty{v}{}$ are defined by mutual recursion and $u,v$ are individual variables of sort $1,\partial$, respectively.

\begin{table}[!htbp]
\caption{Standard Translation of the sorted modal language $\mathcal{L}_s=(\mathcal{L}_1,\mathcal{L}_\partial)$}
\label{std-trans}
($u$ a sort-1 variable, $v$ a sort-$\partial$ variable)
\hrule
\begin{tabbing}
$\stx{u}{P_i}$\hskip0.8cm\==\hskip2mm\= ${\bf P}_i(u)$ \hskip3.5cm\= $\sty{v}{P^i}$\hskip0.8cm\==\hskip2mm\= ${\bf P}^i(v)$\\
$\stx{u}{\top}$ \>=\> $u=u$ \> $\stx{v}{\top}$\>=\> $v=v$\\
$\stx{u}{\bot}$ \>=\> $u\neq u$ \> $\sty{v}{\bot}$ \>=\> $v\neq v$\\
$\stx{u}{\alpha\cap \eta}$\>=\> $\stx{u}{\alpha}\cap\stx{u}{\eta}$\> $\sty{v}{\beta\cap \delta}$\>=\> $\sty{v}{\beta}\cap\sty{v}{\delta}$\\
$\stx{u}{\alpha\cup \eta}$\>=\> $\stx{u}{\alpha}\cup\stx{u}{\eta}$\> $\sty{v}{\beta\cup \delta}$\>=\> $\sty{v}{\beta}\cup\sty{v}{\delta}$\\
$\stx{u}{\beta'}$ \>=\> $\forall^\partial v\;(u{\bf I}v\;\lra\;\neg\sty{v}{\beta})$ \> $\sty{v}{\alpha'}$ \>=\> $\forall^1 u\;(u{\bf I}v\;\lra\;\neg\stx{u}{\alpha})$
\\
$\stx{u}{\type{u}}$\>=\> $\mathbf{U}(u)$\hskip3cm $\sty{v}{\alpha\tright\beta}$ \;=\; $\exists^1x\exists^\partial y(v\mathbf{T}^{\partial 1\partial}xy\wedge\stx{x}{\alpha}\wedge\sty{y}{\beta})$
\end{tabbing}
\hrule
\end{table}

If $\alpha(Q_{i_1},\ldots,Q_{i_n})$, where for each $j$, $Q_{i_j}\in\{P_{i_j}, P^{i_j}\}$, is an $\mathcal{L}_1$-sentence with propositional variables among the $Q_{i_j}$, then its second-order translation is defined to be the sentence $\mathrm{ST}^2_{x}(\alpha)=\forall Q_{i_1}\ldots\forall Q_{i_n}\forall^1 u\;\stx{u}{\alpha}$. It is understood that $\forall Q_{i_j}$ is $\forall^1 P_{i_j}$, if $Q_{i_j}=P_{i_j}\in\mathcal{L}_1$ and it is $\forall^\partial P^{i_j}$ otherwise. Similarly for $\beta(Q_{i_1},\ldots,Q_{i_n})$ and $\mathrm{ST}^2_v(\beta)=\forall Q_{i_1}\ldots\forall Q_{i_n}\forall^\partial v\sty{v}{\beta}$.

\begin{Proposition}
\label{std trans prop}
For any sorted modal formulae $\alpha,\beta$ (of sort $1, \partial$, respectively), any model $\mathfrak{M}=(\mathfrak{F},V)$ for $\mathcal{L}^1_s$  and any $x\in Z_1, y\in Z_\partial$, $\mathfrak{F}\models\mathrm{ST}^2_{u}(\alpha)[u:=x][Q_{i_j}:=V(Q_{i_j})]_{j=1}^n$ iff $\mathfrak{M},x\models \alpha$ iff $\mathfrak{F}\models\stx{u}{\alpha}[u:=x]$, where $x=V(u)$.

Similarly, $\mathfrak{M},y\vmodels \beta$ iff $\mathfrak{F}\vmodels\sty{v}{\beta}[v:=y]$ iff $\mathfrak{F}\vmodels\mathrm{ST}_v^2(\beta)[v:=y][Q_{i_j}:=V(Q_{i_j})]_{j=1}^n$.
\end{Proposition}

\begin{Corollary}
A sequent $\alpha\proves\eta$ in the sorted modal logic corresponds to the implication $\stx{u}{\alpha}\ra\stx{u}{\eta}$. In other words,  for any model $\mathfrak{M}=(\mathfrak{F},V)$, $\mathfrak{M}\models\alpha\proves\eta$ iff $\val{\alpha}_\mathfrak{M}\subseteq\val{\eta}_\mathfrak{M}$ iff $\mathfrak{F}\models(\stx{u}{\alpha}\ra\stx{u}{\eta})[V]$.
\end{Corollary}

\subsection{Soundness and Completeness}
\begin{Proposition}
\label{unit completeness}
The extension of the logic $\mathbf{\Lambda}_t$ with the left-unit axioms $\varphi\proves\type{t}\rfspoon\varphi$ and $\type{t}\rfspoon\varphi\proves\varphi$ is sound and complete in the class $\mathbb{PU}\ell$ of frames axiomatized in Table~\ref{substructural implicative poset wth unit frame axioms}.
\end{Proposition}
\begin{proof}
For soundness, it suffices to verify that axioms (U1) and (U2) are returned as the first-order correspondents by the generalized Sahlqvist -- van Benthem algorithm. See \cref{sahlqvist section} for a brief review and \cite{dfmlC} for details.
\paragraph{Case $p\proves\type{t}\rfspoon p$.}
The co-translation returns the formal inequality $(\type{u}\tright P')''\leq_\partial P'$ and by reduction we obtain the formal system $\langle Q=_\partial P'\midsp \type{u}\tright Q\leq_\partial Q\rangle$. Letting $\type{t}-\mathrm{INV}_Q$ be the $\type{t}$-invariance constraint for the predicate variable $Q$, the guarded second-order translation is
\[
\forall^\partial Q\forall^\partial y[\type{t}-\mathrm{INV}_Q\wedge\exists^1x\exists^\partial v(\mathbf{U}(x)\wedge Q(v)\wedge y\mathbf{T}xv)\lra Q(y)]
\]
which by the standard steps taken in the algorithm return the sentence \[\forall^\partial y\forall^1x\forall^\partial v(\mathbf{U}(x)\wedge y\mathbf{T}xv\lra v\leq y),\] which is just the axiom (U1).

\paragraph{Case $\type{t}\rfspoon p\proves p$.}
The co-translation yields the formal inequality $P'\leq_\partial (\type{u}\tright P')''$ which reduces to $\langle Q=_\partial P'\midsp Q\leq_\partial (\type{u}\tright Q)''$. We trust the reader to write the guarded second-order translation and run the remaining steps of the correspondence algorithm to verify that the first-order correspondent returned is just the axiom (U2).

By the above, $\mathbf{\Lambda}_t+\{p\proves\type{t}\rfspoon p\}+\{\type{t}\rfspoon p\proves p\}$ is sound in the class $\mathbb{PU}\ell$ axiomatized in Table~\ref{substructural implicative poset wth unit frame axioms}.

For completeness, in the canonical frame we have defined $U=\Gamma x_\type{e}$, where $\type{e}=[\type{t}]$ is the equivalence class of the truth constant $\type{t}$. 

For axiom (U1), assume $x\in U$ and $yTxv$, where $x$ is a filter and $y,v$ are ideals. This is equivalent to $\type{e}\in x$ and $x\tright v\subseteq y$. For any $b\in v$ the implication element $\type{e}\ra b$ is in $x\tright v$, by definition of the point operator $\tright$. Hence $\type{e}\ra b\in y$. Since the axiom $p\proves\type{t}\rfspoon p$ is assumed in the logic we obtain $b\leq(\type{e}\ra b)\in y$, hence $b\in y$. This means $v\subseteq y$ and so the canonical frame satisfies axiom (U1).

By reduction, axiom (U2) is semantically equivalent to the rule $\infrule{Q''\vproves Q}{Q\vproves (\type{u}\tright Q)''}$ in the sorted modal logic. Therefore the frame axiom (U2) holds in the canonical frame if it can be shown that for any Galois set $B$ of ideals, $B\subseteq\Gamma x_\mathrm{e}\Mtright B =(\Gamma x_\mathrm{e}\ltright B)''$. We verify that the stronger fact $B\subseteq \Gamma x_\mathrm{e}\ltright B$ actually holds in the canonical frame.  

By definition,  $\Gamma x_\mathrm{e}\ltright B$ is the set of ideals $v$ such that for some $u\in\Gamma x_\mathrm{e}$ and $w\in B$ we have $yTuw$. The latter is equivalent to $u{\tright} w\subseteq y$. Given $y\in B$, the ideal $x_\mathrm{e}{\tright} y\in \Gamma x_\mathrm{e}\ltright B$ (obviously $(x_\mathrm{e}{\tright} y)Tx_\mathrm{e}y$). But $x_\mathrm{e}{\tright} y\subseteq y$. This is because if $a\leq \mathrm{e}\ra b\in x_\mathrm{e}{\tright} y$, for $b\in y$, then by $\mathrm{e}\ra b\leq b$ (because the axiom $\type{t}\rfspoon p\proves p$ is assumed in the logic) we obtain $a\leq b\in y$, and since $y$ is an ideal we obtain $a\in y$. Hence $x_\mathrm{e}{\tright}y\subseteq y$, i.e. $yTx_\mathrm{e}y$ holds. Hence $y\in\Gamma x_\mathrm{e}\ltright B$, which proves $B\subseteq \Gamma x_\mathrm{e}\ltright B$, as needed.
\end{proof}

\begin{Corollary}
By the proof of \cref{unit completeness},  the logic $\mathbf{\Lambda}_t+\{\varphi\proves\type{t}\rfspoon\varphi\}+\{\type{t}\rfspoon\varphi\proves\varphi\}$, extending $\mathbf{\Lambda}_t$ with the left-unit axioms, is sound and complete in the smaller class $\mathbb{PU}\ell^*$ of frames, axiomatized in \cref{substructural up star implicative poset wth unit frame axioms}, which is the same set of axioms as in \cref{substructural implicative poset wth unit frame axioms} but replacing axiom (U2) by the stronger axiom  $\forall y\in W_\partial\exists x\in W_1\exists v\in W_\partial(x\in U\wedge yT^{\partial 1\partial}xv\wedge y\leq v)$, 
which corresponds to the fact that for any co-stable set $B$, the inclusion $B\subseteq U\ltright B$ obtains.
\end{Corollary}

\begin{Corollary}
The canonical relation $T^{\partial 1\partial}$, satisfying $yTxv$ iff $x{\tright}v\subseteq y$ is decreasing in the second and third argument place. Thereby  $\forall y\in W_\partial\exists x\in W_1\exists v\in W_\partial(x\in U\wedge yT^{\partial 1\partial}xv\wedge y\leq v)$ reduces, equivalently, to the axiom  $\forall y\in W_\partial\exists x\in W_1(x\in U\wedge yT^{\partial 1\partial}xy)$. 

Hence, by the proof of \cref{unit completeness}, the logic $\mathbf{\Lambda}_t+\{\varphi\proves\type{t}\rfspoon\varphi\}+\{\type{t}\rfspoon\varphi\proves\varphi\}$ is sound and complete in the smaller frame class $\mathbb{PU}\ell_*$.
\end{Corollary}

It is established knowledge that every implicative poset (integral, or with a unit element) is a fragment of a residuated implicative poset. For details, \cite{dunn-partial,choiceFreeHA} can be consulted. The logic of implicative residuated posets is the Lambek calculus, to which we turn in the next Section.

\section{The Lambek Calculus}
\label{lambek section}
\subsection{The Non-Associative Calculus}
A partially ordered residuated unital groupoid $\mathbf{P}=(P,\leq,1,\type{e},\la,\circ,\ra)$ is an implicative poset with two additional operations $\la$ and $\circ$ satisfying the residuation condition $a\leq c\la b$ iff $a\circ b\leq c$ iff $b\leq a\ra c$. Equivalently, the axioms $b\leq a\ra a\circ b$ and $a\circ(a\ra b)\leq b$ hold, for residuation of $\circ$ with $\ra$, and similarly the axioms $(b\la a)\circ a\leq b$ and $a\leq a\circ b\la b$ hold, for residuation of $\circ,\la$. In addition, the special element $\type{e}$ satisfies the left-right unit identities $a\circ\type{e}=a=\type{e}\circ a$. Note that it then follows that $a\leq b$ iff $\type{e}\leq a\ra b$ iff $\type{e}\leq b\la a$.

The logic of partially ordered residuated unital groupoids is the non-associative Lambek Calculus. The monotonicity types of the additional operators are $\delta(\circ)=(1,1;1)$ and $\delta(\la)=(\partial, 1;\partial)$. The language and proof system of the calculus are too familiar to need reviewing here. Given the monotonicity types of the operators, frames for the Lambek calculus in the generalized J\'{o}nsson-Tarski framework approach will have three ternary relations of corresponding sorts, as indicated, $S^{\partial\partial 1},R^{111},T^{\partial 1\partial}$, in accordance to \cref{corresponding frames}. Each generates a sorted image operator in the dual sorted powerset algebra of the frame, as shown below,  defined by instantiating~\eqref{sorted image ops},
\begin{align*}
\ltleft &:\powerset(W_\partial)\times\powerset(W_1)\lra\powerset(W_\partial) \\
\mbox{$\bigodot$}&:\powerset(W_1)\times\powerset(W_1)\lra\powerset(W_1)\\
\ltright &:\powerset(W_1)\times\powerset(W_\partial)\lra\powerset(W_\partial).
\end{align*}

For semantics, the (co)satisfaction relation in a model $\mathfrak{M}=(\mathfrak{F},V)$ is specified by instantiating the general definition in \cref{sat}, shown in \cref{sat-Lambek}  for the logical operators $\la,\circ,\rfspoon$ and $\type{t}$.

\begin{equation}\label{sat-Lambek}
\begin{array}{lcl}
u\forces\type{t} &\mbox{ iff }& u\in U\\
 z\forces\varphi\rfspoon\psi &\mbox{ iff }&\forall x\in W_1\forall v\in W_\partial(x\forces\varphi\wedge v\dforces\psi\lra zT'xv)\\
 v\dforces\varphi\circ\psi &\mbox{ iff }& \forall x,z\in W_1(x\forces\varphi\wedge z\forces\psi\lra vR'xz)\\
 x\forces\psi\lfspoon\varphi &\mbox{ iff }& \forall z\in W_1\forall v\in W_\partial(z\forces\varphi\wedge v\dforces\psi\lra xS'vz).
\end{array}
\end{equation}

The language of the companion  modal logic of the Lambek calculus extends the sorted language $\mathcal{L}_s$  to also include binary diamond operators $\tleft,\odot$ where for $\alpha,\eta\in\mathcal{L}_1$ and $\beta\in\mathcal{L}_\partial$ new sentences $\beta\tleft\alpha\in\mathcal{L}_\partial$ and $\alpha\odot\eta\in\mathcal{L}_1$ are added to the sorted language. 
\begin{eqnarray*}
\mathcal{L}_1\ni\alpha,\eta,\zeta &=& P_i\;(i\in\mathbb{N})\midsp \top\midsp\bot\midsp \alpha\cap\alpha\midsp\alpha\cup\alpha\midsp\beta' \midsp\alpha\odot\alpha
\\
\mathcal{L}_\partial\ni\beta,\delta,\xi &=& P^i(i\in\mathbb{N})\midsp\top\midsp\bot\midsp \beta\cap\beta\midsp\beta\cup\beta\midsp\alpha'\midsp \beta\tleft\alpha\midsp \alpha\tright\beta. 
\end{eqnarray*}

They are interpreted by  $\val{\alpha\odot\eta}_\mathfrak{M}=\val{\alpha}_\mathfrak{M}\bigodot\val{\eta}_\mathfrak{M}$ and $\yvval{\beta\tleft\alpha}_\mathfrak{M}=\yvval{\beta}_\mathfrak{M}\ltleft\val{\alpha}_\mathfrak{M}$. Explicitly, the semantic clauses of \cref{sorted sat table} are extended by adding the clauses for $\odot$ and $\tleft$ in \cref{sat for odot-tleft}
 
 \begin{equation}\label{sat for odot-tleft}
 \begin{array}{lcl}
 W_1\ni u\models\alpha\odot\eta &\mbox{ iff } &\exists z_1,z_2\in W_1(z_1\models\alpha\wedge z_2\models\eta \wedge uR^{111}z_1z_2)\\
 W_\partial\ni v\vmodels \beta\tleft\alpha &\mbox{ iff } &\exists y\in W_\partial\exists u\in W_1(y\vmodels\beta\wedge u\models\alpha\wedge vS^{\partial\partial 1}yu).
 \end{array}
 \end{equation}
 
The proof system of Table~\ref{sorted proof system table 1} of the sorted modal companion logic needs to be extended, but we defer details for after a discussion on the structure of frames and their dual sorted powerset algebras.

Letting $\Mtleft,\bigovert,\Mtright$ stand for the closure of the restriction of $\ltleft,\bigodot,\ltright$ to Galois sets, implication operations $\La,\Ra$ are defined by $C\La A=(C'\Mtleft A)'=(C'\ltleft A)'$ and $A\Ra C=(A\Mtright C')'=(A\ltright C')'$, as in~\eqref{2single-sorted}.

The class $\mathbb{LK}$ of frames $\mathfrak{F}=(s,W,I,U,S^{\partial\partial 1},R^{111},T^{\partial 1\partial})$ for the Lambek calculus is axiomatized  in Table~\ref{lambek frame class axioms}. Given residuation, there are alternative/equivalent ways to axiomatize frames so as to validate the unit axioms of the logic, involving either the relations $S,T$ of the frame, or the relation $R$. 

\begin{table}[!htbp]
\caption{Axioms for the Lambek Frame Class $\mathbb{LK}$}
\label{lambek frame class axioms}
\hrule
\begin{enumerate}
\item[(F1)]\hskip4mm For all $x,z\in W_1$ and all $v\in W_\partial$, the sections $Txv$, $Rxz$, $Svx$ of the relations $S,R,T$ are Galois sets.
\item[(RES)]\hskip4mm For all $x,z\in W_1$ and all $v\in W_\partial$, the equivalences $xS'vz$ iff $vR'xz$ iff $zT'xv$ hold, where
$S',R'$ and $T'$ are the Galois dual relations of the frame relations $S,R,T$.
\\[1mm] \underline{Unit Axioms}
\item[(U$^*$)]\hskip4mm $U$ is a Galois stable set.
\item[(F2.1)]\hskip4mm $\forall x,z_1,z_2\in W_1(xR^{111}z_1z_2\ra(z_1\in U\ra z_2\leq x) )$.
\item[(F2.2)]\hskip4mm $\forall x,z_1,z_2\in W_1(xR^{111}z_1z_2\ra(z_2\in U\ra z_1\leq x))$
\item[(F3.1)]\hskip4mm $\forall x\in W_1\forall y\in W_\partial(xIy\ra\exists z\in W_1(zIy\wedge\exists z_1,z_2\in W_1(z_1\in U\wedge zR^{111}z_1z_2\wedge x\leq z_2)))$.
\item[(F3.2)]\hskip4mm $\forall x\in W_1\forall y\in W_\partial(xIy\ra\exists z\in W_1(zIy\wedge\exists z_1,z_2\in W_1(z_2\in U\wedge zR^{111}z_1z_2\wedge x\leq z_1)))$.
\end{enumerate}
\hrule
\end{table}

\begin{Theorem}
\label{lambek completeness}
The non-associative Lambek calculus is sound and complete in the frame class $\mathbb{LK}$ axiomatized as in Table~\ref{lambek frame class axioms}.
\end{Theorem}
\begin{proof}
For soundness of the residuation axioms in the logic, given that for a Galois set $G$ and any set $X$, $X''\subseteq G$ iff $X\subseteq G$, it is immediate that the residuation condition in $\gpsi$, 
\[
\mbox{$A\subseteq C\La F$ iff $A\bigovert F\subseteq C$ iff $F\subseteq A\Ra C$},
\] 
is equivalent to the condition
\[
\mbox{$C'{\ltleft}F\subseteq A'$ iff $A\bigodot F\subseteq C$ iff $A{\ltright}C'\subseteq F'$.}
\]
We give details for the proof of the equivalence $A\bigodot F\subseteq C$ iff $A{\ltright}C'\subseteq F'$ and leave the similar proof of the equivalence $C'{\ltleft}F\subseteq A'$ iff $A\bigodot F\subseteq C$ to the interested reader.
\begin{tabbing}
$A\bigodot F\subseteq C$\hskip3mm\= iff\hskip2mm\= $\forall u,x,z(x\in A\wedge z\in F\wedge uRxz\lra u\in C)$\hskip6mm\=\\
\>iff\> $\forall x,z(x\in A\wedge z\in F\lra Rxz\subseteq C)$ \> $Rxz$ is a Galois set\\
\>iff\> $\forall x,z(x\in A\wedge z\in F\lra  C'\subseteq R'xz)$\\
\>iff\> $\forall y,x,z(x\in A\wedge z\in F\wedge C\upv y\lra yR'xz)$ \> Using axiom (RES)\\
\>iff\> $\forall y,x(x\in A\wedge  C\upv y\lra F\subseteq T'xy)$ \> $Txy$ is a Galois set\\
\>iff\> $\forall y,x(x\in A\wedge  C\upv y\lra Txy\subseteq F')$\\
\>iff\> $\forall y,x,v(x\in A\wedge  C\upv y\wedge vTxy\lra F\upv v)$\\
\>iff\> $A{\ltright}C'\subseteq F'$.
\end{tabbing}
For the unit axioms, (F2.1) validates the axiom $\type{t}\circ p\proves p$ and similarly for (F2.2) and the axiom $p\circ\type{t}\proves p$. This is established by calculating the first-order correspondents of the axioms, using the generalized Sahlqvist -- van Benthem algorithm of \cref{sahlqvist section}. Similarly for (F3.1) and (F3.2) which are the first-order correspondents of the axioms $p\proves \type{t}\circ p$ and $p\proves p\circ\type{t}$, respectively. We leave details to the reader.

For completeness, the canonical frame is as in Definition~\ref{canonical frame defn}, except for the addition of two relations $S^{\partial\partial 1},R^{111}$ and for the distinguished subset $U=\Gamma x_\mathrm{e}$, where $\type{e}=[\type{t}]$ as in 
the amendment of the canonical frame for implication in \cref{posets section}. The canonical relations $S,R$ are defined as in the relational representation of operators in Boolean algebras \cite{jt1}, by
\begin{align}\label{S-R-canonical}
  ySvx & \mbox{ iff } \forall a,b(a\in x\wedge b\in v\lra (b\la a)\in y)\\
  uRxz & \mbox{ iff }\forall a,b(a\in x\wedge b\in z\lra a\circ b\in u).
\end{align}
Let $\tright$ be the point operator defined in~\eqref{x-tright-v}, where for a filter $x$ and ideal $v$, $x{\tright}v$ is an ideal, by Lemma~\ref{tright lemma}. It is useful to consider also point operators $\tleft,\medcircle$ where for an ideal $v$ and filters $x,z$ we set
\begin{equation}\label{x-tleft-v}
  v{\tleft}x=\{e\midsp\exists a\in x\exists b\in v\; e\leq b\la a\}\hskip1cm x\medcircle z=\{e\midsp\exists a\in x\exists b\in z\; a\circ b\leq e\}.
\end{equation}
The reader can verify that $v{\tleft}x$ is an ideal and $x\medcircle z$ is a filter. By an argument similar to that in Proposition~\ref{F1 is canonical} it is shown that $Syz=\Gamma(y{\tleft}z), Rxz=\Gamma(x\medcircle z)$ and $Txy=\Gamma(x{\tright}y)$, hence axiom (F1) holds in the canonical frame. Since we set $U=\Gamma x_\mathrm{e}$, a stable set, axiom (U$^*$) is valid in the canonical frame.

Given definitions, we have $yR'xz$ iff $x\medcircle z\upv y$, $zT'xy$ iff $z\upv x{\tright}y$ and $xS'yz$ iff $x\upv y{\tleft}z$. To verify that axiom (RES) holds in the canonical frame it suffices to show that
\[
x\upv y{\tleft}z\;\mbox{ iff }\; x\medcircle z\upv y\;\mbox{ iff }\; z\upv x{\tright}y.
\]
We give details for the proof of the second equivalence, leaving the first to the reader.

If $x\medcircle z\upv y$, let $a\in x, b\in z$ such that $a\circ b\in y$. By residuation, $b\leq a\ra a\circ b$ and since $b$ is in the filter $z$ we get $(a\ra a\circ b)\in z$, as well. But $a\in x$ and $a\circ b\in y$, hence $a\ra a\circ b\in x{\tright}y$. This means $z\upv x{\tright}y$.

Conversely, let $a\in x, b\in y$ be such that $(a\ra b)\in z$. Then $a\circ(a\ra b)\in x\medcircle z$, which is a filter and then since $a\circ(a\ra b)\leq b$, by residuation, we get $b\in x\medcircle z$. Hence $x\medcircle z\upv y$.

Conclude that the residuation frame axiom (RES) is valid in the canonical frame. 

Validity of (F2.1) and (F2.2) in the canonical frame is immediate. Indeed, let $x,z_1,z_2$ be filters. The assumption $xR^{111}z_1z_2$ in the canonical frame is equivalent to the inclusion $z_1\medcircle z_2\subseteq x$. If for $i=1,2$, $z_i\in U=\Gamma x_\mathrm{e}$, this means that $\mathrm{e}\in z_i$. Then for $i=1$, $e\in z_1$ and for any $b\in z_2$ we get $e\circ b\in z_1\medcircle z_2\subseteq x$, so $b=e\circ b\in x$, i.e. $z_2\subseteq x$. Similarly for $i=2$.

The frame axiom (F3.1) ensures validity of the axiom $p\proves \type{t}\circ p$, while (F3.2) ensures that of the axiom $p\proves p\circ\type{t}$. In the correspondence argument for the first, its translation reduces to $P\leq_1(\type{u}\odot P)''$, with the guard $P''\leq_1 P$, which is semantically equivalent to the inclusion $A\subseteq (U\bigodot A)''$, for any stable set $A\in\gpsi$. Similarly for (F3.2), ensuring that the inclusion $A\subseteq (A\bigodot U)''$, for $A\in\gpsi$, holds in the frame. We show that the canonical frame validates the stronger inclusions $A\subseteq U\bigodot A$ and $A\subseteq A\bigodot U$, where $U=\Gamma x_\mathrm{e}$.

If $x$ is a filter $x\in A$, then consider the filters $x_\mathrm{e}\medcircle x\in U\bigodot A$ and $x\medcircle x_\mathrm{e}\in A\bigodot U$. By definition, $x_\mathrm{e}\medcircle x$ is the filter generated by the elements $d\circ a$ with $\mathrm{e}\leq d$ and $a\in x$. Thus $a=\mathrm{e}\circ a\leq d\circ a\in x$, hence $x_\mathrm{e}\medcircle x\subseteq x$. Conversely, with $a\in x$ we obtain that $a=\mathrm{e}\circ a\in x_\mathrm{e}\medcircle x$. Conclude that $x_\mathrm{e}\medcircle x=x=x\medcircle x_\mathrm{e}$, from which it follows that  $A\subseteq U\bigodot A$ and $A\subseteq A\bigodot U$, where $U=\Gamma x_\mathrm{e}$.
\end{proof}

\begin{Proposition}
\label{sigma-pi in canonical}
The dual full complex algebra of the canonical frame (the dual frame of the Lindenbaum-Tarski algebra) of the non-associative Lambek calculus is a canonical extension of its Lindenbaum-Tarski algebra.
\end{Proposition}
\begin{proof}
That the filter-ideal frame is a canonical extension of the underlying poset of the Lindenbaum algebra was proven in \cref{delta1-lemma} and \cref{compactness lemma}. That $\Ra$ is the $\pi$-extension of the implication operation $\ra$ was also dealt with in \cref{pi extension lemma}. That $\La$ is the $\pi$-extension of $\la$ is proven by a completely analogous argument, left to the reader. It remains to argue that $\bigovert$ is the $\sigma$ extension of the Lambek product operator $\circ$. By residuation, $\bigovert$ distributes over arbitrary joins in each argument place, so that $A\bigovert C=\left(\bigvee_{x\in A}\Gamma x\right)\bigovert\left(\bigvee_{z\in C}\Gamma z\right)=\bigvee^{x\in A}_{z\in C}(\Gamma x\bigovert\Gamma z)$. It is then enough to argue that $\Gamma x\bigovert\Gamma z=\bigcap\{\Gamma x_{a\circ b}\midsp a\in x\mbox{ and }b\in z\}$. 

First $\bigcap\{\Gamma x_{a\circ b}\midsp a\in x\mbox{ and }b\in z\}=\{u\midsp\forall a,b(a\in x\mbox{ and }b\in z,\mbox{ then }a\circ b\in u\}=R^{111}xz$, by the way the relation $R^{111}$ was defined in the canonical frame. Furthermore,
\begin{tabbing}
 $\Gamma x\bigovert\Gamma z$ \hskip2mm\==\hskip2mm\= $(\Gamma x\bigodot\Gamma z)''=\left(\{u\midsp\exists x_1\exists z_1(x\leq x_1\wedge z\leq z_1\wedge uR^{111}x_1z_1)\}\right)''$\\
 \>=\> $\left(\{u\midsp\exists x_1\exists z_1(x\leq x_1\wedge z\leq z_1\wedge x_1\medcircle z_1\subseteq u)\}\right)''$\\
 \>=\> $\left(\{u\midsp x\medcircle z\subseteq u)\}\right)''=(\Gamma(x\medcircle z))''=\Gamma(x\medcircle z)=R^{111}xz$
 \end{tabbing}
 and this concludes the proof.
\end{proof}

\begin{Proposition}
\label{alternative for implication}
The left $\La$ and right $\Ra$ stable set operations in any frame $\mathfrak{F}\in\mathbb{LK}$ are the restrictions of the residuals ${}_R\hskip-1mm\La,\Ra_R$ of $\bigodot$ in the dual powerset algebra of the frame.

Consequently, the left $\lfspoon$ and right $\rfspoon$ implication connectives of the Lambek calculus can be equivalently modeled by the satisfaction clauses
\begin{tabbing}
$W_1\ni x\forces \varphi\rfspoon\psi$\hskip6mm\= iff\hskip4mm\=  $\forall u,w,z\in W_1(u\forces\varphi\wedge x\leq w\wedge zR^{111}uw\lra z\forces\psi)$  \\
$W_1\ni x\forces \psi\lfspoon\varphi$ \> iff \> $\forall u,w,z(x\leq w\wedge u\forces\varphi\wedge zR^{111}wu\lra z\forces\psi)$
\end{tabbing}
\end{Proposition}
\begin{proof}
The image operator $\bigodot$ is completely additive in each argument place. Hence it is residuated with operations ${}_R\hskip-1mm\La,\Ra_R$. They are respectively defined by 
\[
X\Ra_R Y=\bigcup\{Z\subseteq W_1\midsp X\mbox{$\bigodot$} Z\subseteq Y\}\hskip1cm Y{}_R\hskip-1mm\La X=\bigcup\{Z\subseteq W_1\midsp Z\mbox{$\bigodot$} X\subseteq Y\}
\]
By the soundness part in the proof of \cref{lambek completeness}, the closure $\bigovert$ of the restriction of $\bigodot$ to Galois (stable) sets is residuated with $\La,\Ra$. By \mcite{Theorem~3.14}{duality2}, $\La,\Ra$ are the restrictions of ${}_R\hskip-1mm\La,\Ra_R$ to Galois stable sets. By \mcite{Lemma~3.15}{duality2} they are equivalently defined by
\[
A\Ra C=\{x\in W_1\midsp A\mbox{$\bigodot$}\Gamma x\subseteq C\}\hskip1cm C\La A=\{x\in W_1\midsp \Gamma x\mbox{$\bigodot$}A\subseteq C\}
\]
Given the definition of $\bigodot$ it follows that
\begin{tabbing}
$x\in A\Ra C$ \hskip5mm\= iff\hskip3mm\= $\forall u,w,z(u\in A\wedge x\leq w\wedge zR^{111}uw\lra z\in C)$\\
$x\in C\La A$ \>iff\> $\forall u,w,z(x\leq w\wedge u\in A\wedge zR^{111}wu\lra z\in C)$.
\end{tabbing}
It follows that the satisfaction clauses for the implication constructs are as in the statement of the Proposition.
 \end{proof}

\begin{Proposition}
\label{star corollary}
By the proof of \cref{lambek completeness},  the non-associative Lambek calculus is sound and complete in the smaller class $\mathbb{LK}^*$ of frames, see \cref{lambek star frame class axioms}, axiomatized as in \cref{lambek frame class axioms} but replacing axioms (F3.1) and (F3.2) by the stronger axioms
\begin{tabbing}
(F3$^*$.1) \hskip5mm\= $\forall x\in W_1\exists z_1,z_2\in W_1(z_1\in U\wedge xR^{111}z_1z_2\wedge x\leq z_2)$ \\
(F3$^*$.2) \> $\forall x\in W_1\exists z_1,z_2\in W_1(z_2\in U\wedge xR^{111}z_1z_2\wedge x\leq z_1)$ 
\end{tabbing}
corresponding to the fact that both inclusions $A\subseteq U\bigodot A$ and $A\subseteq A\bigodot U$ hold in the frame, for any stable set $A\in\gpsi$.
\end{Proposition}

\begin{table}[!htbp]
\caption{Axioms for the Lambek Frame Class $\mathbb{LK}^*$}
\label{lambek star frame class axioms}
\hrule
\begin{enumerate}
\item[(F1)] For all $x,z\in W_1$ and all $v\in W_\partial$, the sections $Txv$, $Rxz$, $Svx$ of the relations $S,R,T$ are Galois sets.
\item[(RES)] For all $x,z\in W_1$ and all $v\in W_\partial$, the equivalences $xS'vz$ iff $vR'xz$ iff $zT'xv$ hold, where\\
$S',R'$ and $T'$ are the Galois dual relations of the frame relations $S,R,T$.
\\[1mm] \underline{Unit Axioms}
\item[(U$^*$)] $U$ is a Galois stable set.
\item[(F2.1)] $\forall x,z_1,z_2\in W_1(xR^{111}z_1z_2\ra(z_1\in U\ra z_2\leq x) )$.
\item[(F2.2)] $\forall x,z_1,z_2\in W_1(xR^{111}z_1z_2\ra(z_2\in U\ra z_1\leq x))$
\item[(F3$^*$.1)]\hskip2mm $\forall x\in W_1\exists z_1,z_2\in W_1(z_1\in U\wedge xR^{111}z_1z_2\wedge x\leq z_2)$.
\item[(F3$^*$.2)]\hskip2mm $\forall x\in W_1\exists z_1,z_2\in W_1(z_2\in U\wedge xR^{111}z_1z_2\wedge x\leq z_1)$ .
\end{enumerate}
\hrule
\end{table}

For duality purposes \cite{duality2,choiceFreeHA,choiceFreeStLog}, additional frame axioms are assumed. They include the axiom on monotonicity properties of frame relations, referred to in the definition of the frame class $\mathbb{PU}\ell_*$  and stated as follows:
\begin{itemize}
\item[(M)] For any frame relation $R_j$ of sort $(i_{n(j)+1}:i_1\cdots i_{n(j)})$ and any $z\in W_{n(j)+1}$, the $n(j)$-ary relation $zR_j$ is decreasing in every argument place. 
\end{itemize}
The axiom is useful in many cases for simplifying conditions, such as correspondence conditions for various sentential axioms. Assuming this axiom, the satisfaction clauses in the statement of the Proposition simplify to the following
\begin{tabbing}
$W_1\ni x\forces \varphi\rfspoon\psi$\hskip6mm\= iff\hskip4mm\=  $\forall u,z\in W_1(u\forces\varphi\wedge  zR^{111}ux\lra z\forces\psi)$  \\
$W_1\ni x\forces \psi\lfspoon\varphi$ \> iff \> $\forall u,w( u\forces\varphi\wedge zR^{111}xu\lra z\forces\psi)$,
\end{tabbing}
which were used as the alternative modeling of implication in \cite{dfnmlA}.

Assuming the frame axiom (M) simplifies the statement of frame axioms in the restricted frame class $\mathbb{LK}_*\subseteq\mathbb{LK}^*\subseteq\mathbb{LK}$. 

\begin{table}[!htbp]
\caption{Frame axioms for the class $\mathbb{LK}_*$}
\label{LKstar frame axioms}
\hrule
\begin{enumerate}
\item[(F1)] For all $x,z\in W_1$ and all $v\in W_\partial$, the sections $Txv$, $Rxz$, $Svx$ of the frame relations $S^{\partial\partial 1},R^{111},T^{\partial 1\partial}$ are Galois sets.
\item[(M)] For any $y\in W_\partial$ and $x\in W_1$, the binary relations $yT,xR,yS$ are decreasing in both argument places. 
\item[(RES)]\hskip4mm For all $x,z\in W_1$ and $v\in W_\partial$, the equivalences $xS'vz$ iff $vR'xz$ iff $zT'xv$ hold.
\\[1mm] \underline{Unit Axioms}
\item[(U$^*$)] \hskip4mm $U$ is a Galois stable set.
\item[(F2.1)]\hskip4mm $\forall x,z_1,z_2\in W_1(xR^{111}z_1z_2\ra(z_1\in U\ra z_2\leq x) )$.
\item[(F2.2)]\hskip4mm $\forall x,z_1,z_2\in W_1(xR^{111}z_1z_2\ra(z_2\in U\ra z_1\leq x))$
\item[(F3$_*$.1)]\hskip4mm $\forall x\in W_1\exists z\in W_1(z\in U\wedge xR^{111}zx)$ .
\item[(F3$_*$.2)]\hskip4mm $\forall x\in W_1\exists z\in W_1(z\in U\wedge xR^{111}xz)$.
\end{enumerate}
\hrule
\end{table}

\begin{Proposition}
The non-associative Lambek calculus is sound and complete in the frame class $\mathbb{LK}_*\subseteq\mathbb{LK}^*\subseteq\mathbb{LK}$  axiomatized in \cref{LKstar frame axioms}, which is the axiom set in \cref{lambek frame class axioms} but including the monotonicity axiom (M) and replacing axioms (F3.1), (F3.2) by the axioms (F3$_*$.1), (F3$_*$.2).
\end{Proposition}
\begin{proof}
  Assuming the monotonicity properties axiom (M) for the frame relations results in simplifying the statement of axioms (F3$^*$.1) and (F3$^*$.2) given in \cref{star corollary}, since for example $\exists w(x\leq w\wedge xRwz)$ is equivalent to $xRxz$. That the canonical relations satisfy the monotonicity requirement has been proven for the general case in \mcite{Lemma~4.3, Case (3)}{duality2}, but it is immediate in the concrete case at hand by monotonicity of the point operators $\tleft,\medcircle,\tright$ and the way the canonical relations can be defined from them, since we have for example $uRxz$ iff $x\medcircle z\subseteq u$ and similarly for $S$ and $T$.
\end{proof}

\begin{Corollary}
\label{alternative for implication simplified}
For frames $\mathfrak{F}\in\mathbb{LK}_*\subseteq\mathbb{LK}^*\subseteq\mathbb{LK}$, which assume axiom (M), the membership conditions simplify to
\begin{tabbing}
$x\in A\Ra C$ \hskip5mm\= iff\hskip3mm\= $\forall u,z(u\in A\wedge  zR^{111}ux\lra z\in C)$\\
$x\in C\La A$ \>iff\> $\forall u,z(u\in A\wedge zR^{111}xu\lra z\in C)$.
\end{tabbing}
The satisfaction clauses for implication in \cref{alternative for implication} are then simplified accordingly.\telos
\end{Corollary}

Given \cref{alternative for implication}, it is useful to include in the language of the sorted companion  modal logic the implication connectives $\lspoon,\rspoon$,
\begin{eqnarray*}
\mathcal{L}_1\ni\alpha,\eta,\zeta &=& P_i\;(i\in\mathbb{N})\midsp \top\midsp\bot\midsp \alpha\cap\alpha\midsp\alpha\cup\alpha\midsp\beta' \midsp\alpha\lspoon\alpha\midsp\alpha\odot\alpha\midsp\alpha\rspoon\alpha\midsp\type{u}
\\
\mathcal{L}_\partial\ni\beta,\delta,\xi &=& P^i(i\in\mathbb{N})\midsp\top\midsp\bot\midsp \beta\cap\beta\midsp\beta\cup\beta\midsp\alpha'\midsp \beta\tleft\alpha\midsp \alpha\tright\beta, 
\end{eqnarray*}
to be interpreted as the residuals ${}_R\hskip-1mm\La,\Ra_R$ of $\bigodot$ in the dual sorted powerset algebra of a frame. From the proof of \cref{alternative for implication} it follows that the satisfaction clauses for frames in $\mathbb{LK}$ are given by $x\models\type{u}$ iff $x\in U$ and
\begin{tabbing}
$x\models \alpha\rspoon\eta$\hskip4mm\= iff \hskip2mm\= $\forall u,w,z(u\models\alpha\wedge x\leq w\wedge zR^{111}uw\lra z\models\eta)$\\
$x\models \eta\lspoon\alpha$ \>iff\> $\forall u,w,z(x\leq w\wedge u\models\alpha\wedge zR^{111}wu\lra z\models\eta)$
\end{tabbing}
simplifying to the clauses
\begin{tabbing}
$x\models \alpha\rspoon\eta$\hskip4mm\= iff \hskip2mm\= $\forall u,z(u\models\alpha\wedge  zR^{111}ux\lra z\models\eta)$\\
$x\models \eta\lspoon\alpha$ \>iff\> $\forall u,z( u\models\alpha\wedge zR^{111}xu\lra z\models\eta)$
\end{tabbing}
for models on frames $\mathfrak{F}\in\mathbb{LK}_*\subseteq\mathbb{LK}^*\subseteq\mathbb{LK}$, by \cref{alternative for implication simplified}.

For clarity's sake we restate the (extended) full proof system in \cref{extended sorted proof system} in which we have included the association axiom of the associative Lambek calculus, to be studied in the next \cref{associative calculus}.
\begin{table}[!t]
\caption{Sorted Proof System}
\hrule
\label{extended sorted proof system}
A. Axioms and Rules for Both Sorts\\
(if $\sigma$ is $\alpha\in\mathcal{L}_1$, then $\Mapsto$ is $\proves$, and if $\sigma$ is $\beta\in\mathcal{L}_\partial$, then $\Mapsto$ is $\vproves$)
\begin{tabbing}
$\sigma\Mapsto\sigma$ \hskip2cm\= $\bot\Mapsto\sigma$\hskip3cm\= $\sigma\Mapsto\top$ \hskip2cm\= $\top\Mapsto\bot'$
\\
$\sigma_1\Mapsto\sigma_1\cup\sigma_2$  \> $\sigma_2\Mapsto\sigma_1\cup\sigma_2$  \>  $\sigma_1\cap\sigma_2\Mapsto\sigma_1$  \> $\sigma_1\cap\sigma_2\Mapsto\sigma_2$\\
$\sigma_1\cap(\sigma_2\cup\sigma_3)\Mapsto(\sigma_1\cap\sigma_2)\cup(\sigma_1\cap\sigma_3)$
\\[2mm]
$\infrule{\sigma_1\Mapsto\sigma_2}{\sigma_1[\sigma/P]\Mapsto\sigma_2[\sigma/P]}$ \>\> $\infrule{\sigma_1\Mapsto\sigma_2\hskip4mm\sigma_2\Mapsto\sigma_3}{\sigma_1\Mapsto\sigma_3}$
\\[2mm]
$\infrule{\sigma_1\Mapsto\sigma\hskip4mm \sigma_2\Mapsto\sigma}{\sigma_1\cup\sigma_2\Mapsto\sigma}$ \>\>
$\infrule{\sigma\Mapsto\sigma_1\hskip4mm\sigma\Mapsto\sigma_2}{\sigma\Mapsto\sigma_1\cap\sigma_2}$
\end{tabbing}
B. Axioms and Rules for the (sorted) Modal Operators $\tleft,\odot,\tright$
\begin{tabbing}
$(\alpha\cup\eta)\tright\beta\vproves(\alpha\tright\beta)\cup(\eta\tright\beta)$ \hskip2cm\= $\infrule{\alpha\proves\eta\hskip4mm \beta\vproves\delta}{\alpha\tright\beta\vproves\eta\tright\delta}$
\\
$\bot\tright\beta\vproves\bot$ \hskip7mm $\alpha\tright\bot\vproves\bot$  \> $\alpha\tright(\beta\cup\delta)\vproves(\alpha\tright\beta)\cup(\alpha\tright\delta)$
\\[3mm]
$\beta\tleft(\alpha\cup\eta)\proves(\beta\tleft\alpha)\cup(\beta\tleft\eta)$\> $\infrule{\beta\vproves\delta\hskip4mm\alpha\proves\eta}{\beta\tleft\alpha\vproves\delta\tleft\eta}$
\\[1mm]
$\bot\tleft\alpha\vproves\bot$ \hskip7mm $\beta\tleft\bot\vproves\bot$
\> $(\beta\cup\delta)\tleft\alpha\vproves(\beta\tleft\alpha)\cup(\beta\tleft\eta)$
\\[3mm]
$\bot\odot\eta\proves\bot$ \hskip7mm $\alpha\odot\bot\proves\bot$ \> $\infrule{\alpha\proves\alpha_1\hskip4mm\eta\proves\eta_1}{\alpha\odot\eta\proves\alpha_1\odot\eta_1}$
\\[1mm]
$(\alpha\cup\eta)\odot\zeta\proves(\alpha\odot\zeta)\cup(\eta\odot\zeta)$ \>
$\zeta\odot(\alpha\cup\eta)\proves(\zeta\odot\alpha)\cup(\zeta\odot\eta)$
\end{tabbing}
C. Axioms and Rules for Sorted Negation
\begin{tabbing}
$\alpha\proves\alpha''$\hskip2.3cm\= $\beta\vproves\beta''$ 
\hskip2.3cm\=  $\infrule{\beta\vproves\delta}{\delta'\proves\beta'}$ \hskip2.3cm\=  $\infrule{\alpha\proves\eta}{\eta'\vproves\alpha'} $ 
\end{tabbing}
D.  Unit\\[2mm]
$\infrule{\alpha''\proves\alpha\hskip4mm\eta''\proves\eta\hskip4mm\alpha\proves\eta}{\type{u}\proves(\alpha\tright\eta')'}$
\\[2mm]
E. Left-Identity\hskip4.5cm Right-Identity
\begin{tabbing}
$\infrule{\beta''\vproves \beta}{\type{u}\tright \beta\vproves \beta}$\hskip1.8cm\= $\infrule{\beta''\vproves \beta}{\beta\vproves\type{u}\tright \beta}$
\hskip1.8cm\= $\infrule{\beta''\vproves \beta}{\beta\tleft\type{u}\vproves\beta}$ \hskip1.8cm $ \infrule{\beta''\vproves \beta}{\beta\vproves\beta\tleft\type{u}}$
\end{tabbing}
F. Implication Operators $\lspoon, \rspoon$
\begin{tabbing}
$\infrule{\alpha\proves\alpha_1\hskip4mm\eta_1\proves\eta}{\alpha_1\rspoon\eta_1\proves\alpha\rspoon\eta}$ \hskip5mm\= 
$(\alpha\rspoon\eta)\cap(\alpha\rspoon\zeta)\proves\alpha\rspoon\eta\cap\zeta$
\hskip5mm\= $(\alpha\rspoon\zeta)\cap(\eta\rspoon\zeta)\proves(\alpha\cup\eta)\rspoon\zeta$
\\[1mm]
$\infrule{\alpha\proves\alpha_1\hskip4mm\eta_1\proves\eta}{\eta_1\lspoon\alpha_1\proves\eta\lspoon\alpha}$\>
$(\zeta\lspoon\alpha)\cap(\eta\lspoon\alpha)\proves(\eta\cap\zeta)\lspoon\alpha$ \> 
$(\zeta\lspoon\alpha)\cap(\zeta\lspoon\eta)\proves\zeta\lspoon(\alpha\cup\eta)$
\\[1mm]
Residuation \> $\infrule{\alpha\odot\eta\proves\zeta}{\overline{\eta\proves\alpha\rspoon\zeta}}$ \> $\infrule{\alpha\odot\eta\proves\zeta}{\overline{\alpha\proves\zeta\lspoon\eta }}$
\\[1mm]
Restriction to stable \>
$\infrule{\alpha''\proves\alpha\hskip4mm\eta''\proves\eta}{\alpha\rspoon\eta\equiv(\alpha\tright\eta')'}$ \> 
$\infrule{\alpha''\proves\alpha\hskip4mm\eta''\proves\eta}{\eta\lspoon\alpha\equiv (\eta'\tleft\alpha)'}$
\end{tabbing}
G. Association  $\alpha\odot(\zeta\odot\eta)\equiv(\alpha\odot\zeta)\odot\eta$
\\
\hrule
\end{table} 

\subsection{The Associative Calculus} 
\label{associative calculus}
To calculate the axiomatization of the largest frame class in which the associative Lambek calculus is sound we extend the language of its modal companion logic and then apply the (updated) generalized correspondence algorithm of \cite{dfmlC}, briefly reviewed in  \cref{sahlqvist section}.

Given the equivalence $\alpha\rspoon\eta\equiv(\alpha\tright\eta')'$ and $\eta\lspoon\alpha\equiv(\eta'\tleft\alpha)'$, which holds for $\alpha''\equiv\alpha$ and $\eta''\equiv\eta$, there are two ways to translate an implication sentence of the language of the Lambek calculus, 
\[
(\varphi\rfspoon\psi)^\bullet=(\varphi^\bullet\tright\psi^\circ)'=\varphi^\bullet\rspoon\psi^\bullet
\;\mbox{ and }\; 
(\psi\lfspoon\varphi)^\bullet=(\psi^\circ\tleft\varphi^\bullet)'=\psi^\bullet\lspoon\varphi^\bullet.
\]
Because of the equivalence, the full abstraction result (Theorem~\ref{full abstraction of trans in sorted modal}) holds for the extended language, by \mcite{Theorem~3.2}{vb}.

To make use of the extended sorted modal language in the generalized Sahlqvist -- van Benthem algorithm, extend the rule (R5) to include additional rewrite cases and add a new rule (R7) to the reduction rules of \cref{reduction rules}.

The rules with $\rspoon$ were listed and used in \cite{dfmlC} and they were considered in the proof that rule application leads from a system of formal inequations to an equivalent one, see \mcite{Lemma~2}{dfmlC}. Adding the analogous rules for $\lspoon$ is unproblematic. 

Having set things up as needed, we can address the case of frames for the associative Lambek calculus. 

\begin{Proposition}
  \label{association prop}
The axioms   $p_1\circ(p_2\circ p_3)\proves(p_1\circ p_2)\circ p_3$ and $(p_1\circ p_2)\circ p_3\proves p_1\circ(p_2\circ p_3)$ correspond to the first-order frame constraints below, for frames in the respective frame classes, as indicated.\\
\underline{Axiom $p_1\circ(p_2\circ p_3)\proves(p_1\circ p_2)\circ p_3$}
\begin{equation}
\underline{\mathrm{Correspondent}}\hskip1.5cm\forall^1x\forall^1z_1\forall^1z_2\forall^1z_3[\exists^1u(x\mathbf{R}^{111}z_1u\wedge u\mathbf{R}^{111}z_2z_3) \lra\beta(\mathrm{POS})]
\label{assoLK-1}
\end{equation}
\begin{itemize}
\item $\mathfrak{F}\in\mathbb{LK}$, 
\end{itemize}
\begin{tabbing}
$\beta(\mathrm{POS})$\hskip2mm\==\hskip1mm\= $\forall^\partial y[x\mathbf{I}y\ra\exists^1w(w\mathbf{I}y\wedge \exists^1v_1\exists^1\hat{z}_3(w\mathbf{R}^{111}v_1\hat{z}_3\;\wedge\; z_3\leq\hat{z}_3\;\wedge$\\
\>\>$\wedge\forall^\partial y_1(v_1\mathbf{I}y_1\lra\exists^1v(v\mathbf{I}y_1\wedge\exists^1\hat{z}_1\exists^1\hat{z}_2(v\mathbf{R}^{111}\hat{z}_1\hat{z}_2\wedge z_1\leq\hat{z}_1\wedge z_2\leq\hat{z}_2)))\;))]$.
\end{tabbing}

\begin{itemize}
\item $\mathfrak{F}\in\mathbb{LK}^*$, 
\end{itemize}
\begin{tabbing}
$\beta^*(\mathrm{POS})$\hskip2mm\==\hskip1mm\= $\exists^1u\exists^1 \hat{z}_1,\hat{z}_2,\hat{z}_3(z_1\leq \hat{z}_1\wedge z_2\leq\hat{z}_2\wedge z_3\leq\hat{z}_3\wedge (x\mathbf{R}^{111}u\hat{z}_3\wedge u\mathbf{R}^{111}\hat{z}_1\hat{z}_2))$
\end{tabbing}
\begin{itemize}
\item  $\mathfrak{F}\in\mathbb{LK}_*$
\end{itemize}
\begin{tabbing}
$\beta_*(\mathrm{POS})$\hskip2mm\==\hskip1mm\= $\exists^1u (x\mathbf{R}^{111}uz_3\wedge u\mathbf{R}^{111}z_1z_2))$.
\end{tabbing}

\noindent
\underline{Axiom $(p_1\circ p_2)\circ p_3\proves p_1\circ(p_2\circ p_3)$}
\begin{equation}
\underline{\mathrm{Correspondent}}\hskip1.5cm\forall^1x\forall^1z_1\forall^1z_2\forall^1z_3(\exists^1u(x\mathbf{R}^{111}uz_3 \wedge u\mathbf{R}^{111}z_1z_2)\lra\tilde{\beta}(\mathrm{POS}))
\end{equation}
\begin{itemize}
\item  $\mathfrak{F}\in\mathbb{LK}$
\end{itemize}
$\tilde{\beta}(\mathrm{POS})=$\\
\hskip5mm
$\forall^\partial y[x\mathbf{I}y\lra\exists^1w(w\mathbf{I}y\wedge\exists^1\hat{z}_1\exists^1v_1[w\mathbf{R}^{111}\hat{z}_1v_1\wedge z_1\leq \hat{z}_1\wedge\forall^\partial y_1(v_1\mathbf{I}y_1\lra\exists^1v(v\mathbf{I}y_1\wedge$\\
\hspace*{2cm} $\wedge\exists^1\hat{z}_2\exists^1\hat{z}_3(z_2\leq\hat{z}_2\wedge z_3\leq\hat{z}_3\wedge v\mathbf{R}^{111}\hat{z}_2\hat{z}_3)))])]$

\begin{itemize}
\item  $\mathfrak{F}\in\mathbb{LK}^*$
\end{itemize}
\begin{tabbing}
$\tilde{\beta}^*(\mathrm{POS})$\hskip2mm\==\hskip1mm\= $\exists^1u\exists\hat{z}_1\exists\hat{z}_2\exists\hat{z}_3 (z_1\leq\hat{z}_1\wedge z_2\leq\hat{z}_2\wedge z_3\leq\hat{z}_3\wedge x\mathbf{R}^{111}\hat{z}_1u\wedge u\mathbf{R}^{111}\hat{z}_2\hat{z}_3)$
\end{tabbing}

\begin{itemize}
\item $\mathfrak{F}\in\mathbb{LK}_*$
\end{itemize}
$\tilde{\beta}_*(\mathrm{POS})\;=\;\exists^1u(x\mathbf{R}^{111}\hat{z}_1u\wedge u\mathbf{R}^{111}\hat{z}_2\hat{z}_3)$.
\end{Proposition}
\begin{proof}
$p_1\circ(p_2\circ p_3)\proves (p_1\circ p_2)\circ p_3$ is equivalent to $p_2\circ p_3 \proves p_1\rfspoon ((p_1\circ p_2)\circ p_3)$, by residuation, which we can reduce to a system in canonical Sahlqvist form.
\begin{tabbing}
 $\langle(p_2\circ p_3)^\bullet \leq_1 (p_1\rfspoon (p_1\circ p_2)\circ p_3)^\bullet\rangle$\hskip5.6cm\= 
\\
 $\langle (P_2''\odot P_3'')''\leq_1 P_1''\rspoon ((P_1''\odot P_2'')''\odot P_3'')''\rangle$
\\
 $\langle P''_1\leq_1P_1,P''_2\leq_1P_2,P_3''\leq_1P_3\midsp (P_2\odot P_3)''\leq_1 P_1\rspoon ((P_1\odot P_2)''\odot P_3)''\rangle$ \> by (R4), \cref{reduction rules}
\\
$\langle P''_1\leq_1P_1,P''_2\leq_1P_2,P_3''\leq_1P_3\midsp P_2\odot P_3\leq_1 P_1\rspoon ((P_1\odot P_2)''\odot P_3)''\rangle$ \> by (R3), \cref{reduction rules}
\\
$\langle P''_1\leq_1P_1,P''_2\leq_1P_2,P_3''\leq_1P_3\midsp P_1\odot(P_2\odot P_3)\leq_1  ((P_1\odot P_2)''\odot P_3)''\rangle$. \> by (R7), \cref{reduction rules}
\end{tabbing}
The guarded second-order translation is
\begin{tabbing}
$\forall^1P_1\forall^1P_2\forall^1P_3\forall^1x$ \= ${\large [}(\bigwedge_{i=1}^3\forall^1w_i[\texttt{t}(P_i)(w_i)\ra P_i(w_i)])\wedge
\exists^1z_1\exists^1u(P_1(z_1)\wedge$
\\
\> $\wedge\exists^1z_2\exists^1z_3(P_2(z_2)\wedge P_3(z_3)\wedge  x\mathbf{R}^{111}z_1u \wedge u\mathbf{R}^{111}z_2z_3))\lra \mathrm{POS}  {\large ]}$,
\end{tabbing}
where
\begin{tabbing} 
$\mathrm{POS}$\hskip2mm\==\hskip1mm\= $\mathrm{ST}_x(((P_1\odot P_2)''\odot P_3)'')$ \\
\>=\> $\forall^\partial y[x\mathbf{I}y\ra\exists^1w(w\mathbf{I}y\wedge\stx{w}{(P_1\odot P_2)''\odot P_3}\;)]$\\
\>=\> $\forall^\partial y[x\mathbf{I}y\ra\exists^1w(w\mathbf{I}y\wedge \exists^1v_1\exists^1\hat{z}_3(w\mathbf{R}^{111}v_1\hat{z}_3\wedge P_3(\hat{z}_3)\wedge\stx{{v_1}}{(P_1\odot P_2)''}\;))]$\\
\>=\> $\forall^\partial y[x\mathbf{I}y\ra\exists^1w(w\mathbf{I}y\wedge \exists^1v_1\exists^1\hat{z}_3(w\mathbf{R}^{111}v_1\hat{z}_3\wedge P_3(\hat{z}_3)\wedge$\\
\>\>\hskip1.2cm\= $\wedge\forall^\partial y_1(v_1\mathbf{I}y_1\lra\exists^1v(v\mathbf{I}y_1\wedge\stx{v}{P_1\odot P_2}))\;))]   $\\
\>=\> $\forall^\partial y[x\mathbf{I}y\ra\exists^1w(w\mathbf{I}y\wedge \exists^1v_1\exists^1\hat{z}_3(w\mathbf{R}^{111}v_1\hat{z}_3\wedge P_3(\hat{z}_3)\wedge$\\
\>\>\> $\wedge\forall^\partial y_1(v_1\mathbf{I}y_1\lra\exists^1v(v\mathbf{I}y_1\wedge\exists^1\hat{z}_1\exists^1\hat{z}_2(v\mathbf{R}^{111}\hat{z}_1\hat{z}_2\wedge P_1(\hat{z}_1)\wedge P_2(\hat{z}_2))))\;))]   $\\
\end{tabbing}

Pull out existential quantifiers, set $\lambda(P_1)=\lambda s.z_1\leq s$, $\lambda(P_2)=\lambda s.z_2\leq s$, $\lambda(P_3)=\lambda s.z_3\leq s$, substitute the $\lambda$-terms for the $P_i$, perform $\beta$-reduction and, letting $\beta(\mathrm{POS})$ be the result at the consequent position, we obtain
\[
\forall^1x\forall^1z_1\forall^1u\forall^1z_2\forall^1z_3(x\mathbf{R}^{111}z_1u\wedge u\mathbf{R}^{111}z_2z_3 \lra\beta(\mathrm{POS})),
\]
where 
\begin{tabbing}
$\beta(\mathrm{POS})$\hskip2mm\==\hskip1mm\= $\forall^\partial y[x\mathbf{I}y\ra\exists^1w(w\mathbf{I}y\wedge \exists^1v_1\exists^1\hat{z}_3(w\mathbf{R}^{111}v_1\hat{z}_3\;\wedge\; z_3\leq\hat{z}_3\;\wedge$\\
\>\>$\wedge\forall^\partial y_1(v_1\mathbf{I}y_1\lra\exists^1v(v\mathbf{I}y_1\wedge\exists^1\hat{z}_1\exists^1\hat{z}_2(v\mathbf{R}^{111}\hat{z}_1\hat{z}_2\wedge z_1\leq\hat{z}_1\wedge z_2\leq\hat{z}_2)))\;))]$.
\end{tabbing}
The above specify the first-order correspondent for frames in the class $\mathbb{LK}$, axiomatized in \cref{lambek frame class axioms}. 

The smaller class $\mathbb{LK}^*\subseteq\mathbb{LK}$, axiomatized in \cref{LKstar frame axioms}, satisfies the stronger inequality $P_1\odot(P_2\odot P_3)\leq_1(P_1\odot P_2)\odot P_3$, with guard $P''_i\leq P_i$, for $i=1,2,3$. A similar calculation, left to the reader, indeed returns
\begin{tabbing}
$\beta^*(\mathrm{POS})$\hskip2mm\==\hskip1mm\= $\exists^1u\exists^1 \hat{z}_1,\hat{z}_2,\hat{z}_3(z_1\leq \hat{z}_1\wedge z_2\leq\hat{z}_2\wedge z_3\leq\hat{z}_3\wedge (x\mathbf{R}^{111}u\hat{z}_3\wedge u\mathbf{R}^{111}\hat{z}_1\hat{z}_2))$.
\end{tabbing}

The smallest class $\mathbb{LK}_*\subseteq\mathbb{LK}^*\subseteq\mathbb{LK}$, axiomatized in \cref{LKstar frame axioms}, satisfies the same stronger inequality but it also assumes the monotonicity axiom (M) for frame relations and this allows simplifying the consequent to
\begin{tabbing}
$\beta_*(\mathrm{POS})$\hskip2mm\==\hskip1mm\= $\exists^1u (x\mathbf{R}^{111}uz_3\wedge u\mathbf{R}^{111}z_1z_2))$.
\end{tabbing}

For the converse direction $(p_1\circ p_2)\circ p_3\proves p_1\circ(p_2\circ p_3)$, we similarly process the (equivalent) sequent $p_1\circ p_2\proves (p_1\circ(p_2\circ p_3))\lfspoon p_3$, which reduces to
\[
\langle P''_1\leq_1P_1,P''_2\leq_1P_2,P_3''\leq_1P_3\midsp (P_1\odot P_2)\odot P_3\leq_1(P_1\odot (P_2\odot P_3)'')''\rangle.
\]
The computation of the correspondent for each of the frame classes $\mathbb{LK}, \mathbb{LK}^*,\mathbb{LK}_*$ is similar to that presented above. Details can be safely left to the reader.
\end{proof}

In the joint article \cite{dunn-gehrke} of Dunn, Gehrke and Palmigiano, semantics of the Lambek calculus in RS-frames $(W_1,I,W_\partial,R)$ was presented. The frame constraint implying association that was specified is the constraint
\begin{tabbing}
$\forall^1x_1\forall^1x_2\forall^1x_3\forall^\partial m $ \\
\hskip1cm\= $([\forall^1x_2'(\forall^\partial m'[R(x_2,x_3,m')\Lra x_2'\leq m']\Lra R(x_1,x_2',m))] $ \\
\hskip3mm $\Longleftrightarrow[\forall^1x_1'([\forall^\partial m''(R(x_1,x_2,m'')\Lra x_1'\leq m'')]\Lra R(x_1',x_3,m))])$
\end{tabbing}
which is as involved as the constraint we calculated for the $\mathbb{LK}$ frame class.

It is known from the literature on Relevance logic (consult \cite{relevance3}, for example) that the constraint 
\begin{equation}
\label{associativity constraint}
 \exists^1 u(xRz_1u\wedge uRz_2z_3)\;\mbox{ iff }\; \exists^1 u(uRz_1z_2\wedge xRuz_3)
\end{equation}
corresponds to associativity of the fusion operator of the logic. This is precisely the constraint computed for the frame class $\mathbb{LK}_*$ in \cref{association prop}. The relation $R$ is the composition relation and $xRuz$ is to be understood as something like ``$x$ is a composition of $u$ and $z$''. Read in this way, the constraint in~\eqref{associativity constraint} expresses associativity in a direct way.  

\begin{Theorem}
\label{completeness for associative Lambek}
The associative Lambek calculus is sound in the frame class $\mathbb{LK}$ (hence also in its subclasses $\mathbb{LK}^*,\mathbb{LK}_*$) and it is complete in the frame class $\mathbb{LK}_*$ (hence also in $\mathbb{LK}^*$ and in $\mathbb{LK}$).
\end{Theorem}
\begin{proof}
Soundness is covered by the correctness proof of the generalized Sahlqvist -- van Benthem algorithm, cf. \mcite{Section~4}{dfmlC}.

For completeness, by the reduction detailed in the proof of \cref{association prop}, the axiom $p_1\circ(p_2\circ p_3)\proves (p_1\circ p_2)\circ p_3$ and the rule $\infrule{P_1''\proves P_1\hskip4mm P_2''\proves P_2\hskip4mm P_3''\proves P_r}{P_1\odot(P_2\odot P_3)\proves((P_1\odot P_2)''\odot P_3)''}$ in the companion modal logic are semantically equivalent in $\mathbb{LK}$. In the canonical frame, the stronger fact $A\bigodot(C\bigodot F)\subseteq (A\bigodot C)\bigodot F$ holds, for $A,C,F\in\gpsi$, and similarly for the converse inclusion and the axiom $(p_1\circ p_2)\circ p_3\proves p_1\circ(p_2\circ p_3)$. This is verified in the sequel.

That the canonical frame relations satisfy the monotonicity axiom (M) was proven for the general case of frames for normal lattice expansions \mcite{Lemma~4.3, Case (3)}{duality2}, but this can be easily verified by the reader for the concrete case at hand.

Observe that the point operator $\medcircle$ is associative in the canonical frame of the associative Lambek calculus. 

Indeed, if $e\in u\medcircle(x\medcircle z)$, let $a\in u, b\in(x\medcircle z)$ such that $a\circ b\leq e$. Since $b\in x\medcircle z$, let $c\in x, d\in z$ such that $c\circ d\leq b$. Then $a\circ(c\circ d)\leq a\circ b\leq e$. By association in the logic $(a\circ c)\circ d\leq e$. This implies $e\in (u\medcircle x)\medcircle z$, hence the inclusion $u\medcircle(x\medcircle z)\subseteq (u\medcircle x)\medcircle z$ holds. The converse inclusion is established similarly. 

Furthermore, as pointed out in the proof of \cref{sigma-pi in canonical}, for any points $x,z\in W_1$, $\Gamma x\bigovert\Gamma z=\Gamma x\bigodot\Gamma z=\Gamma(x\medcircle z)=R^{111}xz$. It follows that the identity $\Gamma x\bigodot(\Gamma z\bigodot\Gamma w)=(\Gamma x\bigodot\Gamma z)\bigodot\Gamma w$ (*) holds in the canonical frame. This implies in particular that $A\bigodot(C\bigodot F)= (A\bigodot C)\bigodot F$.

Rewriting the identity (*) we get $\Gamma x\bigodot Rzw=Rxz\bigodot\Gamma w$. Unfolding definitions we obtain 
\begin{tabbing}
$u\in \Gamma x\bigodot R^{111}zw$ \hskip2mm\= iff \hskip2mm\= $\exists\hat{u}\exists\hat{x}(x\leq\hat{x}\wedge\hat{u}R^{111}zw\wedge uR^{111}\hat{x}\hat{w})$\\
\>iff\> $\exists\hat{u}\exists\hat{x}(x\leq\hat{x}\wedge z\medcircle w\subseteq\hat{u}\wedge \hat{x}\medcircle\hat{w}\subseteq u)$\\
\>iff\> $\exists\hat{u}(z\medcircle w\subseteq\hat{u}\wedge x \medcircle\hat{w}\subseteq u)$\\
\>iff\> $\exists\hat{u}(uR^{111}x\hat{u}\wedge\hat{u}R^{111}zw)$.\\
$u\in R^{111}xz\bigodot\Gamma w$\>iff\> $\exists\hat{u}(uR^{111}\hat{u}w\wedge \hat{u}R^{111}xz)$, 
\end{tabbing}
leaving computation details of the latter to the reader.
\end{proof}

\begin{Proposition}
\label{odot associative}
The following are equivalent in the frame class $\mathbb{LK}_*$.
\begin{enumerate}
\item[{\rm(1)}] the frame constraint $ \exists^1 u(xRz_1u\wedge uRz_2z_3)\;\mbox{ iff }\; \exists^1 u(uRz_1z_2\wedge xRuz_3)$
\item[{\rm(2)}] the image operator $\bigodot$ generated by the ternary relation $R^{111}$ is associative in the dual sorted powerset algebra of the frame
\item[{\rm(3)}] the stable sets operator $\bigovert$ is associative in the full complex algebra of the frame.
\end{enumerate}
\end{Proposition}
\begin{proof}
For (1)$\Leftrightarrow$(2),  use complete additivity of the image operator $\bigodot$ and verify that the identity $\{z_1\}\bigodot(\{z_2\}\bigodot\{z_3\})=(\{z_1\}\bigodot\{z_2\})\bigodot\{z_3\}$ is equivalent to the frame constraint \eqref{associativity constraint}, listed as (1) in the statement of the Proposition. 

The proof that (1)$\Leftrightarrow$(3) is similar, based on the fact that $\Gamma z_1\bigovert\Gamma z_2= Rz_1z_2=\{z_1\}\bigodot\{z_2\}$, where the first identity uses the monotonicity axiom (M) of $\mathbb{LK}_*$.
\end{proof}

\subsection{Adding Contraction, Exchange, Weakening}
\label{adding section}
There are different ways to translate a sequent in each case and we choose the one that is most convenient for the needed elementary computations. \cref{corr table} collects together correspondence results for implicative logics. As an example, it displays correspondence conditions for the left-unit axiom that can be stated both as $\type{t}\circ p\equiv p$ and as $p\equiv\type{t}\rfspoon p$.

\paragraph{Weakening}
The sequent $p\proves q\rfspoon p$ translates to $P''\proves Q''\rspoon P''$. Exploiting residuation in the sorted modal logic we obtain the formal inequality $Q\odot P\leq_1P$, with guards $P''\leq_1P, Q''\leq_1Q$ and guarded second-order translation $\forall^1P\forall^1Q\forall^1x[\type{t}-\mathrm{INV}\wedge\exists^1u\exists^1z(Q(u)\wedge P(z)\wedge x\mathbf{R}^{111}uz)\lra P(x)]$. The reader can verify that we obtain the first-order local correspondent 
\[
\forall^1x\forall^1u\forall^1z(x\mathbf{R}^{111}uz\lra z\leq x).
\]
In the canonical frame, $xRuz$ is defined by $u\medcircle z\subseteq x$. Letting $a\in z$ (a filter), for any $b\in u$ we get $b\circ a\in u\medcircle z\subseteq x$. By weakening, $b\circ a\leq a$, hence $a\in x$. Conclude that the weakening axiom is canonical. This implies completeness of the (associative, or not) Lambek calculus extended with the weakening axiom.

\paragraph{Contraction}
The translation of the contraction sequent $p\rfspoon(p\rfspoon q)\proves p\rfspoon q$ leads to the system $\langle P''\leq_1P,Q''\leq_1Q\midsp P\rspoon(P\rspoon Q)\leq_1P\rspoon Q\rangle$. Residuation in the companion modal logic and reduction rules lead to the inequality $P\leq_1P\odot P$, with guard $P''\leq_1P$. The reader can work out details to conclude that in the frame class $\mathbb{LK}_*$, which assumes the monotonicity axiom (M) for frame relations, the returned correspondent is $\forall^1x\;x\mathbf{R}^{111}xx$. 

In the canonical frame of the logic assuming contraction, we verify that the operator $\bigodot$ in the sorted powerset algebra is contractive, i.e. for any set $X$ of filters $X\subseteq X\bigodot X$. Contractivity of $\bigodot$ is equivalent to $xRxx$, for any filter $x$, as well as to the inclusion $X\bigcap Y\subseteq X\bigodot Y$, for sets $X,Y$ of filters. Note also that $xRxx$ is equivalent to the inclusion $\{x\}\subseteq\{x\}\bigodot\{x\}$, i.e. to $x\in\{x\}\bigodot\{x\}=\{u\midsp\exists a,b\in x\;a\circ b\in u\}$. But clearly taking $b=a\in x$ we obtain $a\circ b=a\circ a\geq a\in x$, hence $a\circ a=a\circ b\in x$ and thereby $x\in\{x\}\bigodot\{x\}$.

It follows that for any stable set $A$ of filters we have $A\subseteq A\bigodot A\subseteq(A\bigodot A)''=A\bigovert A$ and so $\bigovert$ is contractive. By the above, the contraction axiom is canonical, which implies completeness of the (associative, or not) Lambek calculus extended with contraction.

\paragraph{Exchange}
From the exchange sequent $p_1\rfspoon(p_2\rfspoon q)\proves p_2\rfspoon(p_1\rfspoon q)$ and translating $\rfspoon$ using $\rspoon$ we obtain the system $\langle P_1''\leq_1 P_1,P_2''\leq_1 P_2, Q''\leq_1 Q\midsp P_1\rspoon(P_2\rspoon Q)\leq_1P_2\rspoon(P_1\rspoon Q)\rangle$. Using residuation in the sorted modal logic and the reduction rules this reduces to the inequality $P_1\odot P_2\leq_1 P_2\odot P_1$, with the guards $P_i''\leq_1P_i$ ($i=1,2$) and $Q''\leq_1Q$. The reader can write-out the guarded second-order translation and carry out the steps of the correspondence algorithm to obtain as a first-order local correspondent 
\[
\forall^1x\forall^1u\forall^1z(x\mathbf{R}^{111}uz\lra x\mathbf{R}^{111}zu).
\]
In the canonical frame, assuming $xRuz$ for filters $x,u,z$, i.e. $u\medcircle z\subseteq x$, it follows immediately that $z\medcircle u\subseteq x$, since for $a\in z, b\in u$, $a\circ b=b\circ a\in u\medcircle z\subseteq x$. Thereby $xRzu$ holds. Conclude that the exchange axiom is canonical, which implies completeness of the (associative, or not) Lambek calculus extended with the exchange axiom.

We conclude with the following result.
\begin{Theorem}
The extensions of the non-associative Lambek calculus obtained by adding any of the ``structural'' axioms of association, exchange, weakening or contraction are sound and complete in the frame class $\mathbb{LK}_*$ with the respective additional frame axioms listed in \cref{corr table}.\telos
\end{Theorem}

\begin{landscape}
\begin{table}[!th]
\caption{Modal and First-Order Correspondents (in the frame class $\mathbb{LK}_*$) for the Lambek calculus}
\label{corr table}
\begin{tabular}{|l|c|l|l|l|}
\hline
Axiom/Rule & $\begin{array}{c} \mbox{Modal}\\ \mbox{Correspondent}\end{array}$ & Sorted Powerset Algebra & 
$\begin{array}{cc} \mbox{Sorted Full Complex}\\ \mbox{Algebra}\end{array}$ 
& Local First-Order Correspondent in $\mathbb{LK}_*$\\
\hline
Residuation & $\infrule{\alpha\odot\eta\proves\zeta}{\overline{\eta\proves\alpha\rspoon\zeta}}$
& $X{\bigodot}Y{\subseteq}Z$ iff $Y{\subseteq}X{\Ra}Z$
& $A{\bigovert} F{\subseteq} C$ iff $F{\subseteq}A{\Ra}C$ & $\forall^1x\forall^1z\forall^\partial y(\mbox{$xS'yz$ iff $yR'xz$ iff $zT'xy$})$\\
\hline
$\infrule{p\proves q}{\overline{\type{t}\proves p\rfspoon q}}$ & $\infrule{\eta'\vproves\alpha'}{\overline{\type{u}\proves(\alpha{\tright}\eta')'}}$ & $Y'\subseteq X'$ iff $U\subseteq (X{\ltright}Y')'$ & $A\subseteq C$ iff $U\subseteq A\Ra C$ & $\forall^1x\forall^\partial y(x\upv y\leftrightarrow\forall^1u{\in} U\; uT'xy)$\\
\hline
$p\proves\type{t}\circ p$ & $\alpha\proves\type{u}\odot\alpha$ & $X\subseteq U\bigodot X$ & $A\subseteq U\bigovert A$ ($A\in\gpsi$) & $\forall^1x\exists^1u(u\in U\wedge xRux)$\\
\hline
$p\proves p\circ\type{t}$ & $\alpha\proves\alpha\odot\type{u}$ & $X\subseteq X\bigodot U$ & $A\subseteq A\bigovert U$ & $\forall^1x\exists^1u(u\in U\wedge xRxu)$\\
\hline
$\type{t}\circ p\proves p$ & $\infrule{\alpha''\proves\alpha}{\type{u}\odot\alpha\proves\alpha}$ & $U\bigodot X\subseteq X$, for $X=X^\uparrow$ & $U\bigovert A\subseteq A$ & $\forall^1x\forall^1u\forall^1z(u\in U\wedge xRuz\lra z\leq x)$\\[0.8mm]
\hline
$p\circ\type{t}\proves p$  & $\infrule{\alpha''\proves\alpha}{\alpha\odot\type{u}\proves\alpha}$ & $X\bigodot U\subseteq X$, for $X=X^\uparrow$ & $A\bigovert U\subseteq A$ & $\forall^1x\forall^1z\forall^1u(u\in U\wedge xRzu\lra z\leq x)$\\[0.8mm]
\hline\hline
$p\proves\type{t}\rfspoon p$ & $\infrule{\beta''\vproves\beta}{\type{u}{\tright}\beta\vproves\beta}$ & $U{\ltright}Y\subseteq Y$, for $Y=Y^\uparrow$ & $U\Mtright B\subseteq B$ ($B\in\gphi$) & $\forall^\partial y\forall^\partial v\forall^1x(x\in U\wedge yTxv\lra v\leq y)$\\
\hline
$p\proves p\lfspoon \type{t}$ & $\infrule{\beta''\vproves\beta}{\beta{\tleft}\type{u}\vproves\beta}$ & $Y\ltleft U\subseteq Y$, for $Y=Y^\uparrow$ & $B\Mtleft U\subseteq B$ & $\forall^\partial y\forall^\partial v\forall^1x(x\in U\wedge ySvx\lra v\leq y)$ \\
\hline\hline
Association & $\infrule{\alpha\odot(\eta\odot\zeta)\proves\xi}{\overline{(\alpha\odot\eta)\odot\zeta\proves\xi}}$
& $\begin{array}{l} X\bigodot (Y\bigodot Z) \\ \hskip6mm =(X\bigodot Y)\bigodot Z\end{array}$ & 
$\begin{array}{l} A\bigovert(C\bigovert F)\\ \hskip6mm =(A\bigovert C)\bigovert F\end{array}$
& $\exists^1 u(xRz_1u\wedge uRz_2z_3)\;\mbox{ iff }\; \exists^1 u(uRz_1z_2\wedge xRuz_3)$\\
\hline
Exchange & $\alpha\odot\eta\proves\eta\odot\alpha$ & $X\bigodot Z\subseteq Z\bigodot X$ & $A\bigovert C\subseteq C\bigovert A$ & $\forall^1x\forall^1u\forall^1z(xRuz\lra xRzu)$ \\
\hline
Contraction & $\alpha\proves\alpha\odot\alpha$ & $X\subseteq X\bigodot X$ & $A\subseteq A\bigovert A$ & $\forall^1x\; xRxx$\\
\hline
Weakening & $\infrule{\eta''\proves\eta}{\alpha\odot\eta\proves\eta}$ & $Z\bigodot X\subseteq X$, for $X=X^\uparrow$ & $A\bigovert C\subseteq C$ & $\forall^1x\forall^1u\forall^1z(xRuz\lra z\leq x)$\\
\hline\hline
\hskip-2mm$\begin{array}{cc}\mbox{Visser's axiom}\\ p{\rfspoon}(q{\rfspoon}p{\wedge}q)\end{array}$ & $\infrule{\alpha''\proves\alpha\hskip3mm\eta''\proves\eta}{\alpha\odot\eta\proves\alpha\cap\eta}$  &
$\begin{array}{cc}X\bigodot Z\subseteq X\bigcap Z\\ \mbox{for }X=X^\uparrow, Z=Z^\uparrow\end{array}$ 
& $A\bigovert C\subseteq A\bigcap C$&
$\forall^1u\forall^1x\forall^1z(uRxz\lra x\leq u\wedge z\leq u)$\\
\hline
\end{tabular}
\end{table}
\end{landscape}

\section{Implicative Semilattices and their Logic}
\label{semilattices section}
We simplify the presentation by switching in this section to integral posets, where $\type{t}=1$. If this is unwanted, just include the unit frame axioms of \cref{LKstar frame axioms}.

An implicative meet semilattice $\mathbf{M}=(M,\leq,\wedge,1,\ra)$ is an implicative poset whose reduct $(M,\leq,\wedge,1)$ is a meet semilattice and where implication also satisfies the distribution property $a\ra b\wedge c=(a\ra b)\wedge(a\ra c)$. The minimal logic of implicative meet semilattices, in the language including conjunction, is that of implicative posets extended with the meet semilattice rules from Table~\ref{minimal proof system}, displayed below,

\begin{tabular}{lllllll}
 \multicolumn{2}{l}{Meet semilattice rules} &
 $\infrule{\varphi\proves\vartheta}{\varphi\wedge\psi\proves\vartheta}$  & $\infrule{\psi\proves\vartheta}{\varphi\wedge\psi\proves\vartheta}$ &
 $\infrule{\varphi\proves\psi\hskip5mm\varphi\proves\vartheta}{\varphi\proves\psi\wedge\vartheta}$
\end{tabular}

\noindent
together with the distribution axiom for meet semilattices from Table~\ref{minimal proof system} instantiated to implication
\[
(\varphi\rfspoon\psi)\wedge(\varphi\rfspoon\vartheta)\proves\varphi\rfspoon\psi\wedge\vartheta.
\]
Since the only quasi-operator (other than the (semi)lattice operations) in its Lindenbaum-Tarski algebra is implication, the language is interpreted in frames $\mathfrak{F}=(s,W,I,T,\sigma)$, as for implicative posets, but with an addition of an axiom to validate the axiom of distribution over conjunction in the consequent place.

The frame class $\mathbb{S}$ consists of frames  $\mathfrak{F}=(s,W,I,T,\sigma)$ validating the axioms in Table~\ref{frame axioms for semilattice-based logics}.

\begin{table}[!htbp]
\caption{Frame axioms for $\wedge$-semilattice-based implicative logics}
\label{frame axioms for semilattice-based logics}
\hrule
\begin{enumerate}
\item[(F1)] For all $x\in W_1$ and all $v\in W_\partial$, the section $Txv\subseteq W_\partial$ of the relation $T^{\partial 1\partial}$ is a Galois set.
\item[(F2)] For all $x\in W_1$ and all $v\in W_\partial$, $x\upv v$ holds iff for all $u$ in $W_1$, $uT'xv$ holds.
\item[(F3)] For all $x,z\in W_1$, the section $xT'z[\;]=\{y\in W_\partial\midsp xT'zy\}\subseteq W_\partial$ of the Galois dual $T'$ of the frame relation $T$ is a Galois (co-stable) set.
\end{enumerate}
\hrule
\end{table}

Models $\mathfrak{M}=(\mathfrak{F},V)$ are defined as before and the interpretation of the language is as for the language of implicative posets, with implication interpreted by the clause in \eqref{sat implication}, but with the addition of the standard satisfaction clause for conjunction
\[
W_1\ni x\forces\varphi\wedge\psi\mbox{ iff }x\forces\varphi\mbox{ and }x\forces\psi.
\]

\begin{Theorem}
The minimal logic of implicative $\wedge$-semilattices, axiomatized as detailed above, is sound in the class $\mathbb{S} $ of frames.
\end{Theorem}
\begin{proof}
Given a frame $\mathfrak{F}\in\mathbb{S} $, define the sorted image operator $\ltright$, the closure $\Mtright$ of its restriction to Galois sets and the implication operation $\Ra$ as in the case of implicative posets (see \cref{posets section}).

It suffices to prove that $\Mtright$ distributes over arbitrary joins in the second argument place, $A\Mtright\bigvee_{j\in J}B_j=\bigvee_{j\in J}(A\Mtright B_j)$. This is because in that case, by definition of $\Ra$, 
\[
A\Ra\bigcap_{j\in J}C_j=\left(A\Mtright(\bigcap_{j\in J}C_j)'\right)'=(A\Mtright\bigvee_{j\in J}C'_j)'=(\bigvee_{j\in J}(A\Mtright C'_j))'=\bigcap_{j\in J}(A\Mtright C'_j)'=\bigcap_{j\in J}(A\Ra C_j).
\]

Distribution of $\Mtright$ over joins in the second argument place, assuming $xT'z[\;]$ is a Galois set, follows as an instance of  \mcite{Theorem~3.12}{duality2}.
\end{proof}

\begin{Theorem}
\label{completeness for semi}
The minimal logic of implicative $\wedge$-semilattices is complete in the class $\mathbb{S} $ of frames.
\end{Theorem}
\begin{proof}
Let $\mathfrak{F}=(s,W,I,T,\sigma)$ be the dual frame of the Lindenbaum-Tarski algebra $\mathbf{S}$ of the logic, regarded merely as a poset, as detailed in Definition~\ref{canonical frame defn}. In the carrier sorted set $W=(W_1,W_\partial)=(\filt(\mathbf{S}),\idl(\mathbf{S}))$ a poset filter is now necessarily a $\wedge$-semilattice filter (an upset $x\subseteq S$ such that $a,b\in x$ iff $a\wedge b\in x$). As argued for in \cite{hilary-sem}, $\gpsi$ with the canonical embedding map $\alpha(e)=\{x\in\filt(\mathbf{S})\midsp e\in x\}$ is a canonical extension of the  $\wedge$-semilattice $\mathbf{S}$. 

The canonical relation $T$ is defined as in the poset case and axioms (F1) and (F2) have been verified to hold in Proposition~\ref{F1 is canonical} and Proposition~\ref{U is canonical}. It remains to establish that the frame axiom (F3) also holds in the canonical frame. The proof argument is an adaptation to the meet semilattice case of the argument in the proof of \mcite{Lemma~4.6}{duality2}. It is useful to present a particular instance of that argument as it will also help in clarifying the rather bulky notation involved in the proof of \mcite{Lemma~4.6}{duality2}.

To prove that the section $uT'x[\;]$ of the Galois dual relation of the canonical relation $T^{\partial 1\partial}$ is a Galois (co-stable) set, define the set $W=\{b\in S\midsp\exists a\in S(a\in x \wedge (a\ra b)\in u)\}$ and let $w$ be the semilattice filter generated by $W$, i.e. $b\in w$ iff there exist $b_1,\ldots,b_s\in W$, for some $s$, such that $b_1\wedge\cdots\wedge b_s\leq b$.

Observe first that $w\in(uT'x[\;])'$. Indeed, let $v$ be any ideal such that $uT'xv$. By Proposition~\ref{F1 is canonical}, this is equivalent to $u\upv(x{\tright}v)$, where $x{\tright}v$ was defined in~\eqref{x-tright-v}. Let then $a\in x$ and $b\in v$ such that $(a\ra b)\in u$. This means that $b\in W\subseteq w$ and since also $b\in v$, we obtain $w\upv v$. The ideal $v$ was arbitrary such that $uT'xv$ and therefore it follows that $w\upv uT'x[\;]$, equivalently $w\in(uT'x[\;])'$, as needed.

Let now $q$ be an ideal $q\in(uT'x[\;])''$. We argue that $uT'xq$ holds, which proves the co-stability property we need, i.e. that $(uT'x[\;])''\subseteq uT'x[\;]$.

From $w\upv uT'x[\;]$ and $q\in(uT'x[\;])''$ we conclude $w\upv q$, so there exists a semilattice element $b\in w\cap q\neq \emptyset$. By definition of $w$ as the filter generated by $W$, let $b_1,\ldots,b_s\in W$ such that $b_1\wedge\cdots\wedge b_s\leq b$. 

Since for each $r=1,\ldots,s$ we have $b_r\in W$, let $c_r\in x$ be such that $(c_r\ra b_r)\in u$. 

Define $e=c_1\wedge\cdots\wedge c_s$. We obtain that, for each $r=1,\ldots,s$, $c_r\ra b_r\leq e\ra b_r\in u$, because $u$ is an upset and, because $u$ is a filter, we obtain $(e\ra b_1)\wedge\cdots\wedge(e\ra b_s)\in u$. Implication is assumed to distribute over meets in the consequent position, hence $e\ra b_1\wedge\cdots\wedge b_s\in u$ and since $b_1\wedge\cdots\wedge b_s\leq b$ we get $(e\ra b)\in u$. 

We obtained that for the elements $e\in x$ and $b\in q$, we have $(e\ra b)\in u$. This means that $u\upv x{\tright}q$. The latter is equivalent to $uT'xq$, i.e. $q\in uT'x[\;]$. Hence the inclusion $(uT'x[\;])''\subseteq uT'x[\;]$ holds and thereby the section $uT'x[\;]$ of the Galois dual relation of the canonical relation $T^{\partial 1\partial}$ is a Galois co-stable set.  Thus the canonical frame satisfies axiom (F3), as well, and so it belongs to the frame class $\mathbb{S} $.
\end{proof}

An extension of interest is obtained by adding Visser's axiom, $p\rfspoon(q\rfspoon p\wedge q)$. We trust the reader to calculate its correspondent, listed in \cref{corr table}.

\section{Implicative Lattices and their Logic}
\label{lattices section}
An implicative lattice $\mathbf{M}=(M,\leq,\wedge,\vee,0,1,\ra)$ is an implicative poset whose reduct $(M,\leq,0,1)$ is a lattice and where implication distributes over meets in the consequent position and it co-distributes over joins in the antecedent position $a\vee c\ra b=(a\ra b)\wedge(c\ra b)$. The minimal logic of implicative lattices, in a language including both conjunction and disjunction, is that of implicative meet semilattices extended with the join semilattice rules from Table~\ref{minimal proof system}, displayed below,

\begin{tabular}{lllllll}
\multicolumn{2}{l}{Join semilattice rules} &
  $\infrule{\vartheta\proves\varphi}{\vartheta\proves\varphi\vee\psi}$ & $\infrule{\vartheta\proves\psi}{\vartheta\proves\varphi\vee\psi}$ & $\infrule{\varphi\proves\vartheta\hskip5mm\psi\proves\vartheta}{\varphi\vee\psi\proves\vartheta}$
\end{tabular}

\vspace*{1mm}
\noindent
together with the co-distribution axiom from Table~\ref{minimal proof system} instantiated to implication
\[
\varphi\vee\psi\rfspoon\vartheta\proves(\varphi\rfspoon\vartheta)\wedge(\psi\rfspoon\vartheta).
\]
The language is interpreted in frames $\mathfrak{F}=(s,W,I,T,\sigma)$, as for implicative posets and meet semilattices, but with an addition of an axiom to validate the axiom of co-distribution over disjunction in the antecedent place. 

The frame class $\mathbb{L}$ consists of frames  $\mathfrak{F}=(s,W,I,T,\sigma)$ validating the axioms in Table~\ref{frame axioms for lattice-based logics}.

\begin{table}[!htbp]
\caption{Frame axioms for lattice-based implicative logics}
\label{frame axioms for lattice-based logics}
\hrule
\begin{enumerate}
\item[(F1)] For all $x\in W_1$ and all $v\in W_\partial$, the section $Txv\subseteq W_\partial$ of the relation $T^{\partial 1\partial}$ is a Galois set.
\item[(F2)] For all $x\in W_1$ and all $v\in W_\partial$, $x\upv v$ holds iff for all $u$ in $W_1$, $uT'xv$ holds.
\item[(F3a)] For all $x,z\in W_1$, the section $xT'z[\;]=\{y\in W_\partial\midsp xT'zy\}\subseteq W_\partial$ of the Galois dual $T'$ of the frame relation $T$ is a Galois (costable) set.
\item[(F3b)] For all $x\in W_1$ and $v\in W_\partial$, the section $xT'[\;]v=\{z\in W_1\midsp xT'zv\}\subseteq W_1$ of the Galois dual $T'$ of the frame relation $T$ is a Galois (stable) set.
\end{enumerate}
\hrule
\end{table}

Models $\mathfrak{M}=(\mathfrak{F},V)$ are defined as before and the interpretation of the language is as for the language of implicative meet semilattices, with implication interpreted by the clause in \eqref{sat implication}, but with the addition of the co-satisfaction clause for disjunction from Table~\ref{sat}, repeated below
\[
y\dforces\varphi\vee\psi \mbox{ iff } y\dforces\varphi\mbox{ and }y\dforces\psi.
\]

\begin{Theorem}
The minimal logic of implicative lattices, axiomatized as detailed above, is sound in the class $\mathbb{L}$ of frames.
\end{Theorem}
\begin{proof}
Soundness of the dual satisfaction rule for disjunction is immediate. For the distribution properties of implication, distribution over meets in the consequent place was discussed in the case of meet semilattice based logics.

It suffices to further verify that $\Mtright$ distributes over arbitrary joins in the first argument place, $(\bigvee_{j\in J}A_j)\Mtright C=\bigvee_{j\in J}(A_j\Mtright C)$. Given this and the definition of $\Ra$ from $\Mtright$ the needed conclusion follows. 

For the complete distribution property of $\Mtright$ over joins in the first argument place, given the assumption of stability of the section $uT'[\;]v$ of the Galois dual of the frame relation $T$, we refer the reader again to the general result proven in \mcite{Theorem~3.12}{duality2}.
\end{proof}

\begin{Theorem}
The minimal logic of implicative lattices is complete in the class $\mathbb{L}$ of frames.
\end{Theorem}
\begin{proof}
The dual (canonical) frame of the Lindenbaum-Tarski algebra of the logic is constructed as in the case of posets, see Definition~\ref{canonical frame defn}, with the observation that posets filters and ideals, as these were defined, are now immediately seen to be lattice filters and ideals. 

It remains to verify that the canonical frame is in the frame class $\mathbb{L}$. The first three axioms were verified in \cref{F1 is canonical}, \cref{U is canonical} and \cref{completeness for semi}. The verification of the frame axiom (F3b) is again an instance of \mcite{Lemma~4.6}{duality2}, as was that of the frame axiom (F3a). For the reader's benefit, we instantiate the proof to the case of implication and its co-distribution property over disjunctions. The argument is ``dual'' to that in the proof of \cref{completeness for semi}, switching from filters to ideals, from meets to joins etc. We show that for any filter $u$ and ideal $v$ the section $uT'[\;]v$ is a Galois stable set.

Set $W=\{a\in L\midsp\exists b\in L (b\in v\wedge (a\ra b)\in u)\}$ and let now $w$ be the ideal generated by $W$, i.e. $a\in w$ iff there exist $a_1,\ldots,a_s\in W$ such that $a\leq a_1\vee\cdots\vee a_s$.

Observe that $w\in(uT'[\;]v)'$. Indeed, let $z$ be a filter such that $uT'zv$, i.e. $u\upv (z{\tright}v)$. Let then $a\in z, b\in v$ such that $(a\ra b)\in u$. Then $a\in W\subseteq w$ and $a\in z$, hence $z\upv w$. Thus $uT'[\;]v\upv w$, which is to say that $w\in (uT'[\;]v)'$.

Let now $p$ be a filter such that $p\in(uT'[\;]v)''$, i.e. $p\upv (uT'[\;]v)'$.

We show that $uT'pv$ holds, which then implies that $(uT'[\;]v)''\subseteq uT'[\;]v$, i.e. that the section $uT'[\;]v$ is Galois stable.

Note that $p\upv w$, since $w\in (uT'[\;]v)'$. Let then $a\in p\cap w$ and $a_1,\ldots,a_s\in W$ such that $a\leq a_1\vee\cdots\vee a_s$.

Since $a_r\in W$, for each $r=1,\ldots,s$, let $c_r\in v$ be such that $a_r\ra c_r\in u$.

Let $e=c_1\vee\cdots\vee c_s$, so that $a_r\ra c_r\leq a_r\ra e\in u$, for each $r=1,\ldots,s$. Therefore, since $u$ is a filter, we have $\bigwedge_{r=1}^s(a_r\ra e)\in u$. By co-distribution, $\bigwedge_{r=1}^s(a_r\ra e)=(\bigvee_{r=1}^s a_r)\ra e\in u$ and since $a\leq \bigvee_{r=1}^s a_r$ and implication is antitone in the first argument place we obtain $a\ra e\in u$. Since $v$ is an ideal, $e=\bigvee_{r=1}^s c_r\in v$ and also $a\in p$. Therefore $u\upv p{\tright}v$, i.e. $uT'pv$ indeed holds. Thereby the section $uT'[\;]v$, for any filter $u$ and ideal $v$ is a Galois stable set, i.e. the frame axiom (F3a) holds in the canonical frame.
\end{proof}

\section{Relational Semantics for the Full Lambek Calculus and Substructural Logics}
\label{substructural section}
The full Lambek calculus and substructural logics, more generally, have been thoroughly studied from an algebraic perspective, consult for example \cite{galatos-ono-NFL,ono-galatos}, but relational semantics have not been systematically explored. A first study in this direction by the author was \cite{redm}, but results have been clarified and strengthened in the present article, drawing on the representation and duality results of \cite{duality2} and on the correspondence theory of \cite{dfmlC}. 

Combining the results of \cref{posets section}, \cref{lambek section} and \cref{lattices section}, we may define the frame classes corresponding to $\mathbf{NFL}$, the non-associative full Lambek calculus, $\mathbf{FL}$, the associative full Lambek calculus, as well as to the basic structural extensions $\mathbf{FL}_\mathrm{s}$, with $\mathrm{s}\subseteq\{\mathrm{e,c,w}\}$ (indices indicate exchange, contraction and weakening). This provides a classification of frames characterizing the corresponding substructural logics. 

Involutive substructural logics, such as Relevance logic (with, or without distribution) and Linear logic, include negation operators that satisfy the involution law and combining the results in this article and in \cite{choiceFreeStLog} characteristic frame classes for involutive substructural logics, or for logics with various weak negation operators can be specified. 

Subintuitionistic logics \cite{corsi-subint,restall-subint} are a special subclass of distributive substructural logics and we turn to distributive, intuitionistic and classical frames and logics in the next section.

\section{Distributive, Intuitionistic and Boolean Frames and Logics}
\label{classical section}
In \cite{choiceFreeHA} the cases of distributive and intuitionistic logic were discussed, in the context of representation and duality for distributive lattices and Heyting algebras, extending the duality results of \cite{duality2}, hence also extending the generalized J\'{o}nsson-Tarski framework of relational semantics. In particular, the following result characterizes distributive frames.

\begin{Proposition}[\mbox{\mcite{Proposition~3.13}{choiceFreeHA}}]
\label{upper bound rel prop}
Let $\mathfrak{F}=(s,W,I,(R_j)_{j\in J},\sigma)$ be a frame and $\gpsi$ the complete lattice of stable sets.   If all sections of the Galois dual relation $R'_\leq$ of the upper bound relation $R_\leq$ (where $uR_\leq xz$ iff both $x\leq u$ and $z\leq u$) are Galois sets, then $\gpsi$ is completely distributive.\telos
\end{Proposition}

A Heyting frame was characterized in \mcite{Proposition~3.15}{choiceFreeHA} as a frame in which the frame relation $R^{111}$ generating the residual $\bigodot$ of implication in the powerset algebra coincides with the upper bound relation, $R^{111}=R_\leq$. It was then shown in \mcite{Proposition~3.17}{choiceFreeHA} that intuitionistic implication is equivalently interpreted in the generalized J\'{o}nsson-Tarski framework by the standard clause $x\forces\varphi\rfspoon\psi$ iff $\forall z(z\forces\varphi\wedge x\leq z\lra z\forces\psi)$.

Given also the results of \cite{choiceFreeStLog}, applied to negation defined by $\neg\varphi=\varphi\rfspoon\bot$, frame axioms can be stated to characterize a frame class whose dual full complex algebra is a complete Boolean algebra, thus providing semantics for the classical propositional calculus, choice-free, but non-standard. 

If choice is admitted, we may resort to classical frames $\mathfrak{F}=(s,W,I,(R_j)_{j\in J},\sigma)$, in which we include axioms imposing that $W_1=W_\partial$ and $xIy$ iff $x=y$. As detailed in \cref{special case frames}, \cref{Galois map is complement}, \cref{all Galois}, \cref{discrete order} the Galois connection induced by the complement of the distinguished relation $I$ is set-complementation, the order is discrete and every subset is a Galois stable set, hence the dual powerset algebra and the full complex algebra of a frame coincide. Since $\yvval{\varphi}=\val{\varphi}\rperp=-\val{\varphi}$ in such a setting, it follows that the refutation relation $\dforces$ is just the complement $\not\forces$ of $\forces$. Consequently, the co-satisfaction clause for disjunction, $y\dforces \varphi\vee\psi$ iff $y\dforces\varphi$ and $y\dforces\psi$ yields the usual satisfaction clause for disjunction $y\forces\varphi\vee\psi$ iff $y\forces\varphi$ or $y\forces\psi$. 

Canonical frames can be constructed as sorted frames using both filters and ideals (as for any lattice). This is a choice-free construction, leading to non-standard semantics for the classical cases. Alternatively, allowing for choice-principles in the arguments, canonical frames can be constructed in the classical way, using ultrafilters for expansions of the classical propositional calculus, or prime filters for the merely distributive case.

\section*{Appendix A: Generalized Sahlqvist -- van Benthem Correspondence}
\label{sahlqvist section}
Appropriately restricting definitions from \mcite{Section~4.1}{dfmlC}, we say that a {\em positive occurrence} of a propositional variable in a sentence $\zeta$ in the language of sorted modal logic is one in the scope of an even number of applications of the priming operator. The variable {\em occurs positively} in $\zeta$ iff every one of its occurrences is positive. A sentence $\zeta$ is {\em positive} iff every propositional variable that occurs in $\zeta$ occurs positively in it.

\begin{Definition}[Simple Sahlqvist Sequents]
\label{simple sahlqvist}
A {\em simple Sahlqvist sequent} $\alpha\proves\eta$ of the first sort, or $\beta\vproves\delta$ of the second sort, is a sequent with positive consequent $\eta$, respectively $\delta$, and such that the premiss $\alpha$, respectively $\beta$, of the sequent is built from $\top, \bot$  by closing under conjunction and the additive operator $\tright$. 
\end{Definition}

Semantically, a sequent can be equivalently regarded as a formal inequality $\zeta\leq_\sharp \xi$, where if the sequent is $\zeta\proves \xi$, then $\sharp=1$ and if the sequent is $\zeta\vproves \xi$, then $\sharp=\partial$. 

Simple Sahlqvist inequalities are defined in the obvious way, given Definition~\ref{simple sahlqvist}. 

Pre-processing in the generalized Sahlqvist - Van Benthem correspondence algorithm presented in \cite{dfmlC} consists in manipulating (reducing) formal systems of inequalities, which are systems of the following form, where  $n,m\geq 0, \mbox{ and } \sharp,\sharp_i\in\{1,\partial\}$,
\begin{equation}\label{set of inequalities}
S=\langle  Q_1''\leq_{\sharp_1} Q_1,\ldots,Q_n''\leq_{\sharp_n} Q_n, Q_{n+1}=_{\sharp_{P'_1}}P'_1,\ldots,Q_{n+m}=_{\sharp_{P'_m}}P'_m \midsp \zeta\leq_\sharp \xi \rangle, 
\end{equation}
and where $\zeta\leq_\sharp \xi $ is its {\em main inequality} and $Q_i$ are propositional variables of sort determined by $\sharp_i$, for each $i=1,\ldots,n$ and $Q_{n+i}$ are of the same sort as $P_i'$, as indicated by the subscript to the equality symbol. 

Reduction aims at eliminating any occurrence of the priming operator on the left-side of the main inequality. We refer to formal inequalities of the form $Q''\leq_{\sharp Q}Q$ as {\em stability constraints} and to formal equations of the form $Q=_{\sharp_{P'}}P'$ as {\em change-of-variables constraints}. For brevity, we  write $\langle\mathrm{STB,CVC}\midsp\zeta\leq_\sharp\xi\rangle $, at times displaying some constraints of interest included in $\mathrm{STB}$ and/or in $\mathrm{CVC}$.

\begin{Definition}\label{equivalence of sets of inequalities}
For sets of formal inequalities $S_1,S_2$, define an equivalence relation by $S_1\sim S_2$ iff for any model $\mathfrak{M}=(\mathfrak{F},V)$ satisfying all the constraints, of the form $Q''\leq_{\sharp_Q} Q$, or $Q=_{\sharp_{P'}}P'$, in each of $S_1,S_2$, the model validates the main inequality of $S_1$ iff it validates the main inequality of $S_2$.
\end{Definition}

Table~\ref{reduction rules} presents a set of effectively executable reduction rules for sets of inequalities, proven to preserve equivalence in Lemma~\ref{R1-R6}. Note that  rule (R3) is a special instance of (R2), for $n=1$ and $\xi_1=\xi'$, but we include it because of its usefulness.
\begin{table}[!t]
\caption{Reduction Rules}
\label{reduction rules}
\hrule
\begin{enumerate}
\item[(R1)] $\infrule{\langle  \mathrm{STB,CVC}, P''\leq_{\sharp_P}P\midsp \zeta\leq_\sharp\xi\rangle }{\langle  \mathrm{STB,CVC}\midsp \zeta\leq_\sharp\xi\rangle }$,\\ provided the propositional variable $P$ does  not occur in $\zeta$ or $\xi$
    \\
\item[(R2)]  $\infrule{\langle  \mathrm{STB,CVC}\midsp \zeta''\leq_\sharp\xi_1'\cap\cdots\cap\xi'_n\rangle }{\langle  \mathrm{STB,CVC}\midsp \zeta\leq_\sharp\xi_1'\cap\cdots\cap\xi_n'\rangle }$, for $n\geq 1$
\\
\item[(R3)]  $\infrule{\langle  \mathrm{STB,CVC}\midsp \zeta''\leq_\sharp\xi''\rangle }{\langle  \mathrm{STB,CVC}\midsp \zeta\leq_\sharp\xi''\rangle  }$
\\
\item[(R4)] $\infrule{\langle  \mathrm{STB,CVC}\midsp \zeta\leq_\sharp\xi\rangle }{\langle  \mathrm{STB,CVC}, P''\leq_{\sharp_P}P\midsp \zeta[P/P'']\leq_\sharp\xi[P/P'']\rangle }$,\\
provided every occurrence of the propositional variable $P$ in each of $\zeta,\xi$ is double-primed and where
     $\sharp_P$ is the sort of $P$ and $\zeta[P/P''],\xi[P/P'']$ designate the results of uniformly replacing each occurrence of $P''$ by one of $P$ in each of $\zeta,\xi$
\\
\item[(R5)] If a re-write rule from the following $\mathrm{REWRITE}$ list is applicable, update the system of inequalities by carrying out the re-write
 \begin{tabbing}
(R5.1) \hskip2mm\= $P'''$\hskip3cm\= $\mapsto P'$\hskip3cm\=
\\[2mm]
(R5.2)  \> $(\eta\cap\zeta)''$\>$\mapsto   \eta''\cap\zeta''$
\\
\> $(\eta\cup\zeta)'$\>$\mapsto\eta'\cap\zeta'$
\\
(R5.3)\> $P\rspoon Q$\>$\mapsto (P\tright Q')'$  \> $P,Q$ constrained in $\mathrm{STB}$, $\mathrm{CVC}$
\\[2mm]
(R5.4) \> $P_2\rspoon(P_1\rspoon Q)$\>$\mapsto P_1\odot P_2\rspoon Q$ \> for any variables $P_1,P_2,Q$
\\[2mm]
(R5.5)\> $Q\lspoon P$ \> $\mapsto (Q'\tleft P)'$ \> $P,Q$ constrained in $\mathrm{STB}$, $\mathrm{CVC}$
\\[2mm]
(R5.6)\> $Q\lspoon P_2)\lspoon P_1$\> $\mapsto Q\lspoon P_1\odot P_2$ \> for any variables $P_1,P_2,Q$
\end{tabbing}
\item[(R6)] $\infrule{\langle  \mathrm{STB,CVC}\midsp\zeta\leq_\sharp\xi\rangle }{\langle  \mathrm{STB,CVC},Q=_{\sharp_{P'}}P'\midsp\zeta[Q/P']\leq_\sharp\xi[Q/P']\rangle  }$,\\
provided the variable $P$ occurs in the main inequality only single-primed and $Q$ is a fresh variable of the same sort as $P'$.
\item[(R7)] Apply residuation to rewrite the main inequality according to the related rewrite rule
    \begin{tabbing}
    $\alpha\leq_1\eta\rspoon\zeta$ \hskip1cm\=$\mapsto$\hskip4mm\= $\eta\odot\alpha\leq_1\zeta$
    \hskip3cm\= $\eta\leq_1\zeta\lspoon\alpha$ \hskip1cm\=$\mapsto$\hskip4mm\= $\eta\odot\alpha\leq_1\zeta$
    \end{tabbing}
\end{enumerate}
\hrule
\end{table}

\begin{Lemma}\label{R1-R6}
Executing any of the actions listed in Table~\ref{reduction rules} to a system $\mathcal{S}_1$ of inequalities leads to an equivalent system $\mathcal{S}_2$.  
\end{Lemma}
\begin{proof}
Consult \mcite{Lemma~4.3}{dfmlC}.
\end{proof}

\begin{Definition}
  \label{canonical Sahlqvist form}
A system $\langle \mathrm{STB,CVC}\midsp\zeta\leq_\sharp\xi\rangle  $ is in {\em canonical Sahlqvist form} if the main inequality $\zeta\leq_\sharp\xi$ is simple Sahlqvist and for any stability constraint $P''_1\leq_{\sharp_{P_1}}P_1$ in $\mathrm{STB}$ and change-of-variables constraint $P_2=_{\sharp_{Q'}}Q'$ in $\mathrm{CVC}$, $P_1$ and $P_2$ occur only unprimed in $\zeta,\xi$.
\end{Definition}
Note that in the right-hand-side of the inequality an unprimed variable $P$ may  be within the scope of $(\;)'$, if $P$ occurs as a subterm of a primed term.

\begin{Definition}\label{sahlqvist sequents and inequality systems}
A system of inequalities as in \eqref{set of inequalities} {\em is Sahlqvist} if it can be reduced to canonical Sahlqvist form, using the reduction rules of Table~\ref{reduction rules}. 

A 1-sequent $\zeta\proves\xi$ is Sahlqvist if  the associated inequality system $\langle \zeta\leq_1\xi\rangle  $ is Sahlqvist. Similarly for a $\partial$-sequent $\zeta\vproves\xi$.

A sequent $\varphi\proves\psi$ in the language of distribution-free modal logic with negation and implication is Sahlqvist iff either its translation $\varphi^\bullet\proves\psi^\bullet$, or its co-translation (dual translation) $\psi^\circ\vproves\varphi^\circ$ is Sahlqvist.
\end{Definition}

We outline the structure of the generalized Sahlqvist -- van Benthem algorithm and refer the reader to \cite{dfmlC} for further details.

\paragraph{Step 1 (Reduce to Canonical Sahlqvist Form).}
\noindent
\underline{Input:} A sequent $\varphi\proves\psi$ in the language of DfML.\\
Non-deterministically choose  to process
 either the translation \mbox{$\varphi^\bullet\le_1\psi^\bullet$}, or the co-translation (dual translation) $\psi^\circ\leq_\partial\varphi^\circ$ of the input. Run the reduction process. If neither of the (co)translation sequents reduces to a system of formal inequalities in canonical Sahlqvist form, then FAIL, else continue to step 2, with input either a system $\langle \mathrm{STB,CVC}\midsp\alpha\leq_1\eta\rangle $, or a system $\langle \mathrm{STB,CVC}\midsp\beta\leq_\partial\delta\rangle$, whichever was the output of this step.

\paragraph{Step 2 (Calculate $\mathrm{t}$-Invariance Constraints).}
\noindent
\underline{Input:} A system $\langle \mathrm{STB,CVC}\midsp\alpha\leq_1\eta\rangle $ (or $\langle \mathrm{STB,CVC}\midsp\beta\leq_\partial\delta\rangle$) in canonical Sahlqvist form, where $\mathrm{STB}=\{P''_i\leq_{\sharp_i}P_i\midsp i=1,\ldots,n\}$ and $\mathrm{CVC}=\{P_{n+j}=_{\sharp_{Q'_j}}Q'_j\midsp j=1,\ldots,k\}$.\\

For each constraint $P''\leq_1P$ a conjunct $\forall^\partial y[\forall^1 z (z\mathbf{I}y\lra\exists^\partial v(z\mathbf{I}v\wedge{P_i}(v)))\lra {P_i}(y)]$ is introduced in the antecedent of the second-order translation, ensuring that $P$ interprets to a Galois stable set. Similarly for each constraint $Q=_\partial P'$.

\paragraph{Step 3 (Generate the Guarded Second-Order Translation).}
\noindent
\underline{Input:} A guard $\mathrm{t{-}INV}=\bigwedge_{i=1}^{n+k}\forall^{\sharp_i} u_i[\mathrm{t}(P_i)(u_i)\lra P_i(u_i)]$, where for each $i=1,\ldots,n+k$, $\sharp_i\in\{1,\partial\}$ is the sort of $P_i$ and $\forall^{\sharp_i}\in\{\forall^1,\forall^\partial\}$, according to the value of $\sharp_i$.\\
\underline{Output:} The  {\em guarded second-order translation}, an expression of the form
\begin{equation}\label{guarded second order}
\forall^{\sharp_1} P_1\cdots\forall^{\sharp_{n+k}} P_{n+k}\forall^{\flat_1} Q^*_1\cdots\forall^{\flat_m} Q^*_m\forall^1 x(\mathrm{t{-}INV}\wedge\;\mathrm{ST}_x(\alpha)\lra\mathrm{ST}_x(\eta))
\end{equation}
or of the form 
\begin{equation}\label{guarded second order dual}
\forall^{\sharp_1} P_1\cdots\forall^{\sharp_{n+k}} P_{n+k}\forall^{\flat_1} Q^*_1\cdots\forall^{\flat_m} Q^*_m\forall^\partial y(\mathrm{t{-}INV}\wedge\;\mathrm{ST}_y(\beta)\lra\mathrm{ST}_y(\delta))
\end{equation}
depending on whether the translation, or the co-translation is being processed.

\paragraph{Step 4 (Pull-out Existential Quantifiers).}
This step is the same as in the classical case, using familiar equivalences to pull existential quantifiers in prenex position. It is detailed in the course of the proof of the correspondence result \mcite{Theorem~4.7}{dfmlC}.

\paragraph{Step 5 (Determine Minimal Instantiations).}
For the minimal instantiation in the classical case we set $\lambda(P)=\lambda s.(s=x)$. 
Since $\{x\}$ is not a stable set  $P$ is to be interpreted as a principal upper set, i.e. a closed element $\Gamma x$, hence we set $\lambda(P)=\lambda s.x\leq s$.

\paragraph{Step 6 (Eliminate Second-Order Quantifiers).}
The rationale is the same as in the classical Sahlqvist -- Van Benthem algorithm, substituting $\lambda(P)$ for $P$ and performing $\beta$-reduction. 

\begin{Theorem}\label{Sahlqvist thm}
  Every Sahlqvist sequent in the language of the logic of implicative posets, semilattices or lattices has a first-order local correspondent, effectively computable from the input sequent.
\end{Theorem}
\begin{proof}
We refer the reader to the proof of \mcite{Theorem~4.11}{dfmlC}.
\end{proof}

\bibliographystyle{plain}

\end{document}